\documentclass[11pt,letterpaper]{article}
\pdfoutput=1

\usepackage[utf8]{inputenc}
\usepackage[T1]{fontenc}
\usepackage[english]{babel}
\usepackage{fullpage}
\usepackage{amsmath,amssymb,amsthm}
\usepackage{mathrsfs}
\usepackage{thmtools}
\usepackage{thm-restate}
\usepackage{mathtools}
\usepackage{comment}
\usepackage{multirow}
\usepackage[textsize=scriptsize]{todonotes}
\usepackage{algpseudocode}
\usepackage{algorithm}
\PassOptionsToPackage{hyphens}{url}
\usepackage{etoolbox}
\usepackage{tikz}
\usepackage{pgfplots}
\usepackage[flushleft]{threeparttable}
\usepackage[hidelinks]{hyperref}
\usepgfplotslibrary{external}
\newtoggle{plotexperiments}
\toggletrue{plotexperiments}

\declaretheorem{theorem}
\declaretheorem[sibling=theorem]{definition}
\declaretheorem[sibling=theorem]{lemma}
\declaretheorem[sibling=theorem]{corollary}
\declaretheorem[sibling=theorem]{fact}
\declaretheorem{example}

\newcommand\numberthis{\addtocounter{equation}{1}\tag{\theequation}}

\DeclarePairedDelimiter{\norm}{\lVert}{\rVert}
\DeclarePairedDelimiter{\abs}{\lvert}{\rvert}

\newcommand{\specialcell}[2][c]{%
  \begin{tabular}[#1]{@{}c@{}}#2\end{tabular}}

\newcommand{\flatten}[1]{\ensuremath{\overline{#1}}}
\newcommand{\eps}{\ensuremath{\epsilon}}
\newcommand{\pdeg}{\ensuremath{d}}
\newcommand{\fhat}{\ensuremath{\widehat{f}}}

\newcommand{\hopt}{\ensuremath{h^*}}
\newcommand{\classP}{\ensuremath{\mathcal{P}}}

\newcommand{\OPT}{\ensuremath{\mathrm{OPT}}}
\newcommand{\Ak}{\ensuremath{{\mathcal{A}_k}}}
\newcommand{\A}{\ensuremath{\mathcal{A}}}
\newcommand{\setI}{\mathcal{I}}
\newcommand{\familyI}{\mathfrak{I}}

\newcommand{\err}{e}
\newcommand{\eqdef}{\stackrel{\rm {def}}{=}}

\newcommand{\EE}{\mathbb{E}}

\newcommand{\ed}{\stackrel{\text{def}}{=}}

\newcommand{\poly}{\mathrm{poly}}

\newcommand{\Otilde}{\ensuremath{\widetilde{O}}}

\newcommand{\setC}{\ensuremath{\mathcal{C}}}
\newcommand{\setP}{\ensuremath{\mathcal{P}}}
\newcommand{\R}{\ensuremath{\mathbb{R}}}

\newcommand{\oracle}{\ensuremath{\mathcal{O}}}
\newcommand{\som}{\ensuremath{\text{SOM}}}
\newcommand{\diff}{\ensuremath{\mathop{}\!\mathrm{d}}}
\newcommand{\disc}{\ensuremath{\mathrm{disc}}}

\algnewcommand{\LineComment}[1]{\State \(\triangleright\) #1}
\makeatletter
\let\OldStatex\Statex
\renewcommand{\Statex}[1][3]{%
\setlength\@tempdima{\algorithmicindent}%
\OldStatex\hskip\dimexpr#1\@tempdima\relax}

\newcommand{\besov}{B}
\newcommand{\Besov}[3]{\besov_{#1}^{#2}(L_{#3}([0,1]))}
\newcommand{\Besovapq}{\Besov{q}{\alpha}{p}}

\newcommand{\SO}{\mathcal{SO}}
\newcommand{\NSO}{\mathcal{NSO}}
\newcommand{\TO}{\mathcal{TOS}}
\newcommand{\NTO}{\mathcal{NTO}}

\title{Sample-Optimal Density Estimation in Nearly-Linear Time}
\author{Jayadev Acharya\thanks{Supported by a grant from the MIT-Shell Energy Initiative.}\\EECS, MIT\\\tt{jayadev@csail.mit.edu}
\and
Ilias Diakonikolas\thanks{Supported by a Marie Curie CIG, EPSRC grant EP/L021749/1 and a SICSA grant.}\\Informatics, U. of Edinburgh\\\tt{ilias.d@ed.ac.uk}
\and
Jerry Li\thanks{Supported by NSF grant CCF-1217921 and DOE grant DE-SC0008923.}\\EECS, MIT\\\tt{jerryzli@csail.mit.edu}
\and
Ludwig Schmidt\thanks{Supported by MADALGO and a grant from the MIT-Shell Energy Initiative.}\\EECS, MIT\\\tt{ludwigs@mit.edu}}

\begin{document}
\maketitle

\thispagestyle{empty}

\begin{abstract}
We design a new, fast algorithm for agnostically learning univariate probability distributions
whose densities are well approximated by piecewise polynomial functions. 
Let $f$ be the density function of an arbitrary univariate distribution, 
and suppose that $f$ is $\OPT$ close in $L_1$-distance to an unknown piecewise polynomial function
with $t$ interval pieces and degree $\pdeg$. Our algorithm 
draws $n = O(t(\pdeg+1)/\eps^2)$ samples from $f$,  runs in time  $\Otilde (n \cdot \poly (\pdeg))$, and with probability at least 
$9/10$ outputs an $O(t)$-piecewise degree-$\pdeg$ hypothesis $h$ that is 
$4 \cdot \OPT +\eps$ close to $f$.

Our general algorithm yields (nearly) sample-optimal and {\em nearly-linear time} estimators for  a wide range of structured distribution families
over both continuous and discrete domains in a unified way. For most of our applications, these are the {\em first} sample-optimal and nearly-linear time
estimators in the literature. As a consequence, our work resolves the sample and computational complexities of a broad class of inference
tasks via a single ``meta-algorithm''. Moreover, we experimentally demonstrate that our algorithm performs very well in practice.

Our algorithm consists of three ``levels'':
(i) At the top level, we employ an iterative greedy algorithm for finding a good partition of the real line into the pieces of a piecewise polynomial.
(ii) For each piece, we show that the sub-problem of finding a good polynomial fit on the current interval can be solved efficiently with a separation oracle method.
(iii) We reduce the task of finding a separating hyperplane to a combinatorial problem and give an efficient algorithm for this problem.
Combining these three procedures gives a density estimation algorithm with the claimed guarantees.
\end{abstract}
\newpage

\thispagestyle{empty}

\setcounter{tocdepth}{2}

\tableofcontents

\thispagestyle{empty}

\newpage

\setcounter{page}{1}


\section{Introduction}  \label{sec:intro}
Estimating  an unknown probability density function based on observed data  is a classical problem
in statistics that has been studied since the late nineteenth century, starting with the pioneering
work of Karl Pearson~\cite{Pearson}. Distribution estimation has become a paradigmatic and fundamental unsupervised learning problem
with a rich history and extensive literature (see e.g.,~\cite{BBBB:72, DG85, Silverman:86,Scott:92,DL:01}).
A number of general methods for estimating distributions have been proposed in the mathematical statistics literature, 
including histograms, kernels, nearest neighbor estimators, orthogonal series estimators, maximum
likelihood, and more. We refer the reader to \cite{Izen91} for a survey of these techniques.
During the past few decades, there has been a large body of work on this topic in computer science with a focus on {\em computational efficiency}~\cite{KMR+:94short,
FreundMansour:99short,FOS:05focs,BelkinSinha:10, KMV:10,
MoitraValiant:10,KSV08,VempalaWang:02,DDS12soda,DDS12stoc, DDOST13focs, CDSS14}.

Suppose that we are given a number of samples from an unknown distribution that belongs to (or is well-approximated by) a given family of distributions ${\cal C}$, 
e.g., it is a mixture of a small number of Gaussians.
Our goal is to estimate the unknown distribution in a precise, well-defined way. 
In this work, we focus on the problem of {\em density estimation} (non-proper learning), where the goal is to output  an approximation of the unknown density 
without any constraints on its representation. That is, the output hypothesis is not necessarily a member of the family ${\cal C}$. 
The ``gold standard'' in this setting is to design learning algorithms that are both {\em statistically} and {\em computationally} efficient.
More specifically, the ultimate goal is to obtain estimators whose sample size is  information--theoretically optimal, and whose running time is
{\em (nearly) linear} in their sample size. An important additional requirement is that our learning algorithms are
{\em agnostic} or  robust under model misspecification, i.e., they succeed even if the target distribution does not belong to the given family ${\cal C}$
but is merely well-approximated by a distribution in ${\cal C}$.

We study the problem of density estimation for univariate distributions, i.e., distributions with a density  
$f: \Omega \to \R_+$, where the sample space $\Omega$ is a subset of the real line.
While density estimation for families of univariate distributions has been studied for several decades, both the sample and time complexity were not yet well understood before this work, even for surprisingly simple classes of distributions, such as mixtures of Binomials and mixtures of Gaussians.
Our main result is a general learning algorithm that can be used to estimate a wide variety of structured distribution families.
For each such family, our general algorithm simultaneously satisfies all three of the aforementioned criteria, 
i.e., it is agnostic, (nearly) sample optimal, and runs in {\em nearly-linear time.}  

Our algorithm is based on learning a piecewise polynomial function that approximates the unknown density.
The approach of using piecewise polynomial approximation has been employed in this context before ---
our main contribution is to improve the computational complexity of this method and to make 
it nearly-optimal for a wide range of distribution families. The key idea of using piecewise polynomials for learning 
is that the {\em existence} of good piecewise polynomial approximations for a family ${\cal C}$ of distributions 
can be leveraged for the design of efficient learning algorithms for the family ${\cal C}$.
The main algorithmic ingredient that makes this method possible is an efficient procedure for agnostically learning piecewise polynomial density functions.
In prior work, Chan, Diakonikolas, Servedio, and Sun~\cite{CDSS14} obtained a nearly-sample optimal and 
polynomial time algorithm for this learning problem. Unfortunately, however, 
the polynomial exponent in their running time is quite high, which makes their algorithm prohibitively slow for most applications.

In this paper, we design a new, fast algorithm for agnostically learning piecewise polynomial distributions, 
which in turn yields sample-optimal and {\em nearly-linear time} estimators for  a wide range of structured distribution families
over both continuous and discrete domains. For most of our applications, these are the {\em first} sample-optimal and nearly-linear time
estimators in the literature. As a consequence, our work resolves the sample and computational complexity of a broad class of inference
tasks via a single ``meta-algorithm''. Moreover, we experimentally demonstrate that our algorithm performs very well in practice.
We stress that a significant number of new algorithmic and technical ideas are needed for our main result, as we explain next.

\subsection{Our results and techniques} \label{sec:results}

In this section, we describe our results in detail, compare them to prior work, and give an overview of our new algorithmic ideas.

\medskip

\noindent {\bf Preliminaries.}
We consider univariate probability density functions (pdf's) defined over a known finite interval $I \subseteq \R$.
(We remark that this assumption is without loss of generality and our results easily apply to densities defined over the entire real line.)

We focus on a standard notion of learning an unknown probability distribution from samples~\cite{KMR+:94short}, 
which is a natural analogue of Valiant's well-known PAC model for learning Boolean functions~\cite{val84} 
to the unsupervised setting of learning an unknown probability distribution. (We remark that our definition
is essentially equivalent to the notion of the $L_1$ minimax rate of convergence in statistics~\cite{DL:01}.)
A distribution learning problem is defined by a class ${\cal C}$ of probability distributions over a domain $\Omega$.
Given $\eps>0$ and sample access to an unknown distribution with density $f$, the goal of an 
\emph{agnostic learning algorithm for ${\cal C}$} is to compute a hypothesis $h$ such that,  
with probability at least $9/10$, it holds $\| h - f\|_1 \leq C \cdot \OPT_{{\cal C}}(f)+\eps,$
where  $\OPT_{{\cal C}}(f) := \inf_{q \in {\cal C}} \|q- f\|_1,$
i.e.,  $\OPT_{{\cal C}}(f)$ is the $L_1$-distance between the unknown density $f$ and the closest
distribution to it in ${\cal C}$, and $C \ge 1$ is a universal constant.

We say that a function $f$ over an interval $I$ is a \emph{$t$-piecewise degree-$\pdeg$ polynomial} 
if there is a partition of {$I$} into $t$ disjoint intervals $I_1,\dots,I_t$ such that
$f(x)=f_j(x)$ for all $x \in I_j$, where each of $f_1,
\dots,f_t$ is a polynomial of degree at most $\pdeg$.
Let $\mathcal{P}_{t, \pdeg}(I)$ denote the class of all $t$-piecewise
degree-$\pdeg$ polynomials over the interval $I$.  

\medskip

\noindent {\bf Our Results.} 
Our main algorithmic result  is the following:

\begin{theorem}[Main] \label{thm:main-intro}
Let $f: I \to \R_+$ be the density of an unknown distribution over $I$, where $I$ is either an interval or the discrete set $[N]$.  
There is an algorithm with the following performance guarantee:
Given parameters $t, \pdeg \in \mathbb{Z}_+$, an error tolerance $\eps>0$, and any
$\gamma>0$, the algorithm draws $n = O_{\gamma}(t(\pdeg+1)/\eps^2)$ samples from the unknown distribution, 
runs in time  $\Otilde (n \cdot \poly (\pdeg+1))$, and with probability at least 
$9/10$ outputs an $O(t)$-piecewise degree-$\pdeg$ hypothesis $h$ 
such that $\| f - h \|_1 \le (3+\gamma) \OPT_{t, \pdeg}(f) +\eps$, where
$\OPT_{t, \pdeg}(f) := \inf_{r \in \mathcal{P}_{t,\pdeg}(I)} \|f-r\|_1$
is the error of the best $t$-piecewise degree-$\pdeg$ approximation to $f$.
\end{theorem}

In prior work, \cite{CDSS14} gave a learning algorithm for this problem that uses $\Otilde(t(\pdeg+1)/\eps^2)$ samples
and runs in $\poly(t, \pdeg+1, 1/\eps)$ time. We stress that the algorithm of~\cite{CDSS14}  is prohibitively slow. In particular, 
the running time of their approach is $\widetilde{\Omega}(t^3 \cdot (\pdeg^{3.5} / \eps^{3.5}+ \pdeg^{6.5}/\eps^{2.5}) )$, which renders
their result more of  a ``proof of principle'' than a computationally efficient algorithm. 

This prompts the following question:  {\em Is such a high running time necessary to achieve this level of sample efficiency?}
Ideally, one would like a sample-optimal algorithm with a low-order polynomial running time (ideally, linear).

Our main result shows that this is indeed possible in a very strong sense. 
The running time of our algorithm is \emph{linear} in $t/\eps^2$ (up to a $\log(1/\eps)$ factor), which is essentially the best possible;
the polynomial dependence on $\pdeg$ is $\Otilde (\pdeg^{3+\omega})$, where $\omega$ is the matrix multiplication exponent.
This substantially improved running time is of critical importance 
for the applications of Theorem~\ref{thm:main-intro}.
Moreover, the sample complexity of our algorithm removes the extraneous logarithmic factors present in the sample complexity of \cite{CDSS14}
and matches the information-theoretic lower bound up to a constant factor.
As we explain below, Theorem~\ref{thm:main-intro} leads to 
(nearly) sample-optimal and {\em nearly-linear time} estimators for a wide range of natural and well-studied families. For most of
these applications, ours is the {\em first} estimator with simultaneously nearly optimal sample and time complexity.

Our new algorithm is clean and modular. As a result, 
Theorem~\ref{thm:main-intro} also applies to discrete distributions over an ordered domain (e.g., $[N]$). 
The approach of~\cite{CDSS14} does not extend to polynomial approximation over discrete domains, and designing such an algorithm was left as an open problem in their work. As a consequence, we obtain several {\em new} applications to learning mixtures of discrete distributions. In particular, we obtain the first nearly sample optimal and nearly-linear time estimators for mixtures of Binomial and Poisson distributions.
To the best of our knowledge, no polynomial time algorithm with nearly optimal sample complexity was known for these basic learning problems prior to this work.

\paragraph{Applications.}
We now explain how to use Theorem~\ref{thm:main-intro} in order to agnostically learn structured distribution families.
Given a class ${\cal C}$ that we want to learn, we proceed as follows:
(i) Prove that any distribution in ${\cal C}$ is $\eps/2$-close in $L_1$-distance
to a $t$-piecewise degree-$\pdeg$ polynomial, for appropriate values of $t$ and $\pdeg$.
(ii) Use Theorem~\ref{thm:main-intro} for these values of $t$ and $\pdeg$ to 
agnostically learn the target distribution up to error $\eps/2$.
Note that $t$ and $\pdeg$ will depend on the desired error $\eps$ and the underlying class $\mathcal{C}$.
We emphasize that there are many combinations of $t$ and $\pdeg$ that guarantee an $\eps/2$-approximation of ${\cal C}$ in Step (i).
To minimize the sample complexity of our learning algorithm in Step (ii), we would like to determine the values of $t$ and $\pdeg$ that minimize the product $t(\pdeg+1)$.
This is, of course, an approximation theory problem that depends on the structure of the family ${\cal C}$.

For example, if ${\cal C}$ is the family of log-concave distributions, the optimal $t$-histogram approximation with accuracy $\eps$
requires $\Theta(1/\eps)$ intervals. This leads to an algorithm with sample complexity $\Theta(1/\eps^3)$.
On the other hand, it can be shown that any log-concave distribution has a piecewise {\em linear} $\eps$-approximation
with $\Theta(1/\eps^{1/2})$ intervals~\cite{CDSS14, DiakonikolasK15}, which yields an algorithm with sample complexity $\Theta(1/\eps^{5/2})$.
Perhaps surprisingly, this sample bound cannot be improved using higher degree piecewise polynomials;
one can show an information-theoretic lower bound of $\Omega(1/\eps^{5/2})$ for learning log-concave densities~\cite{DL:01}.
Hence, Theorem~\ref{thm:main-intro} gives a sample-optimal and nearly-linear time agnostic learning algorithm for this fundamental problem.
We remark that piecewise polynomial approximations are ``closed'' under taking mixtures. As a corollary, 
Theorem~\ref{thm:main-intro} also yields an $O(k/\eps^{5/2})$ sample and nearly-linear time algorithm for learning an arbitrary mixture
of $k$ log-concave distributions. Again, there exists a matching information-theoretic lower bound of $\Omega(k/\eps^{5/2})$.

As a second example, let ${\cal C}$ be the class of mixtures of $k$ Gaussians in one dimension. 
It is not difficult to show that learning such a mixture of Gaussians up to $L_1$-distance $\eps$ requires $\Omega(k/\eps^2)$ samples.
By approximating the corresponding probability density functions 
with piecewise polynomials of degree $O(\log(1/\eps))$, we obtain an agnostic learning algorithm for this class 
that uses $n = \Otilde(k/\epsilon^2)$ samples and runs in time  $\Otilde(n)$. Similar bounds
can be obtained for several other natural parametric mixture families.

Note that for a wide range of structured families,\footnote{This includes all structured families considered in \cite{CDSS14} and several previously-studied distributions not covered by \cite{CDSS14}.} the optimal choice of the degree $\pdeg$ (i.e., the choice minimizing $t(\pdeg+1)$ among all $\eps/2$-approximations) 
will be {\em at most poly-logarithmic in $1/\eps$}.
For several classes (such as unimodal, monotone hazard rate, and log-concave distributions), the degree $\pdeg$ is even a constant.
As a consequence, Theorem~\ref{thm:main-intro} yields (nearly) sample optimal and nearly-linear time estimators for all these families in a unified way.
In particular, we obtain sample optimal (or nearly sample optimal) and nearly-linear time estimators 
for a wide range of structured distribution families, including arbitrary mixtures of natural distributions such as multi-modal, concave, convex, log-concave, monotone hazard rate, Gaussian, Poisson, Binomial, functions in Besov spaces, and others.

See Table~\ref{tab:long} for a summary of these applications. 
For each distribution family in the table, we provide a comparison to the best previous result.
Note that we do not aim to exhaustively cover all possible applications of Theorem~\ref{thm:main-intro}, 
but rather to give some selected applications that are indicative of the generality and power of our method.

\renewcommand{\arraystretch}{1.3}
\begin{table}[htb!]
\begin{threeparttable}
\begin{center}
\begin{tabular}{|c|c|c|c|c|}
\hline
\textbf{Class of distributions} & \specialcell{\textbf{Sample}\\[-.25cm]\textbf{complexity}} & \specialcell{\textbf{Time}\\[-.25cm]\textbf{complexity}} & \textbf{Reference} & \textbf{Optimality} \\
\hline
$t$-histograms& $\Otilde(\frac{t}{\eps^2})$ & $\Otilde(\frac{t}{\eps^2})$ &
\cite{CDSS14b}&\\
& $O(\frac{t}{\eps^2})$ & $O(\frac{t}{\eps^2}\log(1/\eps))$ &
Theorem~\ref{thm:mergingmain}&$\SO$, $\TO$ \\
\hline
\multirow{2}{*}[.3em]{\specialcell[t]{$t$-piecewise\\[-.25cm]degree-\pdeg{} polynomials}} &
$\Otilde(\frac{t\cdot \pdeg}{\eps^2})$
&$\Otilde \left(t^3 \cdot (\frac{\pdeg^{3.5}}{\eps^{3.5}}+ \frac{\pdeg^{6.5}}{\eps^{2.5}}) \right)$ & 
\cite{CDSS14}& \\
& $O(\frac{t \cdot\pdeg}{\eps^2})$ & $\Otilde(\frac{t \cdot\pdeg^{\omega+3}}{\eps^2})$ &
Theorem~\ref{thm:main-intro}& $\NSO$\\
\hline
$k$-mixture of log-concave & $\Otilde(\frac{k}{\eps^{5/2}}) $
& $\Otilde(\frac{k^3}{\eps^{5}})$ &
\cite{CDSS14}&\\
& $O(\frac{k}{\eps^{5/2}})$ & $\Otilde(\frac{k}{\eps^{5/2}})$ &
Theorem~\ref{thm:log-concave}&$\SO$, $\NTO$\\
\hline
$k$-mixture of Gaussians &
$\Otilde(\frac{k}{\eps^2})$
& $\Otilde(\frac{k^3}{\eps^{3.5}})$ &
\cite{CDSS14}&\\
& $O(\frac{k\log(1/\eps)}{\eps^2})$ & $\Otilde(\frac{k }{\eps^2})$ &
Theorem~\ref{thm:mog}&$\NSO$, $\NTO$\\
\hline
Besov space $\Besovapq$ &
$O_{\alpha}\left(\frac{\log^2(1/\eps)}{\eps^{2+1/\alpha}}\right)$ & $\Otilde_{\alpha}\left(\frac1{\eps^{6+3/\alpha}}\right)$ & 
\cite{WillettN07}&\\
&
$O_{\alpha}\left(\frac{1}{\eps^{2+1/\alpha}}\right)$ & $\Otilde_{\alpha}\left(\frac{1}{\eps^{2+1/\alpha}}\right)$ & 
Theorem~\ref{thm:besov}&$\SO$, $\NTO$\\
\hline
$k$-mixture of $t$-monotone  &
$\Otilde (\frac{t\cdot k}{\eps^{2+1/t}})$ & $\Otilde(\frac{k^3}{\eps^{3/t}}\cdot
(\frac{t^{3.5}}{\eps^{3.5}}+\frac{t^{6.5}}{\eps^{2.5}}))$&  
\cite{CDSS14}&\\
&
$O (\frac{t\cdot k}{\eps^{2+1/t}})$ & $\Otilde (\frac{k\cdot t^{2+\omega}}{\eps^{2+1/t}})$ & 
Theorem~\ref{thm:kmonotone}&\specialcell{$\SO, \NTO$\\[-.25cm]for $t=1,2$}\\
\hline
\multirow{2}{*}[.3em]{\specialcell[t]{$k$-mixture of $t$-modal}} &
$ \Otilde(\frac{t\cdot k\log(N)}{\eps^3})$
& $ \Otilde(\frac{t\cdot k\log(N)}{\eps^3})$ &
\cite{CDSS14b}&\\
&
$O(\frac{t\cdot k\log(N)}{\eps^3})$
& $O(\frac{t\cdot k\log(N)}{\eps^3}\log(1/\eps))$ &
Theorem~\ref{thm:tmodal}&$\SO$, $\TO$\\
\hline
$k$-mixture of MHR &
$ \Otilde(\frac{k\log(N/\eps)}{\eps^3}))$
& $ \Otilde(\frac{k\log(N/\eps)}{\eps^3}))$ &
\cite{CDSS14b}&\\
&
$O(\frac{k\log(N/\eps)}{\eps^3})$
& $O(\frac{k\log(N/\eps)}{\eps^3}\log(1/\eps))$ &
Theorem~\ref{thm:mhr}&$\SO$, $\TO$\\
\hline
\specialcell{$k$-mixture of\\[-.25cm]Binomial, Poisson} &
$\Otilde(\frac{k}{\eps^3})$
& $ \Otilde(\frac{k}{\eps^3})$ &
\cite{CDSS14b}&\\
&
$ O(\frac{k\log(1/\eps)}{\eps^2})$ &
$\Otilde(\frac{k}{\eps^2})$ &
Theorem~\ref{thm:poi}&$\NSO$, $\NTO$\\
\hline
\end{tabular}
 \begin{tablenotes}
 \item 
$\SO:$ Sample complexity is optimal up to a constant factor.
\item
$\NSO:$ Sample complexity is optimal up to a poly-logarithmic factor.
\item
$\TO:$ Time complexity is optimal (up to sorting the samples).
\item
$\NTO:$ Time complexity is optimal up to a poly-logarithmic factor.
\end{tablenotes}
\end{center}
\end{threeparttable}
\caption{A list of applications to agnostically learning specific
  families of distributions. For each class, the
  first row is the best known previous result and the second row is
  our result. Note that for most of the examples, our algorithm
  runs in time that is nearly-linear in the information-theoretically
  optimal sample complexity. The last three classes are over discrete
  sets, and $N$ denotes the size of the support. 
}
\label{tab:long}
\end{table}

Moreover, our non-proper learning algorithm is also useful for proper learning.
Indeed, Theorem~\ref{thm:main-intro} has recently been used~\cite{LS15} as a crucial component to obtain the fastest known agnostic algorithm for properly learning a mixture of univariate Gaussian distributions.
Note that non-proper learning and proper learning for a family ${\cal C}$ are equivalent in terms of sample complexity:
given any (non-proper) hypothesis, we can perform a brute-force search to find its closest approximation in the class ${\cal C}$.
The challenging part is to perform this computation efficiently. 
Roughly speaking, given a piecewise polynomial hypothesis, \cite{LS15} design an efficient algorithm to find the closest mixture of $k$ Gaussians.

\medskip

\noindent {\bf Our Techniques.}
We now provide a brief overview of our new algorithm and techniques in parallel with a comparison to the previous algorithm of ~\cite{CDSS14}.
 We require the following definition.
 For any $k \geq 1$ and an interval $I \subseteq \R$, define the $\Ak$-norm of a function $g: I \to \R$ to be
\[  \| g \|_\Ak \ed \sup_{I_1, \ldots, I_k} \sum_{i = 1}^k |g(I_i)| \; ,\]
where the supremum is over all sets of $k$ disjoint intervals $I_1, \ldots, I_k$ in $I$, and $g (J) \ed \int_J g(x) \diff x$ for any measurable set 
$J \subseteq I$.
Our main probabilistic tool is the following well-known version of the VC inequality:
\begin{theorem}[VC Inequality~\cite{VapnikChervonenkis:71, DL:01}]
\label{thm:vc}
Let $f: I \to \R_+$ be an arbitrary pdf over $I$, and let $\fhat$ be the empirical pdf obtained after taking $n$ i.i.d.\ samples from $f$. 
Then 
\[\EE [\| f - \fhat \|_{\Ak} ] \leq O\left(\sqrt{\frac{k}{n}} \right).\]
\end{theorem}
 
Given this theorem, it is not difficult to show that the following two-step procedure
is an agnostic learning algorithm for $ \mathcal{P}_{t,\pdeg}$:
\begin{itemize}
\item[(1)] Draw a set of $n = \Theta(t (\pdeg+1)/\eps^2)$ samples from $f$;
\item[(2)] Output the piecewise-polynomial hypothesis $h \in {\cal P}_{t, \pdeg}$ that minimizes the quantity 
$\|h - \fhat \|_{\mathcal{A}_k}$ up to an additive error of $O(\eps)$, where $k = O(t (\pdeg + 1))$.
\end{itemize}
We remark that the optimization problem in Step (2) is non-convex. 
However, it has sufficient structure so that it can be solved in polynomial time.
Intuitively, an algorithm for Step (2) involves two main ingredients:
\begin{itemize}
\item[(2.1)] An efficient procedure to find a good set $t$ intervals.
\item[(2.2)] An efficient procedure to agnostically learn a degree-$\pdeg$ polynomial in a given sub-interval of the domain.
\end{itemize}
We remark that the procedure for (2.1) will use the procedure for (2.2) multiple times as a subroutine.

\cite{CDSS14} solve an appropriately relaxed version of Step (2) by a combination
of linear programming and dynamic programming. Roughly speaking, they formulate a polynomial size linear program 
to agnostically learn a degree-$\pdeg$ polynomial in a given interval, and use a dynamic program 
in order to discover the correct $t$ intervals. It should be emphasized that the algorithm of~\cite{CDSS14} 
is theoretically efficient (polynomial time), but prohibitively slow for real applications with large data sets.
In particular,  the linear program of \cite{CDSS14} has $\Omega(\pdeg/\eps)$ variables and 
$\Omega(\pdeg^2/\eps^2 + \pdeg^5 / \eps)$ constraints.
Hence, the running time of their algorithm using the fastest known LP solver for their instance \cite{LS14} is 
at least $\widetilde{\Omega}(\pdeg^{3.5} / \eps^{3.5}+ \pdeg^{6.5}/\eps^{2.5})$. 
Moreover, the dynamic program to implement (2) 
has running time at least $\Omega(t^3)$.
This leads to an overall running time of $\widetilde{\Omega} \left(t^3 \cdot (\pdeg^{3.5} / \eps^{3.5}+ \pdeg^{6.5}/\eps^{2.5}) \right)$, 
which quickly becomes unrealistic even for modest values of $\eps, t$, and $d$.

We now provide a sketch of our new algorithm. At a high-level, we implement procedure (2.1) above using 
an {\em iterative greedy} algorithm. Our algorithm circumvents the need for a dynamic programming approach as follows:
The main idea is to iteratively merge pairs of intervals by calling an oracle for procedure (2.2) in every step until the number of intervals
becomes $O(t)$. Our iterative algorithm and its subtle analysis are directly inspired by the VC inequality.
Roughly speaking, in each iteration the algorithm estimates the contribution to an appropriate notion of error 
when two consecutive intervals are merged, and it merges pairs of intervals with small error. 
This procedure ensures that the number of intervals in our partition decreases geometrically with the number of iterations. 

Our algorithm for procedure (2.2) is based on convex programming and runs in time $\poly(\pdeg+1)/\eps^2$ (note that the dependence on $\eps$ is optimal).
To achieve this significant running time improvement, we use a novel approach that is quite different from that of~\cite{CDSS14}.  
Roughly speaking, we are able to exploit the problem structure inherent in the $\mathcal{A}_{k}$ optimization problem in order to 
separate the problem dimension $\pdeg$ from the problem dimension $1/\eps$, 
and only solve a convex program of dimension $\pdeg$ (as opposed to dimension $\poly(\pdeg/\eps)$ in \cite{CDSS14}).
More specifically, we consider the convex set of non-negative polynomials with $\mathcal{A}_{\pdeg+1}$ 
distance at most $\tau$ from the empirical distribution. While this set is defined through a large number of constraints, 
we show that it is possible to design a separation oracle that runs in time {\em nearly linear} in both the number of samples and the degree $\pdeg$.
Combined with tools from convex optimization such as the Ellipsoid method or Vaidya's algorithm, this gives an efficient algorithm for procedure (2.2).

\subsection{Related work} \label{ssec:lit}

There is a long history of research in statistics on estimating structured families of distributions.
For distributions over continuous domains, 
a very natural type of structure to consider is some sort of  ``shape constraint'' on the probability density function (pdf) defining the distribution.
Statistical research in this area started in the 1950's, and the reader is referred to the book~\cite{BBBB:72}
for a summary of the early work. Most of the literature in shape-constrained density estimation 
has focused on one-dimensional distributions, with a few exceptions during the past decade.
Various structural restrictions have been studied over the years, starting from
monotonicity, unimodality, convexity, and concavity~\cite{Grenander:56, Brunk:58, PrakasaRao:69, Wegman:70, HansonP:76, Groeneboom:85, Birge:87, Birge:87b,
Fougeres:97,ChanTong:04,JW:09},
and more recently focusing on structural restrictions such as log-concavity and $k$-monotonicity
\cite{BW07aos, DumbgenRufibach:09, BRW:09aos, GW09sc, BW10sn, KoenkerM:10aos, Walther09, DossW13, ChenSam13, KimSam14, BalDoss14, HW15}.
The reader is referred to~\cite{GJ:14} for a recent book on the subject.
Mixtures of structured distributions have received much attention
in statistics \cite{Lindsay:95,RednerWalker:84,TSM:85,Li99mixturedensity}
and, more recently, in theoretical computer science~\cite{Dasgupta:99, DasguptaSchulman:00, AroraKannan:01, VempalaWang:02, FOS:05focs, AchlioptasMcSherry:05, MoitraValiant:10}.

The most common method used in statistics to address shape-constrained inference problems is the Maximum Likelihood Estimator (MLE) and its variants.
While the MLE is very popular and quite natural, we note that it is not agnostic, and it may in general require solving an intractable optimization problem (e.g., for the case of mixture models.)

Piecewise polynomials (splines) have been extensively used  as tools for inference tasks, including density estimation, 
see, e.g.,~\cite{WegW83, WillettN07, Stone94, Stone97}. We remark that splines in the statistics literature have been used in the context of the MLE, 
which is very different than our approach. Moreover, the degree of the splines used is typically bounded
by a small constant and the underlying algorithms are heuristic in most cases.
A related line of work in mathematical statistics~\cite{KerkPic92, Don95, KPicT96, Donoho96, Donoho98} uses non-linear estimators based on 
wavelet techniques to learn continuous distributions whose densities satisfy various smoothness constraints, such as Triebel and Besov-type smoothness. 
We remark that the focus of these works is usually on the statistical efficiency of the proposed estimators. 

For the problem of learning piecewise constant distributions with $t$ unknown interval pieces, 
\cite{CDSS14b} recently gave
an $n = \Otilde(t/\eps^2)$ sample and $\Otilde(n)$ time algorithm. However, their approach
does not seem to generalize to higher degrees. 
Moreover, recall that Theorem~\ref{thm:main-intro} removes all logarithmic factors from the sample complexity.
Furthermore, our algorithm runs in time proportional to the time required to sort the samples, while their algorithm has additional logarithmic factors in the running time (see Table \ref{tab:long}).

Our iterative merging idea is quite robust: together with Hegde, the authors of the current paper have shown that an analogous approach yields sample optimal and efficient algorithms for agnostically learning discrete distributions with piecewise constant functions under the $\ell_2$-distance metric \cite{ADHLS15}.
We emphasize that learning under the $\ell_2$-distance is easier than under the $L_1$-distance, and that the analysis of \cite{ADHLS15} is significantly simpler than the analysis in the current paper.
Moreover, the algorithmic subroutine of finding a good polynomial fit on a fixed interval required by the $\ell_2$ algorithm is substantially simpler than the subroutine we require here.
Indeed, in our case, the associated optimization problem has exponentially many linear constraints, and thus cannot be fully described, even in polynomial time.

\medskip

\noindent {\bf Paper Structure.}
After some preliminaries in Section~\ref{sec:prelims}, we give an outline of our algorithm in Section~\ref{sec:outline}.
Sections~\ref{sec:merging} --~\ref{sec:seporacle} contain the various components of our algorithm.
Section~\ref{sec:applications} gives a detailed description of our applications to learning structured distribution families, 
and we conclude in Section~\ref{sec:experiments} with our experimental evaluation.


\section{Preliminaries} \label{sec:prelims}
We consider univariate probability density functions (pdf's) defined over a known finite interval $I \subseteq \R$.
For an interval $J \subseteq I$ and a positive integer $k$, we will denote by $\familyI_J^k$ the family
of all sets of $k$ disjoint intervals $I_1, \ldots, I_k$ where each $I_i \subseteq J$.
For a measurable function $g: I \to \R$ and a measurable set $S$, let $g(S)
  \eqdef \int_S g$.
The $L_1$-norm of $g$ over a subinterval $J \subseteq I$ is defined
as $\| g \|_{1, J} \eqdef \int_J |g(x)| dx$.
More generally, for any set of disjoint intervals $\mathcal{J} \in
\familyI_I^k$, we define $\| g \|_{1, \mathcal{J}} = \sum_{J \in
  \mathcal{J}} \| g \|_{1, J}$.

We now define a norm which induces a corresponding distance metric that will be crucial for this paper:

\begin{definition}[$\Ak$-norm]
Let $k$ be a positive integer and let $g: I \to \R$ be measurable.
For any subinterval $J \subseteq I$, {\em the
  $\Ak$-norm of $g$ on $J$} is defined as
\[\| g \|_{\Ak, J} \eqdef\sup_{\setI \in \familyI_J^k} \sum_{M \in \setI} |g(M)| \; .\]
When $J = I$, we omit the second subscript and simply write $\| g \|_{\Ak}$.

More generally, for any set of disjoint intervals $\mathcal{J} = \{J_1, \ldots, J_\ell\}$ where each $J_i \subseteq I$, we define
\[\| g \|_{\Ak, \mathcal{J}} = \sup_{\setI} \sum_{J \in \setI} |g(J)|\]
where the supremum is taken over all $\setI \in \familyI_J^k$ such that for all $J \in \setI$ there is a $J_i \in \mathcal{J}$ with $J \subseteq J_i$.
\end{definition}

We note that the definition of the  $\Ak$-norm in this work is slightly different than that in~\cite{DL:01, CDSS14} but is easily
seen to be essentially equivalent. The VC inequality (Theorem~\ref{thm:vc}) along with uniform
convergence bounds (see, e.g., Theorem 2.2. in~\cite{CDSS13} or p.\ 17 in~\cite{DL:01}), yields the following:

\begin{corollary}
\label{cor:vc_whp}
Fix $0< \eps$ and $\delta <1$. Let $f: I \to \R_+$ be an arbitrary pdf over $I$, and let $\fhat$ be the empirical pdf obtained after taking $n = \Theta ((k + \log 1 / \delta) / \epsilon^2)$ i.i.d.\ samples from $f$.
Then with probability at least $1 - \delta$,
\[\| f - \fhat \|_{\Ak} \leq \epsilon \; .\]
\end{corollary}

\begin{definition}
Let $g: I \to \R$. We say that $g$ has at most $k$ sign changes if there exists a partition of $I$ into intervals
$I_1, \ldots, I_{k + 1}$ such that for all $i \in [k+1]$ either $g(x) \ge 0$ for all $x \in I_i$ or $g(x) \le 0$ for all $x \in I_i$.
\end{definition}

We will need the following elementary facts about the $\Ak$-norm.

\begin{fact}
\label{lem:Ak-basics}
Let $J \subseteq I$ be an arbitrary interval or a finite set of intervals. 
Let $g: I \to \R$ be a measurable function.
\begin{enumerate}
\item[(a)] If $g$ has at most $k - 1$ sign
  changes in $J$, then
$\norm{g}_{1, J} \; = \; \norm{g}_{\Ak,J} \; .$
\item[(b)]
For all $k \ge 1$, we have
$
  \norm{g}_{\Ak, J} \; \leq \; \norm{g}_{1,J} \; .$
\item[(c)] Let $\alpha$ be a positive integer. Then,
$
  \norm{g}_{\mathcal{A}_{\alpha \cdot k},I} \; \leq \; \alpha \cdot \norm{g}_{\Ak,I} \; .
$
\item[(d)] Let $f:I \to \R_+$ be a pdf over $I$, and let $\mathcal{J}_1, \ldots, \mathcal{J}_\ell$ be finite sets of disjoint subintervals of $I$, such that for all $i, i'$ and for all $I \in \mathcal{J}_i$ and $I' \in \mathcal{J}_{i'}$, $I$ and $I'$ are disjoint. Then, for all positive integers $m_1, \ldots, m_\ell$,
$
\sum_{i = 1}^\ell \norm{f}_{\mathcal{A}_{m_i}, \mathcal{J}_i} \leq \norm{f}_{\mathcal{A}_{M}} \; ,
$
where $M = \sum_{i = 1}^\ell m_i$.
\end{enumerate}
\end{fact}

\section{Paper outline} \label{sec:outline}

In this section, we give a high-level description of our algorithm for learning $t$-piecewise degree-\pdeg{} polyonomials.
Our algorithm can be divided into three layers.

\paragraph{Level 1: General merging (Section \ref{sec:merging}).} At the top level, we design an iterative merging algorithm for finding the closest piecewise polynomial approximation to the unknown target density.
Our merging algorithm applies more generally to broad classes of piecewise hypotheses.
Let $\mathcal{D}$ be a class of hypotheses satisfying the following:
(i) The number of intersections between any two hypotheses in $\mathcal{D}$ is bounded.
(ii) Given an interval $J$ and an empirical distribution $\fhat$, we can efficiently find the best fit to $\fhat$ from functions in $\mathcal{D}$ with respect to the $\Ak$-distance.
(iii) We can efficiently compute the $\Ak$-distance between the empirical distribution and any hypothesis in $\mathcal{D}$.
Under these assumptions, our merging algorithm agnostically learns piecewise hypotheses where each piece is in the class $\mathcal{D}$.

In Section \ref{ssec:hist}, we start by presenting our merging algorithm for 
the case of piecewise constant hypotheses. This interesting special case captures many of the ideas of the general case. 
In Section \ref{sec:generalmerging}, we proceed to present our general merging algorithm that applies all classes of distributions satisfying properties (i)-(iii).

When we adapt the general merging algorithm to a new class of piecewise hypotheses, the main algorithmic challenge is constructing a procedure for property (ii).
More formally, we require a procedure with the following guarantee.
\begin{restatable}{definition}{generalprojection}
\label{def:proj_gen}
Fix $\eta > 0$. An algorithm $\oracle_p (\fhat, J,
\eta)$ is an \emph{$\eta$-approximate  $\Ak$-projection oracle} for
$\mathcal{D}$ if it takes as input an interval $J$ and $\fhat$, and returns
a hypothesis $h \in \mathcal{D}$ such that \[ \| h - \fhat \|_{\Ak} \leq
\inf_{h' \in \mathcal{D}} \| h' - \fhat \|_{\Ak,J} + \eta \; .\] 
\end{restatable}

One of our main contributions is an efficient $\Ak$-projection oracle for the class of degree-\pdeg{} polynomials, which we describe next.

\paragraph{Level 2: $\Ak$-projection for polynomials (Section \ref{sec:ak-polys}).}
Our $\Ak$-projection oracle computes the coefficients $c\in\R^{\pdeg+1}$ of a degree-\pdeg{} polynomial $p_c$ that approximately minimizes the $\Ak$-distance to the 
empirical distribution $\fhat$ in the given interval $J$.
Moreover, our oracle ensures that $p_c$ is non-negative on $J$.

At a high-level, we formulate the \Ak-projection as a convex optimization problem.
A key insight is that we can construct an efficient, approximate \emph{separation oracle} for the set of polynomials that have an \Ak-distance of at most $\tau$ to the empirical distribution $\fhat$.
Combining this separation oracle with existing convex optimization algorithms allows us to solve the feasibility problem of checking whether we can achieve a given \Ak-distance $\tau$.
We then convert the feasibility problem to the optimization variant via a binary search over $\tau$.

Note that the set of non-negative polynomials is a spectrahedron (the
feasible set of a semidefinite program).
After restricting the set of coefficients to non-negative polynomials,
we can simplify the definition of the \Ak-distance: it suffices to
consider sets of intervals with endpoints at the locations of samples.
Hence, we can replace the supremum in the definition of the \Ak-distance
by a maximum over a finite set, which shows that the set of polynomials that are both non-negative and $\tau$-close to $\fhat$ in \Ak-distance  is also a spectrahedron. 
This suggests that the \Ak-projection problem could be solved by a black-box application of an SDP solver.
However, this would lead to a running time that is \emph{exponential}
in $k$ because there are more than $\binom{s}{2k}$ possible sets
of intervals, where $s$ is the number of sample points in the current interval $J$.\footnote{While the authors of \cite{CDSS14} introduce an encoding of the
\Ak-constraint with fewer linear inequalities, their approach
increases the number of variables in the optimization problem to
depend polynomially on $1/\eps$, which leads to an $\Omega(\poly(\pdeg+1)/\eps^{3.5})$
running time.
In contrast, our approach achieves a nearly optimal dependence on $\eps$ that is $\Otilde(\poly(\pdeg + 1)/ \eps^2)$.}

Instead of using black-box SDP or LP solvers, we construct an
algorithm that exploits additional structure in the \Ak-projection
problem. 
Most importantly, our algorithm separates the dimension of the desired
degree-\pdeg{} polynomial from the number of samples (or equivalently,
the error parameter $\eps$). 
This allows us to achieve a running time that is \emph{nearly-linear} for a wide range of distributions.
Interestingly, we can solve our SDP significantly faster than the LP
which has been proposed in \cite{CDSS14} for the same problem.
We achieve this by combining Vaidya's cutting plane method \cite{Vaidya:96} with an efficient separation oracle that leverages the structure of the \Ak-distance.
This separation oracle is the third level of our algorithm, which we describe next.

\paragraph{Level 3: \Ak-separation oracle for polynomials (Section \ref{sec:seporacle}).}
Our separation oracle efficiently tests two properties for a given polynomial $p_c$ with coefficients $c \in \R^{\pdeg + 1}$:
(i) Is the polynomial $p_c$ non-negative on the given interval $J$?
(ii) Is the \Ak-distance between $p_c$ and the empirical distribution $\fhat$ at most $\tau$?
We implement Test (i) by using known algorithms for finding roots of real polynomials efficiently \cite{Pan2001}.
Note, however, that root-finding algorithms cannot be exact for degrees larger than $4$.
Hence, we can only approximately Test (i), which necessarily leads to an \emph{approximate} separation oracle.
Nevertheless, we show that such an approximate oracle is still sufficient for solving the convex program outlined above.

At a high level, our algorithm proceeds as follows.
We first verify that our current candidate polynomial $p_c$ is ``nearly'' non-negative at every point in $J$.
Assuming that $p_c$ passes this test, we then focus on the problem of computing the $\Ak$-distance between $p_c$ and $\fhat$.
We reduce this problem to a discrete variant by showing that the endpoints of intervals jointly maximizing the \Ak-distance are guaranteed to coincide with sample points of the empirical distribution (assuming $p_c$ is nearly non-negative on the current interval).
Our discrete variant of this problem is related to a previously studied question in computational biology, namely 
finding maximum-scoring DNA segment sets~\cite{Csuros04}.
We exploit this connection and give a combinatorial algorithm for this discrete variant that runs in time nearly-linear in the number of sample points in $J$ and the degree $\pdeg$.
Once we have found a set of intervals maximizing the \Ak-distance, we can convert it to a separating hyperplane for the polynomial coefficients $c$ and the set of non-negative polynomials with \Ak-distance at most $\tau$ to $\fhat$.

Combining these ingredients yields our general algorithm with the performance guarantees stated in Theorem \ref{thm:main-intro}.

\section{Iterative merging algorithm} \label{sec:merging}
In this section, we describe and analyze our iterative merging algorithm.
We start with the case of histograms and then provide the generalization to piecewise polynomials.

\subsection{The histogram merging algorithm} \label{ssec:hist}
A \emph{$t$-histogram} is a function $h : I \to \R$ that is piecewise constant with at most $t$ interval pieces, i.e., 
there is a partition of $I$  into intervals 
$I_1, \ldots, I_{t'}$ with $t' \leq t$  such that $h$ is constant on each $I_i$.
Given sample access to an arbitrary pdf $f$ over $I$ and a positive integer $t$, 
we would like to efficiently compute a good $t$-histogram approximation to $f$. 
Namely, if $\mathcal{H}_t = \mathcal{H}_t (I)$ denotes the set of $t$-histogram probability density functions over $I$ and
$\OPT_t = \inf_{g \in \mathcal{H}_t} \| g - f \|_1$, 
our goal is to output an $O(t)$-histogram $h: I \to \R$ that satisfies
$\| h - f \|_1 \leq C \cdot \OPT_t + O(\eps)$
with high probability over the samples, where $C$ is a universal constant.

The following notion of flattening a function over an interval will be crucial for our algorithm:

\begin{definition}
\label{def:flattening}
For a function $g: I \to \R$ and an interval $J = [u, v] \subseteq I$, 
we define the \emph{flattening of $g$ over $J$}, denoted
$\flatten{g}_{J}$, to be the constant function defined on $J$ as
\[
\flatten{g}_{J}(x)\eqdef \frac{g(J)}{v - u} \quad \text{for all } x\in J.
\]
For a set $\mathcal{I}$ of disjoint intervals in $I$, we define the \emph{flattening of $g$ over $\mathcal{I}$} to be the function $\flatten{g}_{\mathcal{I}}$ on 
$\cup_{J \in \mathcal{I}} J$ which for each $J \in \mathcal{I}$ satisfies $\flatten{g}_{\mathcal{I}} (x) = \flatten{g}_J (x)$ for all $x \in J$.
\end{definition}

We start by providing an intuitive explanation of our algorithm followed by a proof of correctness.
The algorithm draws $n = \Theta((t + \log 1 / \delta) / \eps^2)$ samples $x_1 \leq x_2 \leq \ldots \leq x_n$ from $f$. 
We start with the following partition of $I=[a,b]$:
 \begin{align}
\label{eqn:i0}
\mathcal{I}_0 = \{ [a, x_1), [x_1, x_1], (x_1, x_2), [x_2, x_2],
\ldots, (x_{n - 1} , x_n), [x_n, x_n], (x_n, b] \}.
\end{align}
This is the partition where each interval is either a single
sample point or the interval between two consecutive samples.
Starting from this partition, our algorithm greedily merges pairs of consecutive intervals
in a sequence of iterations. When deciding which interval pairs to merge, the following notion of
approximation error will be crucial:

\begin{definition}
For a function $g: I \to \R$ and an interval $J \subseteq I$, define 
$\err (g, J) = \norm{g - \overline{g}_J}_{\mathcal{A}_1, J} \; .$
We call this quantity the \emph{$\mathcal{A}_1$-error of $g$ on $J$}.
\end{definition}

In the $j$-th iteration,
given the current interval partition $\mathcal{I}_j$, we greedily merge pairs of
consecutive intervals to form the new partition $\mathcal{I}_{j +1}$.  
Let $s_j$ be the number of intervals in $\mathcal{I}_j$. 
In particular, given $\setI_j = \{ I_{1, j}, \ldots, I_{s_j, j} \}$, we consider the intervals 
$$I'_{\ell, j + 1} = I_{2\ell - 1, j} \cup I_{2\ell, j}$$ for all $1 \leq \ell \leq s_j / 2$.\footnote{We assume $s_j$ is even for simplicity.}
We first iterate through $1 \leq \ell \leq s_j / 2$ and calculate the quantities 
$$e_{\ell, j} = \err (\fhat, I'_{\ell, j + 1}) \;,$$ 
i.e., the $\mathcal{A}_1$-errors of the empirical distribution on the
candidate intervals.

To construct $\setI_{j + 1}$, the algorithm keeps track of the largest $O(t)$ errors $e_{\ell, j}$.
For each $\ell$ with $e_{\ell, j}$ being one of the $O(t)$ largest errors, we do not merge the corresponding intervals $I_{2\ell - 1, j}$ and $I_{2\ell, j}$.
That is, we include $I_{2 \ell -1, j}$ and $I_{2 \ell, j}$ in the new partition $\setI_{j + 1}$.
Otherwise, we include their union $I'_{\ell, j+1}$ in $\setI_{j + 1}$.
We perform this procedure $O(\log \frac{1}{\eps})$ times and arrive at some final partition $\setI$. 
Our output hypothesis is the flattening of $\widehat{f}$ with respect to $\setI.$

For a formal description of our algorithm, see the pseudocode given in Algorithm 1 below. 
In addition to the parameter $t$, the algorithm has a parameter
$\alpha \ge 1$  that controls the trade-off between the 
approximation ratio $C$ achieved by the algorithm and the number of pieces in the output histogram.  

\begin{algorithm}[htb]
\begin{algorithmic}[1]
\Function{ConstructHistogram}{$f, \, t, \, \alpha, \, \eps, \, \delta$}
\State Draw $n = \Theta((\alpha t + \log 1 / \delta) / \eps^2)$ samples $x_1 \leq x_2 \le \ldots \leq x_n$.
\State Form the empirical distribution $\fhat$ from these samples.

\State \label{line:setI0def} Let $ \mathcal{I}_0 \gets \{ [a, x_1), [x_1, x_1], (x_1, x_2), \ldots, (x_{n - 1} , x_n), [x_n, x_n], (x_n, b] \}$ be the initial partition.

\State $j \gets 0$

\While{ $| \setI_j | > 2 \alpha \cdot t$}
  \State Let $ \setI_j = \{I_{1, j}, I_{2, j}, \ldots, I_{s_j-1, j}, I_{s_j, j}\} $
  \For{$\ell \in \{1, 2, \ldots, \frac{s_j}{2}\}$}
    \State $ I'_{\ell, j + 1} \gets I_{2\ell - 1, j} \cup I_{2\ell, j}$
    \State $e_{\ell, j} \gets \err (\fhat, I'_{\ell, j + 1})$
  \EndFor
  \State \label{line:largeerrors} Let $L$ be the set of $\ell \in  \{1, 2, \ldots, \frac{s_j}{2}\}$ 
  with the $\alpha t$ largest errors $e_{\ell, j}$.
  \State Let $M$ be the complement of $L$.
  \State $\setI_{j+1} \gets \bigcup\limits_{\ell \in L} \{I_{2\ell - 1,j}, I_{2\ell,j}\}$
  \State $\setI_{j+1} \gets \setI_{j+1} \cup \{I'_{\ell, j+1} \mid \ell \in M \}$ \label{line:largeerrors2}
  \State $j \gets j + 1$
\EndWhile
\State \textbf{return} $\setI = \setI_j$ and the flattening $\flatten{\fhat_{\setI}}$
\EndFunction
\end{algorithmic}
\caption{Approximating with histograms by merging.}
\label{alg:merging}
\end{algorithm}

The following theorem characterizes the performance of Algorithm~1, 
establishing the special case of Theorem~\ref{thm:main-intro} corresponding to $\pdeg=0$.

\begin{theorem} \label{thm:mergingmain}
Algorithm $\textsc{ConstructHistogram}(f, t, \alpha, \eps, \delta)$ draws 
$n= O((\alpha t + \log 1 / \delta) / \eps^2)$ samples from $f$,  
runs in time $O(n \left(\log (1/\eps) + \log \log(1/\delta)\right))$, and outputs
a hypothesis $h$ and a corresponding partition $\setI$ of size $|\setI| \; \leq \; 2 \alpha \cdot t$  such that  
with probability at least $1-\delta$ we have
\begin{equation}
\label{eq:main}
\| h - f \|_{1}  \; \leq \;  2 \cdot \OPT_t + \frac{4 \cdot \OPT_t + 4 \eps}{\alpha - 1} + \eps\;.
\end{equation}
\end{theorem}

\begin{proof}
We start by analyzing the running time.
To this end, we show that the number of intervals decreases exponentially with the number of 
iterations. In the $j$-th iteration, we merge all but $\alpha t$
intervals. Therefore, 
\[
s_{j+1}=\alpha t + \frac{s_j-\alpha t}2=\frac {3s_j}4+\frac{2\alpha t-s_j}4.
\]
Note that the algorithm enters the while loop when $s_j>2\alpha t$,
implying that
\[
s_{j+1}<\frac {3s_j}4.
\]
By construction, the number of intervals is at least $\alpha t$ when the algorithm
exits the while loop. Therefore, the number of iterations of the while loop
is at most
\[
O\left(\log\frac{ {n}}{\alpha t}\right)=O\left(\log (1/\eps) + \log \log(1/\delta) \right), 
\]
which follows by substituting the value of $n$ from the 
statement of the theorem. 
We now show that each iteration takes time $O(n)$.
Without loss of generality, assume that we compute the $\mathcal{A}_1$-distance only over intervals ending at a data sample. 
For an interval $J=[c,d]$ containing $m$ sample points, $x_1, \ldots,
x_m$, let $C_j = \frac{(x_j-x_1)}{jn}-\frac{(d-c)}n$. The $\mathcal{A}_1$-error of $\fhat$ on $J$ is given by $\max C_j - \min C_j$ and can therefore
be computed in time proportional to the number of sample points in the
interval. Therefore, the total time of the
algorithm is $O(n (\log (1/\eps) + \log \log(1/\delta)))$, as claimed.

We now proceed to bound the learning error.
Let $\setI = \{I_1, \ldots, I_{t'}\}$ be the partition of $I$ returned by \textsc{ConstructHistogram}.
The desired bound on $|\setI|$ follows immediately because the algorithm terminates only when $|\setI | \leq 2 \alpha t$.
The rest of the proof is dedicated to Equation (\ref{eq:main}).

Fix $\hopt \in \mathcal{H}_t$ such that $\|\hopt - f\|_1 = \OPT_t.$
Let $\setI^{\ast} = \{I^{\ast}_1, \ldots, I^{\ast}_t\}$ be the
partition induced by the discontinuities of $\hopt$. 
Call a point at a boundary of any $I^\ast_j$ a \emph{jump} of $\hopt$.
For any interval $J \subseteq I $, we define $\Gamma (J)$ to be the
number of jumps of $\hopt$ in the interior of $J$.
Since we draw $n = \Omega((\alpha t + \log 1 / \delta) / \eps^2)$ samples,  
Corollary~\ref{cor:vc_whp} implies that with probability at least $1-\delta$, we have
\[ \| \fhat - f\|_{\mathcal{A}_{(2 \alpha + 1) t}} \le \eps \; . \]
We condition on this event throughout the analysis.

We split the total error into three terms based on the final partition $\setI$:

\begin{description}
\item[Case 1:] Let $\mathcal{F}$ be the set of intervals  in $\setI$
  with zero jumps in $\hopt$, i.e., $\mathcal{F} =
  \{J \in \setI \, | \, \Gamma(J) = 0\}$. 
\item[Case 2a:] Let $\mathcal{J}_0$ be the set of intervals in $\setI$ that were created in the initial partitioning step of the algorithm and which contain a jump of $\hopt$, 
i.e., $\mathcal{J}_0 = \{J \in \setI \mid \Gamma(J) > 0 \mbox{ and } J \in \setI_0 \}$.
\item[Case 2b:] Let $\mathcal{J}_1$ be the set of intervals in
    $\setI$ that contain at least one jump and were created by merging two other intervals,
i.e., $\mathcal{J}_1 = \{J \in \setI \mid  \Gamma(J) >0 \text{ and } J \notin \setI_0 \}$.
\end{description}
Notice that $\mathcal{F}$, $\mathcal{J}_0$, and $\mathcal{J}_1$ form a partition of $\setI$, and thus
\[ \| h - f \|_1 = \| h - f \|_{1, \mathcal{F}} + \| h - f \|_{1, \mathcal{J}_0} + \| h - f \|_{1, \mathcal{J}_1} \; .\]

We will bound these three terms separately. In
particular, we will show:
\begin{align}
 \| h - f \|_{1, \mathcal{F}} \; &\leq \; 2 \cdot \| f - \hopt \|_{1, \mathcal{F}}  + \| \fhat - f\|_{\mathcal{A}_{|\mathcal{F}|}, \mathcal{F}} \;, \label{eq:case1} \\
\| h - f \|_{1, \mathcal{J}_0} \; &\leq \; \| \fhat - f\|_{\mathcal{A}_{|\mathcal{J}_0|}, \mathcal{J}_0} \;, \label{eq:case2a}\\
 \| h - f \|_{1, \mathcal{J}_1} \; &\leq  \; \frac{4 \cdot \OPT_t + 4 \eps}{\alpha - 1} +2\cdot \| f - \hopt\|_{1, \mathcal{J}_1} +\| \fhat - f\|_{\mathcal{A}_{|\mathcal{J}_1| + t}, \mathcal{J}_1} \;. \label{eq:case2b}
\end{align}
Using  these results along with the fact that $\| f - \hopt \|_{1, \mathcal{F}} + \| f - \hopt \|_{1, \mathcal{J}_1} \leq \OPT_t$, we
have 
\begin{align*}
\| h - f\|_1 &\leq 2 \cdot \OPT_t + \frac{4 \cdot \OPT_t + 4\eps}{\alpha - 1} +  \| \fhat - f\|_{\mathcal{A}_{|\mathcal{F}|}, \mathcal{F}} + \| \fhat - f\|_{\mathcal{A}_{|\mathcal{J}_0|}, \mathcal{J}_0} + \| \fhat - f\|_{\mathcal{A}_{|\mathcal{J}_1| + t}, \mathcal{J}_1} \\
&\stackrel{(a)}{\leq} 2 \cdot \OPT_t + \frac{4 \cdot \OPT_t + 4\eps}{\alpha - 1} +  \| \fhat - f\|_{\mathcal{A}_{(2 \alpha + 1) t}} \\
& \stackrel{(b)}{\leq} 2 \cdot \OPT_t + \frac{4 \cdot \OPT_t + 4 \eps}{\alpha - 1} + \eps \;,
\end{align*}
where inequality $(a)$ follows from Fact~\ref{lem:Ak-basics}(d) and inequality $(b)$ follows from the VC inequality.
Thus, it suffices to prove Equations (\ref{eq:case1})--(\ref{eq:case2b}).

\paragraph{Case 1.} We first consider the interval
$\mathcal{F}$. By the triangle inequality, we have  
\[\| h - f \|_{1, \mathcal{F}} \leq \| f - \hopt \|_{1, \mathcal{F}} + \| h - \hopt \|_{1, \mathcal{F}} \; .\]
Thus to show (\ref{eq:case1}), it suffices to show that 
\begin{equation}
\label{eq:case1-key}
\| h - \hopt \|_{1, \mathcal{F}} \leq \| f - \hopt \|_{1, \mathcal{F}} + \norm{\fhat - f}_{\mathcal{A}_{|\mathcal{F}|}, \mathcal{F}} \; .
\end{equation}

We prove a slightly more general version of~\eqref{eq:case1-key} that holds
over all finite sets of intervals not containing any jump of $\hopt$. We will use this general version also later in our proof.

\begin{lemma}
\label{lem:flat_error_bound}
Let $\mathcal{J} \in \familyI_I^\ell$ so that $\Gamma(J) = 0$ for all $J \in \mathcal{J}$.
Let $\overline{h} = \flatten{{\fhat}_{\mathcal{J}}}$ denote the flattening of $\fhat$ on $\mathcal{J}$.
Then
\[
\norm{\overline{h} - h^*}_{1, \mathcal{J}} \leq \norm{f - \hopt}_{1, \mathcal{J}} + \norm{\fhat - f}_{\mathcal{A}_\ell, \mathcal{J}} \; .
\]
\end{lemma}
Note that this is indeed a generalization of (\ref{eq:case1-key}) since for any point $x$ in any interval of $\mathcal{F}$, we have $h(x) = \flatten{{\fhat}_\mathcal{F}} (x)$.

\begin{proof}[Proof of Lemma \ref{lem:flat_error_bound}]
In any interval $J \in \mathcal{J}$ with $\Gamma (J) = 0$, we have 
\[
\norm{\overline{h} - \hopt}_{1, J} \; \stackrel{(a)}{=} \;
|\overline{h}(J) - \hopt(J)| \; \stackrel{(b)}{=} \; |\fhat (J) -
\hopt(J)|,
\]
where $(a)$ follows from the fact that $\overline{h}$ and $\hopt$ are constant in $J$, and
$(b)$ follows from the definition of $\overline{h}$.
Thus, we get
\begin{align*}
\| \overline{h} - \hopt \|_{1, \mathcal{J}} &= \sum_{J \in \mathcal{J}} \norm{\overline{h} - \hopt}_{1,  J} \; \\
&= \; \sum_{J \in \mathcal{J}}  |\fhat (J) - \hopt(J)|\\
                            &\stackrel{(c)}{\le} \; \sum_{J \in \mathcal{J}}  |\fhat (J) - f(J)| + \sum_{J \in \mathcal{J}}  |f (J) - \hopt(J)|\ \\
                           &\stackrel{(d)}{\le} \;  \| \fhat -
                           f\|_{\mathcal{A}_{|\mathcal{J}|},     \mathcal{J}} +  \| f - \hopt \|_{1,
                             \mathcal{J}} 
\end{align*}
where $(c)$ uses the triangle inequality, and $(d)$ follows from the
definition of $\mathcal{A}_{k}$-distance. 
\end{proof}

\paragraph{Case 2a.}
Next, we analyze the error for the intervals in $\mathcal{J}_0$.
The set $\setI_0$ contains only singletons and intervals with no
  sample points. By definition, only the intervals in $\setI_0$
  that contain no samples may contain a jump of $\hopt$.
The singleton intervals containing the sample points do not include
jumps and are hence covered by Case 1. 
Since the intervals in $\mathcal{J}_0$ do not contain any samples, our algorithm
assigns
\[
h(J) = \fhat (J) = 0
\]
for any $J\in\mathcal{J}_0$. 
Hence,
\begin{equation*}
  \norm{h - f}_{1, \mathcal{J}_0} = \norm{f}_{1, \mathcal{J}_0} \; . 
\end{equation*}
We thus have the following sequence of (in)equalities:
\begin{align*}
\norm{h-f}_{1,\mathcal{J}_0} &=\norm{f}_{1, \mathcal{J}_0} \\
                        &= \; \sum_{J \in \mathcal{J}_0} |f(J)| \\
                         &= \; \sum_{J \in \mathcal{J}_0}  |f(J) - \fhat (J)|   \\
                         &\leq \; \norm{f -
                           \fhat}_{\mathcal{A}_{|\mathcal{J}_0|},
                           \mathcal{J}_0} \;, 
\end{align*}
where the last step uses the definition of the $\Ak$-norm.

\paragraph{Case 2b.} Finally, we bound the error for intervals in $\mathcal{J}_1$, i.e.,
intervals that were created by merging in some iteration of our algorithm and also contain jumps.
As before, our first step is the following triangle inequality:
\[\| h - f \|_{1, \mathcal{J}_1} \leq \| h - \hopt \|_{1, \mathcal{J}_1} + \| \hopt - f \|_{1, \mathcal{J}_1} \; .\]

Consider an interval  $J \in \mathcal{J}_1$. Since $h$ is
  constant in $J$ and $\hopt$ has  $\Gamma(J)$ jumps in $J$,  
 $h - \hopt$ has at most $\Gamma(J)$ sign changes in $J$. Therefore,
\begin{align*}
\| h - \hopt \|_{1, J} &\stackrel{(a)}{=} \| h - h^* \|_{\mathcal{A}_{\Gamma (J) + 1}, J}  \nonumber \\
	&\stackrel{(b)}{\leq} \| h - \fhat \|_{\mathcal{A}_{\Gamma(J) + 1}, J} + \| \fhat - f\|_{\mathcal{A}_{\Gamma(J) + 1}, J} + \| f - \hopt \|_{\mathcal{A}_{\Gamma(J) + 1}, J} \nonumber \\
	&\stackrel{(c)}{\leq} ( \Gamma(J) + 1) \| h - \fhat \|_{\mathcal{A}_1, J} +  \| \fhat - f\|_{\mathcal{A}_{\Gamma(J) + 1}, J}  + \| f - \hopt \|_{1, J} \numberthis \label{eqn:five} \;,
\end{align*}
where equality $(a)$ follows from Fact~\ref{lem:Ak-basics}(a), inequality $(b)$ is the triangle inequality, and 
inequality $(c)$ uses Fact~\ref{lem:Ak-basics}(c). Finally, we bound
the $\mathcal{A}_1$-distance in the first term above. 
\begin{lemma} \label{lem:2b}
For any $J \in \mathcal{J}_1$, we have
\begin{equation}
\| h - \fhat \|_{\mathcal{A}_1, J} \leq \frac{2 \OPT_t + 2
  \eps}{(\alpha - 1) t} \;. \label{eq:case2bkeylemma} 
\end{equation}
\end{lemma}
Before proving the lemma, we show how to use it to complete Case 2b. 
Since $h$ is the flattening of $\fhat$ over $J$,  
we have that $\| h - \fhat \|_{\mathcal{A}_1, J} = \err(\fhat, J)$.
Applying~\eqref{eqn:five} gives:
\begin{align*}
\| h - \hopt \|_{1, \mathcal{J}_1} &= \sum_{J \in \mathcal{J}_1} \| h - \hopt \|_{1, J} \\
 &\leq \sum_{J \in \mathcal{J}_1} \left(( \Gamma(J)+1) \| h - \fhat \|_{\mathcal{A}_1, J} + \| \fhat - f\|_{\mathcal{A}_{\Gamma(J) + 1}, J} + \| f - \hopt \|_{1, J} \right) \\
	&\leq\frac{2 \cdot \OPT_t + 2\eps}{(\alpha -
          1) t} \cdot \left( \sum_{J \in \mathcal{J}_1} ( \Gamma(J) +
          1) \right) +\sum_{J \in \mathcal{J}_1} \| \fhat - f
        \|_{\mathcal{A}_{\Gamma(J)+1}, J} + \| f -\hopt 
        \|_{1,\mathcal{J}_1} \\ 
	& \stackrel{(a)}{\leq} \frac{4 \cdot \OPT_t + 4 \eps}{(\alpha - 1)} + \| \fhat - f \|_{\mathcal{A}_{t + |\mathcal{J}_1|}, \mathcal{J}_1} + \| f -\hopt \|_{1,\mathcal{J}_1}
\end{align*}
where inequality $(a)$ uses the fact that $\Gamma(J)\ge1$ for these intervals and hence
\[\sum_{J \in \mathcal{J}_1} (\Gamma(J) + 1) \le 2 \sum_{J \in
  \mathcal{J}_1} \Gamma(J) \leq 2 t \; .\]

We now complete the final step by proving Lemma~\ref{lem:2b}.

\begin{proof}[Proof of Lemma~\ref{lem:2b}]

Recall that in each iteration of our algorithm, we merge all pairs of intervals
except those with the $\alpha t$ largest errors. Therefore, if two
intervals were merged, there were at least $\alpha t$ other pairs of intervals with larger error. We will use this fact to bound the error on the intervals in $\mathcal{J}_1$. 

Consider any interval $J \in \mathcal{J}_1$, and suppose it was
created in the $j$th iteration of the while loop of our algorithm, i.e.,  
$J = I'_{i, j+1} = I_{2i-1, j} \cup I_{2i, j}$ for some $i \in \{1, \ldots, s_j/2\}.$
Note that this interval is not merged again in the remainder of the
algorithm. Recall that the intervals $I'_{i, j+1}$, for $i \in  \{1, \ldots,
s_j/2\}$, are the possible candidates for merging at iteration $j$. 
Let $h' = \overline{\fhat}_{\setI'_{j+1}}$ be the distribution
obtained by flattening the empirical distribution over these candidate
intervals $\setI'_{j+1} = \{ I'_{1,j+1}, \ldots, I'_{s_j / 2, j+1}\}$. Note that $h'(x) = h(x)$ for $x \in J$ because $J$ was
created in this iteration.

Let ${\cal L}$ be the set of candidate intervals $I'_{i, j +1}$ in the set $\setI'_{j+1}$ 
with the largest $\alpha \cdot t$ errors $e(\fhat, I'_{i, j +1})$. Let
${\cal L}_0$ be the intervals in $\cal L$ that do not contain any
jumps of $\hopt$. Since $\hopt$ has at most $t$ jumps, $\abs{{\cal
    L}_0} \geq (\alpha - 1) t$.  Moreover, for any $I' \in
\mathcal{L}_0$, by the triangle inequality
\begin{align*}
\err(\fhat, I') &=  \| h' - \fhat \|_{\mathcal{A}_1, I'} \\
&\leq \| h' - \hopt \|_{\mathcal{A}_1, I'} + \| f - \hopt \|_{\mathcal{A}_1, I'} + \|f - \fhat \|_{\mathcal{A}_1, I'} \ \\
&\leq \| h' - \hopt \|_{\mathcal{A}_1, I'} + \| f - \hopt \|_{1, I'} + \| f - \fhat \|_{\mathcal{A}_1, I'} \; .
\end{align*}
Summing over the intervals in ${\cal L}_0$, 
\begin{eqnarray}
\sum_{I' \in \mathcal{L}_0} \err(\fhat, I') &\leq& \sum_{I' \in \mathcal{L}_0} \left( \| h' - \hopt \|_{\mathcal{A}_1, I'} + \| f - \hopt \|_{1, I'} + \| f - \fhat \|_{\mathcal{A}_1, I'} \right) \nonumber \\
&\leq& \left( \sum_{I' \in \mathcal{L}_0} \| h' - \hopt \|_{\mathcal{A}_1, I'} \right) +\| f - \hopt \|_{1, \mathcal{L}_0}  + \| f - \fhat \|_{\mathcal{A}_{2 \alpha t}, \mathcal{L}_0} \nonumber  \\
&\leq& \left( \sum_{I' \in \mathcal{L}_0} \| h' - \hopt
  \|_{\mathcal{A}_1, I'} \right) + \OPT_t + \eps  \label{eqn:dyon} \;, 
\end{eqnarray}
{where recall that we had conditioned on the last term being at most
$\eps$ throughout the analysis. }
Since both $h$ and $\hopt$ are flat on each interval $I' \in \mathcal{L}_0$, Lemma \ref{lem:flat_error_bound} gives
$$\sum_{I' \in \mathcal{L}_0} \norm{h' - \hopt}_{\mathcal{A}_1, I'} \leq \| f- \hopt \|_{1, \mathcal{L}_0} + \| \fhat - f \|_{\mathcal{A}_{|\mathcal{L}_0|}, \mathcal{L}_0} \leq  \OPT_t  + \eps \;.$$
Plugging this into \eqref{eqn:dyon} gives
$$\sum_{I' \in \mathcal{L}_0}  \err(\fhat, I') \leq 2 \cdot \OPT_t +
2\eps \; .$$
Since $J$ was created by merging in this iteration, we have that $\err(\fhat, J)$
is no larger than $\err(\fhat, I')$ for any of the intervals $I'
\in \mathcal{L}_0$ (see lines \ref{line:largeerrors} - \ref{line:largeerrors2} of Algorithm \ref{alg:merging}), 
and therefore $\err(\fhat, J)$ is not larger than their average. 
Recalling that $\abs{\mathcal{L}_0} \geq (\alpha - 1) t$, we obtain
\begin{align*}
\err(\fhat, J) = \norm{h' - \fhat}_{\mathcal{A}_1, J}  =  \norm{h - \fhat}_{\mathcal{A}_1, J} \; &\leq \; \frac{\sum_{I' \in \mathcal{L}_0}  \err(\fhat, I')}{(\alpha - 1) t} 
                              \leq \; \frac{2 \OPT_t + 2\eps}{(\alpha - 1) t} \; ,
\end{align*}
completing the proof of the lemma. 
\end{proof}
\renewcommand{\qedsymbol}{}
\end{proof}


\subsection{The general merging algorithm}
\label{sec:generalmerging}

We are now ready to present our general merging algorithm, which is 
a generalized version of the histogram merging algorithm introduced in Section \ref{ssec:hist}.
The histogram algorithm only uses three main properties of histogram hypotheses: (i) The number of intersections between two $t$-histogram hypotheses is bounded by $O(t)$.
(ii) Given an interval $J$ and an empirical distribution $\fhat$, we can efficiently find a good histogram fit to $\fhat$ on this interval.
(iii) We can efficiently compute the $\A_1$-errors of candidate intervals.

Note that property (i) bounds the complexity of the hypothesis class and leads to a tight sample complexity bound while properties (ii) and (iii) are algorithmic ingredients.
We can generalize these three notions to arbitrary classes of piecewise hypotheses as follows.
Let $\mathcal{D}$ be a class of hypotheses.
Then the generalized variants of properties (i) to (iii) are:
(i) The number of intersections between any two hypotheses in $\mathcal{D}$ is bounded.
(ii) Given an interval $J$ and an empirical distribution $\fhat$, we can efficiently find the best fit to $\fhat$ from functions in $\mathcal{D}$ with respect to the $\Ak$-distance.
(iii) We can efficiently compute the $\Ak$-distance between the empirical distribution and any hypothesis in $\mathcal{D}$.
Using these generalized properties, 
the histogram merging algorithm naturally extends to agnostically learning piecewise hypotheses where each piece is in the class $\mathcal{D}$.

\smallskip

\noindent The following definitions formally describe the aforementioned framework.
We first require a mild condition on the underlying distribution family:

\begin{definition}
Let $\mathcal{D}$ be a family of measurable functions defined over
subsets of $I$. $\mathcal{D}$ is said to be \emph{full} if for each
$J\subseteq I$, there exists a function $g$ in $\mathcal{D}$ whose
domain is $J$. Let $\mathcal{D}_J$ be the elements of $\mathcal{D}$
whose domain is $J$.  
\end{definition}

Our next definition formalizes the notion of piecewise hypothesis whose components come from $\mathcal{D}$:

\begin{definition}
A function $h: I \to \R$ is a \emph{$t$-piece $\mathcal{D}$-function}
if there exists a partition of $I$ into intervals $I_1, \ldots, I_{t'}$
with $t' \leq t$, such that for every $i$, $1 \le i \le t'$, there exists
$h_i\in\mathcal{D}_{I_i}$ satisfying that $h = h_i$ on $I_i$. Let
$\mathcal{D}_t$ denote the set of all $t$-piece $\mathcal{D}$-functions. 
\end{definition}

The main property we require from our full
function class $\mathcal{D}$ is that any two functions in $\mathcal{D}$
intersect a bounded number of times. This is formalized in the definition below:

\begin{definition}
\label{def:sign_gen}
Let $\mathcal{D}$ be a full family over $I$ and $J \subseteq I$.
Suppose $h \in \mathcal{D}_J$ and $h' \in \mathcal{D}_k$ for some $k
\geq 1$. Let $h' = h'_{I_i}$, $1 \le i \le k$, for some interval partition $I_1, \ldots, I_k$ of $I$
and  $h'_{I_i} \in \mathcal{D}_{I_i}$. Let $s$ denote the number of endpoints of the $I_i$'s
contained in $J$.
We say that $\mathcal{D}$ is \emph{$\pdeg$-sign restricted} if
the function $h - h'$ has at most $(s  + 1) \pdeg $ sign changes on $J$, for any $h$ and $h'$.
\end{definition}

The following simple examples illustrate that histograms and more generally piecewise polynomial functions
fall into this framework.

\begin{example}
\label{ex:flat}
Let $\mathcal{H}_J$ be the set of constant functions defined on $J$. 
Then if $\mathcal{H} =\cup_{J\subseteq I} \mathcal{H}_J$, the set $\mathcal{H}_t$ of $t$-piece
$\mathcal{H}$-functions is the set of piecewise constant
functions on $I$ with at most $t$ interval pieces. 
(Note that this class is the set of {\em $t$-histograms}.)
\end{example}

\begin{example} \label{ex:pp}
For $J \subseteq I$, we define $\mathcal{P}_{J, \pdeg}$ to be set of
degree-$\pdeg$ nonnegative polynomials on $J$, and $\mathcal{P}_{\pdeg} \ed
\cup_J \mathcal{P}_{J, \pdeg}$. Since the degree $\pdeg$ will be fixed throughout
this paper, we sometimes simply denote this set by $\mathcal{P}$. 
The set $\mathcal{P}_{t, d}$ of $t$-piece $\mathcal{P}$-functions is the set of
$t$-piecewise degree-$\pdeg$ non-negative polynomials.
It is easy to see that this class is full over $I$. 
Since any two polynomials of degree $\pdeg $ intersect at most $\pdeg$ times,
it is easy to see that $\mathcal{P}_{\pdeg}$ forms a $\pdeg$-sign restricted class. 
\end{example}

\medskip

We are now ready to formally define our general learning problem.
Fix positive integers $t, d$ and a full $d$-sign restricted class of functions $\mathcal{D}$.
Given sample access to any pdf $f: I \to \R_+$, we want to compute a good 
$\mathcal{D}_t$ approximation to $f$.  We define
$\OPT_{\mathcal{D}, t} \ed \inf_{g \in \mathcal{D}_t} \| g - f \|_1 \; .$ 
Our goal is to find an $O(t)$-piece $\mathcal{D}$-function $h:I\to\R$ such that 
$\| h - f \|_1 \leq C \cdot \OPT_{\mathcal{D}, t} + O(\eps), $
with high probability over the samples, where $C$ is a universal
constant. 

\medskip

Our iterative merging algorithm takes as input samples from an arbitrary distribution, 
and outputs an $O(t)$-piecewise $\mathcal{D}$
hypothesis satisfying the above agnostic guarantee. Our algorithm assumes the existence of two subroutines, which we call $\Ak$-projection and $\Ak$-computation oracles. 
The $\Ak$-projection oracle was defined in Definition~\ref{def:proj_gen} and is restated
below along with the definition of the $\Ak$-computation oracle
(Definition~\ref{def:comp_gen}). 

\generalprojection*

\begin{definition}
\label{def:comp_gen}
Fix $\eta > 0$.  An algorithm $\oracle_c (\fhat,
h_J, J, \eta)$  is an \emph{$\eta$-approximate
  $\Ak$-computation oracle} for $\mathcal{D}$ if it takes as input
$\fhat$, a subinterval $J \subseteq I$, and a function $h_J \in
\mathcal{D}_J$, and returns a value $\xi$ such that  
\[\left| \| h_J - \fhat \|_{\Ak, J} - \xi \right| \leq \eta \; .\]
\end{definition}

We consider a $\pdeg$-sign restricted full
family $\mathcal{D}$, and a fixed $\eta > 0$. 
Let $R_p (I) = R_p(I, \fhat, \oracle_p)$ and $R_c (I) = R_c(I, \fhat, \oracle_c)$ be the time used by the oracle $\oracle_p$
and $\oracle_c$, respectively.
With a slight abuse of notation, for a collection of at most $2n$ intervals
containing $n$ points in the support of the empirical distribution, we also define $R_p(n)$ and $R_c(n)$ to be the maximum time taken by $\oracle_p$ and $\oracle_c$, respectively.

We are now ready to state the main theorem of this section:

\begin{restatable}{theorem}{generalmain}
\label{thm:generalmain}
Let $\oracle_p$ and $\oracle_c$ be $\eta$-approximate $\Ak$-projection and $\Ak$-computation oracles for $\mathcal{D}$.
Algorithm $\textsc{General-Merging}(f, t, \alpha, \eps, \delta)$ draws $n= O((\alpha \pdeg t + \log 1 / \delta) / \eps^2)$ 
samples, runs in time $O\left((R_p (n) + R_c (n) ) \log \frac n{\alpha t}\right)$, and outputs 
a hypothesis $h$ and an interval partition $\setI$ such that
$|\setI| \; \leq \; 2 \alpha \cdot t $ and with probability at least $1-\delta$, we have
\begin{equation}
\label{eq:generalmergingmain}
\| h - f \|_{1}  \; \leq \;  3 \cdot \OPT_{\mathcal{D}, t} + \frac{ \OPT_{\mathcal{D}, t} + \eps}{\alpha - 1} + 2\eps + \eta\;.
\end{equation}
\end{restatable}

In the remainder of this section, we provide an intuitive explanation of our general merging algorithm followed by a detailed pseudocode.

\medskip

The algorithm $\textsc{General-Merging}$ and its analysis is a generalization of the  \textsc{ConstructHistogram} algorithm from the previous subsection.
More formally, the algorithm proceeds greedily, as
before. We take $n = O((\alpha \pdeg t + \log 1 / \delta) / \epsilon^2)$ samples
$x_1 \leq \ldots \leq x_n$.  We construct $\setI_0$ as in~\eqref{eqn:i0}. 
In the $j$-th iteration, given the current partition $\setI_j = \{
I_{1, j}, \ldots, I_{s_j, j} \}$  with $s_j$ intervals, consider the intervals
\[
I'_{\ell, j + 1} = I_{2\ell - 1, j} \cup I_{2\ell, j}
\]
for $\ell\le s_j/2$. 
As for histograms, we want to compute the errors in each of the new
intervals created. To do this, we first call the $\Ak$--projection
oracle with $k=\pdeg+1$ on this interval to find the approximately best fit in
$\mathcal{D}$ for $\fhat$ over these new intervals, namely:
\[
h'_{\ell, j} = \oracle_p \left(\fhat, I'_{\ell, j + 1}, \frac{\eta}{ O(t)}\right).
\]
To compute the error, we call the $\Ak$--computation
oracle with $k=\pdeg+1$, i.e.: 
\[
e_{\ell, j} = \oracle_c \left(\fhat, h'_{\ell, j},
I'_{\ell, j + 1}, \frac{\eta}{ O(t)}\right) \; .
\]

As in \textsc{ConstructHistogram}, we keep the intervals with
the largest $O(\alpha t)$ errors intact and merge the remaining pairs of
intervals.  
We perform this procedure $O(\log \frac{n}{\alpha t})$ times and arrive at
some final partition $\setI$ with $O(\alpha t)$ pieces. 
Our output hypothesis is the output of $\oracle_p (\fhat, I)$ over
each of the final intervals $I$. 

The formal pseudocode for our algorithm is given in Algorithm~\ref{alg:merging-general}.
We assume that $\mathcal{D}$ and $\pdeg$ are known and fixed and are not mentioned
explicitly as an input to the algorithm. Note that we run the
algorithm with $\eta=\epsilon$ so that Theorem~\ref{thm:generalmain}
has an additional $O(\eps)$ error. 
The proof of Theorem~\ref{thm:generalmain}
 is very similar to that of the histogram merging algorithm and is deferred to
 Appendix~\ref{app:general-merging-analysis}. 

\begin{algorithm}[htb]
\begin{algorithmic}[1]
\Function{General-Merging}{$f, \pdeg, \, t, \, \alpha, \, \eps, \, \delta$}
\State Draw $n = \Theta((\alpha \pdeg t  + \log 1 / \delta) / \eps^2)$
samples $x_1 \leq x_2 \le \ldots \leq x_n$.
\State Form the empirical distribution $\fhat$ from these samples. 

\State \label{line:generalsetI0def} Let $ \mathcal{I}_0 \gets \{ [a, x_1),
[x_1, x_1], (x_1, x_2), \ldots, (x_{n - 1} , x_n), [x_n,
x_n], (x_n, b] \}$ be the initial partition. 

\State $j \gets 0$

\While{ $| \setI_j | > 2 \alpha \cdot t$}
  \State Let $ \setI_j = \{I_{1, j}, I_{2, j}, \ldots, I_{s_j-1, j}, I_{s_j, j}\} $
  \For{$\ell \in \{1, 2, \ldots, \frac{s_j}{2}\}$}
    \State $ I'_{\ell, j + 1} \gets I_{2\ell - 1, j} \cup I_{2\ell,
      j}$
    \State $h'_{\ell, j} \gets \oracle_p (\fhat, I'_{\ell, j + 1}, \frac{\eps}{2\alpha t})$
    \State $e_{\ell, j} \gets \oracle_c (\fhat, h'_{\ell, j},
    I'_{\ell, j + 1}, \frac{\eps}{ 2\alpha t}) $
  \EndFor
  \State \label{line:generallargeerrors}  Let $L$ be the set of $\ell \in  \{1, 2, \ldots, \frac{s_j}{2}\}$ 
  with the $\alpha t$ largest errors $e_{\ell, j}$.
  \State Let $M$ be the complement of $L$.
  \State $\setI_{j+1} \gets \bigcup\limits_{\ell \in L} \{I_{2\ell - 1,j}, I_{2\ell,j}\}$
  \State $\setI_{j+1} \gets \setI_{j+1} \cup \{I'_{\ell, j+1} \mid \ell \in M \}$ \label{line:generallargeerrors2}
  \State $j \gets j + 1$
\EndWhile
\State \textbf{return} $\setI = \setI_j$ and the functions $\oracle_p
(\fhat, J, \frac{\eps}{2\alpha t})$ for $J\in\setI$ 
\EndFunction
\end{algorithmic}
\caption{Approximating with general hypotheses by merging.}
\label{alg:merging-general}
\end{algorithm}

\subsection{Putting everything together} \label{ssec:together}

In Sections \ref{sec:ak-polys} and \ref{sec:akcomp}, we present an 
efficient approximate $\Ak$-projection oracle and  an $\Ak$-computation oracle for $\mathcal{P}_\pdeg$, respectively.
We show that:
\begin{theorem}
\label{thm:akprojoracle}
Fix $J \subseteq [-1, 1 ]$ and $\eta > 0$.
For all $k \leq \pdeg$,
there is an $\eta$-approximate $\Ak$-projection oracle for $\mathcal{P}_\pdeg$ which runs in time
\[ O\left( \left(\pdeg^3 \log \log 1 / \eta + s \pdeg^2 + \pdeg^{\omega + 2}\right) \log^2 \frac{1}{\eta}\right) .\]
where $s$ is the number of samples in the interval $J$.
\end{theorem}

\begin{theorem}
\label{thm:akcomporacle}
There is an $\eta$-approximate $\Ak$-computation oracle for $\mathcal{P}_d$ which runs in time $O((s + d) \log^2 (s + d))$
where $s$ is the number of samples in the interval $J$.
\end{theorem}

The algorithm \textsc{GeneralMerging}, when used in conjunction with the oracles
$\oracle_p$ and $\oracle_c$ given in
Theorems~\ref{thm:akprojoracle} and \ref{thm:akcomporacle} (for $\eta = \eps$), yields Theorem~\ref{thm:main-intro}. 
For this choice of oracles we have that $R_p (n) + R_c (n) = O(n \pdeg^{\omega + 2} \log^3 1 / \epsilon )$. 
This completes the proof.

\section{\texorpdfstring{A fast $\Ak$-projection oracle for polynomials}{A fast Ak-projection oracle for polynomials}}
\label{sec:ak-polys}
We now turn our attention to the $\Ak$-projection problem, which appears as the main subroutine in the general merging algorithm (see Section \ref{sec:generalmerging}).
In this section, we let $E \subset J$ be the set of samples drawn from the unknown distribution.
To emphasize the dependence of the empirical distribution on $E$, we denote the empirical distribution by $\fhat_E$ in this section.
Given an interval $J = [a, b]$ and a set of samples $E \subset J$, the goal of the \Ak-projection oracle is to find a hypothesis $h \in \mathcal{D}$ such that the \Ak-distance between the empirical distribution $\fhat_E$ and the hypothesis $h$ is minimized.
In contrast to the merging algorithm, the \Ak-projection oracle depends on the underlying hypothesis class $\mathcal{D}$, and here we present an efficient oracle for non-negative polynomials with fixed degree $\pdeg{}$.
In particular, our $\Ak$-projection oracle computes the coefficients $c\in\R^{\pdeg+1}$ of a degree-\pdeg{} polynomial $p_c$ that approximately minimizes the $\Ak$-distance to the given empirical distribution $\fhat_E$ in the interval $J$.
Moreover, our oracle ensures that $p_c$ is non-negative for all $x \in J$.

At a high-level, we formulate the \Ak-projection as a convex optimization problem.
A key insight is that we can construct an efficient, approximate \emph{separation oracle} for the set of polynomials that have an \Ak-distance of at most $\tau$ to the empirical distribution $\fhat_E$.
Combining this separation oracle with existing convex optimization algorithms allows us to solve the feasibility problem of checking whether we can achieve a given \Ak-distance $\tau$.
We then convert the feasibility problem to the optimization variant via a binary search over $\tau$.

In order to simplify notation, we assume that the interval $J$ is $[-1, 1]$ and that the mass of the empirical distribution $\fhat_E$ is 1.
Note that the general \Ak-projection problem can easily be converted to this special case by shifting and scaling the sample locations and weights before passing them to the \Ak-projection subroutine.
Similarly, the resulting polynomial can be transformed to the original interval and mass of the empirical distribution on this interval.\footnote{Technically, this step is actually necessary in order to avoid a running time that depends on the shape of the unknown pdf $f$.
Since the pdf $f$ could be supported on a very small interval only, the corresponding polynomial approximation could require arbitrarily large coefficients (the empirical distribution would have all samples in a very small interval).
In that case, operations such as root-finding with good precision could take an arbitrary amount of time.
In order to circumvent this issue, we make use of the real-RAM model to rescale our samples to $[-1, 1]$ before processing them further.
Combined with the assumption of unit probability mass, this allows us to bound the coefficients of candidate polynomials in the current interval.}

\subsection{The set of feasible polynomials}
For the feasibility problem, we are interested in the set of degree-\pdeg{} polynomials that have an $\Ak$-distance of at most $\tau$ to the empirical distribution $\fhat_E$ on the interval $J = [-1,1]$ and are also non-negative on $J$.
More formally, we study the following set.

\begin{definition}[Feasible polynomials]
Let $E \subset J$ be the samples of an empirical distribution with $\fhat_E(J) = 1$.
Then the set of $(\tau,\pdeg,k,E)$-feasible polynomials is
\[
  \setC_{ \tau,\pdeg,k,E} \; := \; \left\{ c \in \R^{\pdeg+1} \, | \, \norm{p_c - \fhat_E}_{\Ak, J} \leq \tau \textnormal{ and } p_c(x) \geq 0 \textnormal{ for all } x \in J \right\} \; .
\]
When $\pdeg$, $k$, and $E$ are clear from the context, we write only $\setC_\tau$ for the set of $\tau$-feasible polynomials.
\end{definition}

Considering the original \Ak-projection problem, we want to find an element $c^* \in \setC_{\tau^*}$, where $\tau^*$ is the smallest value for which $\setC_{\tau^*}$ is non-empty.
We solve a slightly relaxed version of this problem, i.e., we find an element $c$ for which the $\Ak$-constraint and the non-negativity constraint are satisfied up to small additive constants.
We then post-process the polynomial $p_c$ to make it truly non-negative while only increasing the $\Ak$-distance by a small amount.

Note that we can ``unwrap'' the definition of the $\Ak$-distance and write $\setC$ as an intersection of sets in which each set enforces the constraint $\sum_{i=1}^k \abs{p_c(I_i) - \fhat_E(I_i)} \leq \tau$ for one collection of $k$ disjoint intervals $\{I_1, \ldots, I_k\}$.
For a fixed collection of intervals, we can then write each \Ak-constraint as the intersection of \emph{linear} constraints in the space of polynomials.
Similarly, we can write the non-negativity constraint as an intersection of pointwise non-negativity constraints, which are again linear constraints in the space of polynomials.
This leads us to the following key lemma.
Note that convexity of $\setC_\tau$ could be established more directly\footnote{Norms give rise to convex sets and the set of non-negative polynomials is also convex.}, but considering $\setC_\tau$ as an intersection of halfspaces illustrates the further development of our algorithm (see also the comments after the lemma).

\begin{lemma}[Convexity]
\label{lem:feasibleconvex}
The set of $\tau$-feasible polynomials is convex.
\end{lemma}
\begin{proof}
From the definitions of $\setC_\tau$ and the $\Ak$-distance, we have
\begin{align*}
  \setC_\tau \; &= \; \{ c \in \R^{\pdeg+1} \, | \, \norm{p_c - \fhat_E}_{\Ak, J} \leq \tau \textnormal{ and } p_c(x) \geq 0 \textnormal{ for all } x \in J \} \\
    &= \; \{ c \in \R^{\pdeg+1} \, | \, \sup_{\setI \in \familyI^k_J} \sum_{I \in \setI} \abs{p_c(I) - \fhat_E(I)} \leq \tau \} \; \cap \; \{c \in \R^{\pdeg+1} \, | \, p_c(x) \geq 0 \textnormal{ for all } x \in J \} \\
    &= \; \bigcap_{\setI \in \familyI^k_J} \{ c \in \R^{\pdeg+1} \, | \, \sum_{I \in \setI} \abs{p_c(I) - \fhat_E(I)} \leq \tau \} \; \cap \; \bigcap_{x \in J} \{c \in \R^{\pdeg+1} \, | \, p_c(x) \geq 0  \} \\
    &= \; \bigcap_{\setI \in \familyI^k_J} \; \bigcap_{\xi \in \{-1, 1\}^k} \{ c \in \R^{\pdeg+1} \, | \, \sum_{i = 1}^k \xi_i (p_c(I_i) - \fhat_E(I_i)) \leq \tau \} \;\; \cap \;\; \bigcap_{x \in J} \{c \in \R^{\pdeg+1} \, | \, p_c(x) \geq 0  \} \; .
\end{align*}
In the last line, we used the notation $\setI = \{ I_1, \ldots, I_k\}$.
Since the intersection of a family of convex sets is convex, it remains to show that the individual $\Ak$-distance sets and non-negativity sets are convex.
Let
\begin{align*}
    \mathcal{M} \; &= \; \bigcap_{\setI \in \familyI^k_J} \; \bigcap_{\xi \in \{-1, 1\}^k} \{ c \in \R^{\pdeg+1} \, | \, \sum_{i = 1}^k \xi_i (p_c(I_i) - \fhat_E(I_i)) \leq \tau \} \\
    \mathcal{N} \; &= \; \bigcap_{x \in J} \{c \in \R^{\pdeg+1} \, | \, p_c(x) \geq 0  \} \; .
\end{align*}

We start with the non-negativity constraints encoding the set $\mathcal{N}$.
For a fixed $x \in J$, we can expand the constraint $p_c(x) \geq 0$ as
\[
  \sum_{i=0}^\pdeg c_i \cdot x^i \, \geq \, 0 \; ,
\]
which is clearly a linear constraint on the $c_i$.
Hence, the set $\{c \in \R^{\pdeg+1} \, | \, p_c(x) \geq 0  \}$ is a halfspace for a fixed $x$ and thus also convex.

Next, we consider the $\Ak$-constraints $\sum_{i = 1}^k \xi_i (p_c(I_i) - \fhat_E(I_i)) \leq \tau$ for the set $\mathcal{M}$.
Since the intervals $I_1, \ldots, I_k$ are now fixed, so is $\fhat_E(I_i)$.
Let $\alpha_i$ and $\beta_i$ be the endpoints of the interval $I_i$, i.e., $I_i = [\alpha_i, \beta_i]$.
Then we have
\[
  p_c(I_i) \; = \; \int_{\alpha_i}^{\beta_i} p_c(x) \diff x \; = \; P_c(\beta_i) - P_c(\alpha_i) \; ,
\]
where $P_c(x)$ is the indefinite integral of $P_c(x)$, i.e.,
\[
  P_c(x) \; = \; \sum_{i=0}^\pdeg c_i \cdot \frac{x^{i+1}}{i+1} \; .
\]
So for a fixed $x$, $P_c(x)$ is a linear combination of the $c_i$.
Consequently $\sum_{i=1}^k \xi_i p_c(I_i)$ is also a linear combination of the $c_i$, and hence each set in the intersection defining $\mathcal{M}$ is a halfspace.
This shows that $\setC_\tau$ is a convex set.
\end{proof}
It is worth noting that the set $\mathcal{N}$ is a spectrahedron (the
feasible set of a semidefinite program) because it encodes
non-negativity of a univariate polynomial over a fixed interval. 
After restricting the set of coefficients to non-negative polynomials,
we can simplify the definition of the \Ak-distance: it suffices to
consider sets of intervals with endpoints at the locations of samples
(see Lemma~\ref{lem:akdiscrete}). 
Hence, we can replace the supremum in the definition of $\mathcal{M}$
by a maximum over a finite set, which shows that $\setC_\tau$ is also
a spectrahedron. 
This suggests that the \Ak-projection problem could be solved by a black-box application of an SDP solver.
However, this would lead to a running time that is \emph{exponential}
in $k$ because there are more than $\binom{\abs{E}}{2k}$ possible sets
of intervals. 
While the authors of [CDSS14] introduce an encoding of the
\Ak-constraint with fewer linear inequalities, their approach
increases the number of variables in the optimization problem to
depend polynomially on $\frac{1}{\eps}$, which leads to a super-linear
running time. 

Instead of using black-box SDP or LP solvers, we construct an
algorithm that exploits additional structure in the \Ak-projection
problem. 
Most importantly, our algorithm separates the dimension of the desired
degree-\pdeg{} polynomial from the number of samples (or equivalently,
the error parameter $\eps$). 
This allows us to achieve a running time that is \emph{nearly-linear} for a wide range of distributions.
Interestingly, we can solve our SDP significantly faster than the LP
which has been proposed in [CDSS14] for the same problem. 
 
\subsection{Separation oracles and approximately feasible polynomials}
In order to work with the large number of \Ak-constraints efficiently, we ``hide'' this complexity from the convex optimization procedure by providing access to the constraints only through a separation oracle.
As we will see in Section \ref{sec:seporacle}, we can utilize the structure of the \Ak-norm and implement such a separation oracle for the \Ak-constraints in nearly-linear time.
Before we give the details of our separation oracle, we first show how we can solve the \Ak-projection problem assuming that we have such an oracle.
We start by formally defining our notions of separation oracles.

\begin{definition}[Separation oracle]
\label{def:sep-oracle}
A separation oracle \oracle{} for the convex set $\setC_\tau$ is a function that takes as input a coefficient vector $c \in \R^{\pdeg + 1}$ and returns one of the following two results:
\begin{enumerate}
\item ``yes'' if $c \in \setC_\tau$.
\item a separating hyperplane $y \in \R^{\pdeg + 1}$.
The hyperplane $y$ must satisfy $y^T c' \leq y^T c$ for all $c' \in \setC_\tau$.
\end{enumerate}
\end{definition}

For general polynomials, it is not possible to perform basic operations such as root finding \emph{exactly}, and hence we have to resort to approximate methods.
This motivates the following definition of an \emph{approximate} separation oracle.
While an approximate separation oracle might accept a point $c$ that is not in the set $\setC_\tau$, the point $c$ is then guaranteed to be close to $C_\tau$.
\begin{definition}[Approximate separation oracle]
\label{def:app-sep-oracle}
A $\mu$-approximate separation oracle \oracle{} for the set $\setC_\tau = \setC_{\tau, \pdeg,k,E}$ is a function that takes as input a coefficient vector $c \in \R^{\pdeg + 1}$ and returns one of the following two results, either ``yes'' or a hyperplane $y \in \R^{\pdeg + 1}$.
\begin{enumerate}
\item  If \oracle{} returns ``yes'', then $\norm{p_c - \fhat_E}_{\Ak,J} \leq \tau + 2 \mu$ and $p_c(x) \geq -\mu$ for all $x \in J$.
\item If \oracle{} returns a hyperplane, then $y$ is a separating hyperplane; i.e.
the hyperplane $y$ must satisfy $y^T c' \leq y^T c$ for all $c' \in \setC_\tau$.
\end{enumerate}
In the first case, we say that $p_c$ is a $2\mu$-approximately feasible polynomial.
\end{definition}

Note that in our definition, separating hyperplanes must still be exact for the set $\setC_\tau$.
Although our membership test is only approximate, the exact hyperplanes allow us to employ several existing separation oracle methods for convex optimization.
We now formally show that many existing methods still provide approximation guarantees when used with our approximate separation oracle.

\begin{definition}[Separation Oracle Method]
A separation oracle method (\som{}) is an algorithm with the following guarantee:
let $\setC$ be a convex set that is contained in a ball of radius $2^L$.
Moreover, let $\oracle$ be a separation oracle for the set $\setC$.
Then $\som(\oracle, L)$ returns one of the following two results:
\begin{enumerate}
\item a point $x \in \setC$.
\item ``no'' if $\setC$ does not contain a ball of radius $2^{-L}$.
\end{enumerate}

We say that an \som{} is \emph{canonical} if it interacts with the separation oracle in the following way: the first time the separation oracle returns ``yes'' for the current query point $x$, the \som{} returns the point $x$ as its final answer.
\end{definition}

There are several algorithms satisfying this definition of a separation oracle method, e.g., the classical Ellipsoid method~\cite{Khachiyan:79} and Vaidya's cutting plane method ~\cite{Vaidya:89}.
Moreover, all of these algorithms also satisfy our notion of a \emph{canonical} separation oracle method.
We require this technical condition in order to prove that our approximate separation oracles suffice.
In particular, by a straightforward simulation argument, we have the following:
\begin{theorem}
\label{thm:approxsom}
Let $\mathcal{M}$ be a canonical separation oracle method, and let $\oracle$ be a $\mu$-approximate separation oracle for the set $\setC_\tau = \setC_{\tau, \pdeg, k, E}$.
Moreover, let $L$ be such that $\setC_\tau$ is contained in a ball of radius $2^L$.
Then $\mathcal{M}(\oracle, L)$ returns one of the following two results:
\begin{enumerate}
\item a coefficient vector $c \in \R^{\pdeg+1}$ such that $\norm{p_c - \fhat_E}_{\Ak,J} \leq \tau + 2\mu$ and $p_c(x) \geq -\mu$ for all $x \in J$.
\item ``no'' if $\setC$ does not contain a ball of radius $2^{-L}$.
\end{enumerate}
\end{theorem}

\subsection{Bounds on the radii of enclosing and enclosed balls}
\label{subsec:volbound}
In order to bound the running time of the separation oracle method, we
establish bounds on the ball radii used in Theorem
\ref{thm:approxsom}. 

\paragraph{Upper bound}
When we initialize the separation oracle method, we need a ball of radius $2^L$ that contains the set $\setC_\tau$.
For this, we require bounds on the coefficients of polynomials which are bounded in $L_1$ norm.
Bounds of this form were first established by Markov~\cite{markov1892functions}.

\begin{lemma}
\label{lem:polybound2}
Let $p_c$ be a degree-\pdeg{} polynomial with coefficients $c \in \R^{\pdeg+1}$ such that $p(x) \geq 0$ for $x \in [-1, 1]$ and $\int_{-1}^1 p(x) \diff x \leq \alpha$, where $\alpha > 0$.
Then we have
\[
  \abs{c_i} \leq \alpha \cdot (\pdeg + 1)^2 \cdot (\sqrt{2} + 1)^\pdeg \quad \text{for all } i = 0, \ldots, \pdeg \; .
\]
\end{lemma}
\noindent This lemma is well-known, but for completeness, we include a proof  in Appendix \ref{app:proof-polynomials}. 
Using this lemma, we obtain:

\begin{theorem}[Upper radius bound]
\label{thm:volupper}
Let $\tau \leq 1$ and let $A$ be the $(\pdeg+1)$-ball of radius $r = 2^{L_u}$ where
\[
  L_u = \pdeg \log(\sqrt{2} + 1) + \frac{3}{2} \log \pdeg + 2 \; .
\]
Then $\setC_{\tau, \pdeg, k, E} \subseteq A$.
\end{theorem}
\begin{proof}
Let $c \in \setC_{\tau, \pdeg, k, E}$.
From basic properties of the $L_1$- and $\Ak$-norms, we have
\[
  \int_{-1}^1 p_c \diff x \; = \; \norm{p_c}_{1,J} \; = \; \norm{p_c}_{\Ak,J} \; \leq \; \norm{\fhat_E}_{\Ak, J} + \norm{p_c - \fhat_E}_{\Ak, J} \; \leq \; 1 + \tau \; \leq \; 2 \; .
\]
Since $p_c$ is also non-negative on $J$, we can apply Lemma \ref{lem:polybound2} and get
\[
  \abs{c_i} \leq 2 \cdot (\pdeg + 1) \cdot (\sqrt{2} + 1)^\pdeg \quad \text{for all } i = 0, \ldots, \pdeg \; .
\]

Note that the above constraints define a hypercube $B$ with side length $s = 4 \cdot (\pdeg+1) \cdot (\sqrt{2}+1)^\pdeg$.
The ball $A$ contains the hypercube $B$ because $r = \sqrt{\pdeg + 1} \cdot s$ is the length of the longest diagonal of $B$.
This implies that $C_{\tau, \pdeg, k, E} \subseteq B \subseteq A$.
\end{proof}

\paragraph{Lower bound}
Separation oracle methods typically cannot directly certify that a convex set is empty.
Instead, they reduce the volume of a set enclosing the feasible region until it reaches a certain threshold.
We now establish a lower bound on volumes of sets $\setC_{\tau + \eta}$ that are feasible by at least a margin $\eta$ in the $\Ak$-distance.
If the separation oracle method cannot find a small ball in $\setC_{\tau + \eta}$, we can conclude that achieving an $\Ak$-distance of $\tau$ is infeasible.

\begin{theorem}[Lower radius bound]
\label{thm:vollower}
Let $\eta > 0$ and let $\tau$ be such that $\setC_\tau = \setC_{\tau,\pdeg,k,E}$ is non-empty.
Then $\setC_{\tau + \eta}$ contains a ball of radius $r = 2^{-L_\ell}$, where
\[
  L_\ell \; = \; \log\frac{4 (\pdeg+1)}{\eta} \; .
\]
\end{theorem}
\begin{proof}
Let $c^*$ be the coefficients of a feasible polynomial, i.e., $c^* \in \setC_\tau$.
Moreover, let $c$ be such that
\[
c_i = \begin{cases} c^*_0 + \frac{\eta}{4} & \text{if } i = 0 \\
c^*_i & \text{otherwise}
\end{cases} \; .
\]
Since $p_{c^*}$ is non-negative on $J$, we also have $p_c(x) \geq \frac{\eta}{4}$ for all $x \in J$.
Moreover, it is easy to see that shifting the polynomial $p_{c^*}$ by $\frac{\eta}{4}$ changes the $\Ak$-distance to $\fhat_E$ by at most $\frac{\eta}{2}$ because the interval $J$ has length 2.
Hence, $\norm{p_c - \fhat_E}_{\Ak,J} \leq \tau + \frac{\eta}{2}$ and so $c \in \setC_{\tau + \eta}$.
We now show that we can perturb the coefficients of $c$ slightly and still stay in the set of feasible polynomials $\setC_{\tau + \eta}$.

Let $\nu = \frac{\eta}{4 (\pdeg + 1)}$ and consider the hypercube
\[
  B = \{ c' \in \R^{\pdeg+1} \, | \, c'_i \in [c_i - \nu, \, c_i + \nu] \textnormal{ for all } i \}  \; .
\]
Note that $B$ contains a ball of radius $\nu = 2^{-L_\ell}$.
First, we show that $p_{c'}(x) \geq 0$ for all $x \in J$ and $c' \in B$.
We have
\begin{align*}
p_{c'}(x) \; &= \; \sum_{i=0}^\pdeg c_i' x_i \\
        &= \; \sum_{i=0}^\pdeg c_i x^i \, + \, \sum_{i=0}^\pdeg (c_i' - c_i) x^i \\
        &\geq \; p_c(x) - \sum_{i=0}^\pdeg \nu \abs{x^i} \\
        &\geq \; \frac{\eta}{4} - (\pdeg +1) \cdot \nu \\
        &\geq \; 0 \; .
\end{align*}

Next, we turn our attention to the $\Ak$-distance constraint.
In order to show that $p_{c'}$ also achieves a good $\Ak$-distance, we bound the $L_1$-distance to $p_c$.
\begin{align*}
  \norm{p_c(x) - p_{c'}(x)}_{1,J} \; &= \; \int_{-1}^1 \abs{p_c(x) - p_{c'}(x)} \diff x \\
    &\leq \; \int_{-1}^1 \sum_{i=0}^\pdeg \nu \cdot \abs{x^i} \diff x \\
    &\leq \; \int_{-1}^1 (\pdeg+1) \nu \diff x \\
    &= \; 2 (\pdeg+1) \nu \\
    &\leq \; \frac{\eta}{2} \; .
\end{align*}
Therefore, we get
\begin{align*}
  \norm{p_{c'} - \fhat_E}_{\Ak,J} \; &\leq \; \norm{p_{c'} - p_c}_{\Ak,J} + \norm{p_c - \fhat_E}_{\Ak,J} \\
      &\leq \; \norm{p_{c'} - p_c}_{1,J} + \tau + \frac{\eta}{2} \\
      &\leq \; \tau + \eta \; .
\end{align*}
This proves that $c' \in \setC_{\tau + \eta}$ and hence $B \subseteq \setC_{\tau + \eta}$.
\end{proof}

\subsection{Finding the best polynomial}
We now relate the feasibility problem to our original optimization problem of finding a non-negative polynomial with minimal $\Ak$-distance.
For this, we perform a binary search over the $\Ak$-distance and choose our error parameters carefully in order to achieve the desired approximation guarantee.
See Algorithm \ref{alg:findpoly} for the corresponding pseudocode.

\begin{algorithm}
\begin{algorithmic}[1]
\Function{FindPolynomial}{$\pdeg, k, E, \eta$}
  \LineComment{Initial definitions}
  \State Let $\eta' = \frac{\eta}{15}$.
  \State Let $L_u = \pdeg \log(\sqrt{2} + 1) + \frac{3}{2} \log \pdeg + 2$.
  \State Let $L_\ell = \log\frac{4(\pdeg+1)}{2\eta'}$.
  \State Let $L = \max(L_u, L_\ell)$.
  \State Let $\mathcal{M}$ be a canonical separation oracle method.
  \State Let $\oracle_\tau$ be an $\eta'$-approximate separation oracle
  \Statex for the set of $(\tau, \pdeg, k, E)$-feasible polynomials.
  \vspace{.3cm}
  \State $\tau_\ell \gets 0$
  \State $\tau_u \gets 1$
  \While{$\tau_u - \tau_\ell \geq \eta'$}
    \State $\tau_m \gets \frac{\tau_\ell + \tau_u}{2}$
    \State $\tau_m' \gets \tau_m + 2\eta'$
    \If{$\mathcal{M}(\oracle_{\tau_m'}, L)$ returned a point}
      \State $\tau_u \gets \tau_m$
    \Else
      \State $\tau_\ell \gets \tau_m$ \Comment{$\setC_{\tau_{m'}}$ does not contain a ball of radius $2^{-L}$ and hence $\setC_{\tau_m}$ is empty.}
    \EndIf
  \EndWhile
  \State $c' \gets \mathcal{M}(\oracle_{\tau_u + 10\eta'}, L)$ \Comment{Find final coefficients.} \label{line:lastsom}
  \State $c_0 \gets c'_0 + \eta' \quad\text{and}\quad c_i \gets c'_i$ for $i \neq 0$\Comment{Ensure non-negativity.}
  \State \textbf{return} $c$
\EndFunction
\end{algorithmic}
\caption{Finding polynomials with small $\Ak$-distance.}
\label{alg:findpoly}
\end{algorithm}

The main result for our $\Ak$-oracle is the following:

\begin{theorem}
Let $\eta > 0$ and
let $\tau^*$ be the smallest $\Ak$-distance to the empirical distribution $\fhat_E$ achievable with a non-negative degree-$\pdeg$ polynomial on the interval $J$, i.e., $\tau^* = \min_{h \in \classP_{J,\pdeg}} \norm{h - \fhat_E}_{\Ak,J}$.
Then \textsc{FindPolynomial} returns a coefficient vector $c \in \R^{\pdeg+1}$ such that $p_c(x) \geq 0$ for all $x \in J$ and $\norm{p_c - \fhat_E}_{\Ak,J} \leq \tau^* + \eta$.
\end{theorem}
\begin{proof}
We use the definitions in Algorithm \ref{alg:findpoly}.
Note that $\tau^*$ is the smallest value for which $\setC_{\tau^*} = \setC_{\tau^*,\pdeg,k,E}$ is non-empty.
First, we show that the binary search maintains the following invariants: $\tau_\ell \leq \tau^*$ and there exists a $4\eta'$-approximately $\tau_u$-feasible polynomial.
This is clearly true at the beginning of the algorithm:
(i) Trivially, $\tau^* \geq 0 = \tau_\ell$.
(ii) For $c = (0, 0, \cdots, 0)^T$, we have $\norm{p_c - \fhat_E}_{\Ak,J} \leq 1 = \tau_u$ and $p_c(x) \geq 0$, so $p_c$ is $\tau_u$-feasible (and hence also approximately $\tau_u$-feasible).

Next, we consider the two cases in the while-loop:
\begin{enumerate}
\item If the separation oracle method returns a coefficient vector $c$ such that the polynomial $p_c$ is $2\eta'$-approximately $\tau_m'$-feasible, then $p_c$ is also $4 \eta'$-approximately $\tau_m$-feasible because $\tau_m' = \tau_m + 2\eta'$.
Hence, the update of $\tau_u$ preserves the loop invariant.
\item If the separation oracle method returns that $\setC_{\tau_m'}$ does not contain a ball of radius $2^{-L}$, then $\tau_m$ must be empty (by the contrapositive of Theorem \ref{thm:vollower}).
Hence, we have $\tau^* \geq \tau_m$ and the update of $\tau_\ell$ preserves the loop invariant.
\end{enumerate}

We now analyze the final stage of \textsc{FindPolynomial} after the while-loop.
First, we show that $\setC_{\tau_u + 8\eta'}$ is non-empty by identifying a point in the set.
From the loop invariant, we know that there is a coefficient vector $v'$ such that $p_{v'}$ is a $4 \eta'$-approximately $\tau_u$-feasible polynomial.
Consider $v$ with $v_0 := v'_0 + 2\eta'$ and $v_i := v'_i$ for $i \neq 0$.
Then we have
\[
  \norm{p_v - p_{v'}}_{1,J} \; = \; \int_{-1}^1 \abs{p_v(x) - p_{v'}(x)} \diff x \; = \; \int_{-1}^1 2 \eta' \diff x \; = \; 4 \eta' \; .
\]
Hence, we also get
\[
  \norm{p_v - \fhat_E}_{\Ak,J} \; \stackrel{(a)}{\leq} \; \norm{p_v - p_{v'}}_{\Ak, J} + \norm{p_{v'} - \fhat_E}_{\Ak,J} \; \stackrel{(b)}{\leq} \; \norm{p_v - p_{v'}}_{1,J} + \tau_u + 4 \eta' \; \leq \; \tau_u + 8\eta' \; .
\]
We used the triangle inequality in (a) and the fact that $p_{v'}$ is $4\eta'$-approximately $\tau_u$-feasible in (b).
Moreover, we have $p_{v'}(x) \geq -2\eta'$ for all $x \in J$ and thus $p_v(x) \geq 0$ for all $x \in J$.
This shows that $\setC_{\tau_u + 8\eta'}$ is non-empty because $v \in \setC_{\tau_u + 8 \eta'}$.

Finally, consider the last run of the separation oracle method in line \ref{line:lastsom} of Algorithm \ref{alg:findpoly}.
Since $\setC_{\tau_u + 8\eta'}$ is non-empty, Theorem \ref{thm:vollower} shows that $\setC_{\tau_u + 10\eta'}$ contains a ball of radius $2^{-L}$.
Hence, the separation oracle method must return a coefficient vector $c' \in \R^{\pdeg+1}$ such that $p_{c'}$ is $2\eta'$-approximately $\tau_u+10\eta'$-feasible.
Using a similar argument as for $v$, we can make $p_{c'}$ non-negative while increasing its $\Ak$-distance to $\fhat_E$ by only $2\eta'$, i.e., we can show that $p_c(x) \geq 0$ for all $x \in J$ and that
\[
  \norm{p_c - \fhat_E}_{\Ak, J} \leq \tau_u + 14\eta' \; .
\]
Since $\tau_u - \tau_\ell \leq \eta'$ and $\tau_\ell \leq \tau^*$, we have $\tau_u \leq \tau^* + \eta'$.
Therefore, $\tau_u + 14\eta' \leq \tau^* + 15 \eta' = \tau^* + \eta$, which gives the desired bound on $\norm{p_c - \fhat_E}_{\Ak,J}$.
\end{proof}

In order to state a concrete running time, we instantiate our algorithm \textsc{FindPolynomial} with Vaidya's cutting plane method as the separation oracle method.
In particular, Vaidya's algorithm runs in time $O(T \pdeg L + \pdeg^{\omega + 1} L)$ for a feasibility problem in dimension $\pdeg$ and ball radii bounds of $2^L$ and $2^{-L}$, respectively.
$T$ is the cost of a single call to the separation oracle and $\omega$ is the matrix-multiplication constant.
Then we get:
\begin{theorem}
Let $\oracle$ be an $\frac{\eta}{14}$-approximate separation oracle that runs in time $T$.
Then \textsc{FindPolynomial} has time complexity $O((T \pdeg^2 + \pdeg^{\omega + 2}) \log^2 \frac{1}{\eta})$.
\end{theorem}
\begin{proof}
The running time of \textsc{FindPolynomial} is dominated by the binary search.
It is easy to see that the binary search performs $O(\log \frac{1}{\eta})$ iterations, in which the main operation is the call to the separation oracle method.
Our bounds on the ball radii in Theorems \ref{thm:volupper} and \ref{thm:vollower} imply $L = O(\pdeg + \log\frac{1}{\eta})$.
Combining this with the running time bound for Vaidya's algorithm gives the time complexity stated in the theorem.
\end{proof}

In Section \ref{sec:seporacle} we describe a $\mu$-approximate separation oracle that runs in time $\Otilde(\pdeg k + \pdeg \log \log 1 / \mu + s)$, 
where $s$ is the number of samples in the empirical distribution on the interval $J$.
Plugging this oracle directly into our algorithm \textsc{FindPolynomial} gives an 
$\eta$-approximate $\Ak$-projection oracle which runs in time $O((\pdeg^3 k + \pdeg^3 \log \log 1 / \eta + s \pdeg^2 + \pdeg^{\omega + 2}) \log^2 \frac{1}{\eta})$.
This algorithm is the algorithm promised in Theorem \ref{thm:akprojoracle}.


\section{The separation oracle and the \texorpdfstring{$\Ak$}{Ak}-computation oracle}
\label{sec:seporacle}
In this section, we construct an efficient approximate separation
oracle (see Definition~\ref{def:app-sep-oracle}) for the set $C_\tau$ over the interval $J = [-1, 1]$. We denote our algorithm by \textsc{ApproxSepOracle}. 
Let $A$ be the ball defined in Lemma \ref{thm:volupper}.
We will show:
\begin{theorem}
For all $\mu > 0$, $\textsc{ApproxSepOracle} (c, \mu)$ is
a $\mu$-approximate separation oracle for $C_\tau$ that runs in time
$\Otilde (\pdeg k + \pdeg \log \log \frac1\mu + s)$, where $s$ the number of
samples in $J$, assuming all queries are contained in the ball $A$.
\end{theorem}
Along the way we also develop an approximate $\Ak$-computation oracle
$\textsc{ComputeAk}$.

\subsection{Overview of \textsc{ApproxSepOracle}}

\textsc{ApproxSepOracle} consists of two parts,
\textsc{TestNonnegBounded} and \textsc{AkSeparator}. 
We show:
\begin{lemma}
\label{lem:nonnegbounded}
For any $\tau \leq 2$, given a set polynomial coefficients $c \in A \subset \R^{\pdeg + 1}$, the algorithm
\textsc{Test\-Nonneg\-Bounded}$(c, \mu)$ runs in time $O (\pdeg \log^2 \pdeg
(\log^2 \pdeg + \log \log 1 / \mu))$ and outputs a separating hyperplane
for $\mathcal{C}_\tau$ or ``yes''.
Moreover, if there exists a point $x
\in [-1, 1]$ such that $p_c (x) < - \mu$, the output is always a separating hyperplane.
\end{lemma}
\noindent We show in the next section that whenever
$c\notin\mathcal{C}_\tau$ the output is ``$yes$''. 

\begin{theorem}
\label{thm:akseparator}
Given a set of polynomial coefficients $c \in A \subset \R^{\pdeg + 1}$ such that
$p_c (x) \geq -\mu$ for all $x \in [-1, 1]$, there is an algorithm
\textsc{AkSeparator}$(c, \mu)$ that runs in time $O(d k + (s + d) \log^2 (s + d))$ and either outputs a separating hyperplane for $c$ from $\setC_\tau$ or
returns  ``$yes$''.
Moreover, if $\norm{p_c - \fhat_E}_{\Ak} > \tau + 2 \mu$, the output is always a separating hyperplane.
\end{theorem}

\paragraph{\textsc{ApproxSepOracle} given
  \textsc{TestNonnegBounded} and \textsc{AkSeparator}} 
Given \textsc{TestNonnegBounded} and \textsc{AkSeparator}, it is
straightforward to design \textsc{ApproxSepOracle}. 

We first run \textsc{TestNonnegBounded}$(c, \mu)$. If
it outputs a separating hyperplane, we return the
hyperplane. Otherwise, we run \textsc{AkSeparator}$(c, \mu)$, and
again if it outputs a separating hyperplane, we return it. 
If none of these happen, we return ``$yes$''. 
Lemma~\ref{lem:nonnegbounded} and Theorem~\ref{thm:akseparator} imply
that \textsc{ApproxSepOracle} is correct and runs in the claimed time:
\[
O (\pdeg \log^2 \pdeg (\log^2 \pdeg + \log \log 1 /\mu)) \; + \; O(d k + (s + d) \log^2 (s + d)) \; = \; \Otilde (\pdeg k + \pdeg \log \log 1 / \mu + s) \; .
\]

In the following sections, we prove Lemma~\ref{lem:nonnegbounded} and
Theorem~\ref{thm:akseparator}.  
In Section~\ref{sec:nonneg} we describe \textsc{TestNonnegBounded}
and prove Lemma~\ref{lem:nonnegbounded}, and in
Section~\ref{sec:akcomp} we describe \textsc{AkSeparator} and prove
Theorem~\ref{thm:akseparator}.

\subsection{Testing non-negativity and boundedness}
\label{sec:nonneg}
Formally, the problem we solve here is the following testing problem:

\begin{definition}[Approximate non-negativity test]
\label{prob:nonneg}
An approximate non-negativity tester is an algorithm satisfying the following guarantee.
Given a polynomial $p = \sum_{i = 0}^\pdeg c_i x^i$ with $\max_i \abs{c_i} \leq \alpha$ and a parameter $\mu > 0$, return one of two results:
\begin{itemize}
\item a point $x \in [-1, 1]$ at which $p(x) < -\mu / 2$.
\item ``OK''.
\end{itemize}
Moreover, it must return the first if there exists a point $x' \in [-1, 1]$ so that $p(x') < -\mu$.
\end{definition}

Building upon the classical polynomial root-finding results of \cite{Pan2001}, we show:
\begin{theorem}
Consider $p$ and $\mu$ from Definition~\ref{prob:nonneg}.
Then there exists an algorithm \textsc{TestNonneg}$(p, \mu)$ that is an approximate non-negativity tester and runs in time
\[
  O(\pdeg \log^2 \pdeg \cdot (\log^2 \pdeg + \log \log \alpha + \log \log (1 / \mu))) \; ,
\]
where $\alpha$ is a bound on the coefficients of $p$.
\end{theorem}

This theorem is proved in Section~\ref{app-non-neg}. 

We have a bound on the coefficients $c$ since we may assume that $c \in A$, and so we can use this algorithm to efficiently test non-negativity as we require.
Our algorithm \textsc{TestNonnegBounded} simply runs $\textsc{TestNonneg} (p_c, \mu)$. 
If this returns ''$yes$'',
then \textsc{TestNonnegBounded} outputs ''$yes$''. 
Otherwise, $\textsc{TestNonneg}(p_c, \mu)$ outputs a point $x \in [-1,
1]$ such that $p_c (x) \leq -\mu / 2$. 
In that case, \textsc{TestNonnegBounded} returns the hyperplane defined by $y = -(1, x, x^2, \ldots, x^\pdeg)^T$, i.e., $p_c(x) = -y^T c$.
Note that for all $c' \in \setC_\tau$ we have $p_{c'}(y) \geq 0$ and hence $y^T c' \leq 0$.
This shows that
\[
  y^T c' \; \leq \; 0 \; < \; \frac{\mu}{2} \; \leq \; -p_c(x) \; = \; y^T c
\]
as desired.

\begin{proof}[Proof of Lemma~\ref{lem:nonnegbounded}]
The correctness of this algorithm follows from
 the correctness of 
$\textsc{TestNonneg}$. 
We therefore only bound the running time. 
The worst-case runtime of this algorithm is exactly
the runtime of $\textsc{TestNonneg} (p_c, \mu)$ for any $c \in A$.
Since we run $\textsc{TestNonneg} (p_c, \mu)$ only when 
$\max_{i \in [\pdeg]} |c_i| \leq 2^{L_u} = 2^{O(\pdeg)}$ (see Theorem \ref{thm:volupper}) , the runtime of
$\textsc{TestNonneg} (p_c, \mu)$ is  
\[O(\pdeg \log^2 \pdeg (\log^2 \pdeg + \log \log 1 / \mu)),\]
as claimed. 
\end{proof}

\subsection{An \texorpdfstring{$\Ak$}{Ak}-computation oracle}
\label{sec:akcomp}
We now consider the $\Ak$-distance computation between two
functions, 
one of which is a polynomial and the other an empirical distribution. 
In this subsection, we describe an algorithm \textsc{ComputeAk}, and show:

\begin{theorem}
\label{thm:computeak}
Given a polynomial $p$ such that $p(x) \geq - \mu$ for all $x \in [-1,
1]$ and an empirical distribution $\widehat{f}$ supported on $s$
points, for any $k\le \pdeg$, \textsc{ComputeAk}$(p, \fhat, k)$ runs in time $O((s  + \pdeg) \log^2 (s + \pdeg))$, and computes a value $v \in \R_+$ such that  $| v - \|
p - \widehat{f} \|_{\Ak} | \leq 2 \mu$ and a set of intervals $I_1,
\ldots, I_k$ so that \[\sum_{i = 1}^k \left| p (I_i) - \widehat{f}
  (I_i) \right| = v \; .\] 
\end{theorem}

Note that this theorem immediately implies Theorem \ref{thm:akcomporacle}.

\paragraph{\textsc{AkSeparator} given \textsc{ComputeAk}:}
Before describing \textsc{ComputeAk}, we show how to design
\textsc{AkSeparator} satisfying Theorem~\ref{thm:akseparator} given
such a subroutine \textsc{ComputeAk}.  

The algorithm \textsc{AkSeparator} is as follows: we run
\textsc{ComputeAk}$(p_c, \fhat, k)$, let $v$ be its 
estimate for $ \| p_c - \widehat{f} \|_{\Ak}$, and let $I_1, \ldots,
I_k$ be the intervals it produces. If $v \leq \tau$, we output ``yes''. 

Otherwise, suppose
\[v = \sum_{i = 1}^k |p_c (I_i) - \widehat{f} (I_i)| > \tau \; .\]
Note that if $\| p_c - \widehat{f} \|_{\Ak} > \tau + 2 \mu$, this is
guaranteed to happen since $v$ differs from $\| p_c - \widehat{f}
\|_{\Ak}$ by at most $2 \mu$. 
Let $\sigma_i = \text{sign} (p_c (I_i) - \widehat{f} (I_i))$.
Let $I_i = [a_i, b_i]$.
Then
\begin{align*}
 \sum_{i = 1}^k |p_c(I_i) - \widehat{f} (I_i)| &= \sum_{i = 1}^k \sigma_i \left( \int_{a_i}^{b_i} p_c(x) \diff x - \fhat (I_i) \right) \\
 &= \sum_{i = 1}^k \sigma_i \left( \sum_{j = 0}^\pdeg \frac{1}{j + 1}
   \left( b_i^{j + 1} - a_i^{j + 1} \right) c_j - \fhat (I_i) \right) ,
\end{align*}
and therefore, 
\begin{equation}
\label{eq:akhyperplane}
 \sum_{i = 1}^k \sigma_i  \sum_{j = 0}^\pdeg \frac{1}{j + 1} \left( b_i^{j + 1} - a_i^{j + 1} \right) c_j > \tau + \sum_{i = 1}^k \sigma_i \fhat (I_i) \; .
\end{equation}
Note that the left hand side is linear in $c$ when we fix $\sigma_i$, and this is the separating hyperplane \textsc{AkSeparator} returns in this case.

\begin{proof}[Proof of Theorem \ref{thm:akseparator} given Theorem \ref{thm:computeak}]
We first argue about the correctness of the algorithm.
If $\| p_c - \widehat{f} \|_{\Ak} \geq \tau + 2 \mu$,
then \textsc{ComputeAk} guarantees that
\[v = \sum_{i = 1}^i |p_c (I_i) - \widehat{f} (I_i)| > \tau \; .\]

Consider the hyperplane constructed in~\eqref{eq:akhyperplane}. 
For any $c' \in \setC_\tau$ 
\begin{align*}
\sum_{i = 1}^k \sigma_i  \left( \sum_{j = 0}^\pdeg \frac{1}{j + 1} \left(
    b_i^{j + 1} - a_i^{j + 1} \right) c'_j - \fhat (I_i) \right) &\leq
\sum_{i = 1}^k \left|  \sum_{j = 0}^\pdeg \frac{1}{j + 1} \left( b_i^{j +
      1} - a_i^{j + 1} \right) c'_j - \fhat(I_i) \right| \\ 
&= \sum_{j = 0}^\pdeg \left| p_{c'} (I_j) - \fhat(I_i) \right| \\
&\leq | p_{c'} - \widehat{f} \|_{\Ak} \leq \tau \; ,
\end{align*}
where the last inequality is from the definition of
$\mathcal{C}_\tau$. 
Therefore this is indeed a separating hyperplane for
$c$ and $\setC_\tau$. 
Moreover, given $I_1, \ldots, I_k$ and $v$, this separating hyperplane
can be computed in time $O(\pdeg k)$. Thus the entire algorithm runs in time  $O (\pdeg k +  (s + \pdeg) \log^2 (s + \pdeg)$ as claimed.
\end{proof}

\subsubsection{A Reduction from Continuous to Discrete}
We first show that our $\A_k$--computation problem reduces to the following discrete problem:
For a sequence of real numbers $c_1, \ldots, c_r$ and an 
interval $I=[a,b]$ in $[r]$, let $w(I)= \sum_{a \leq i \leq b} c_i$. 
We show that our problem reduces to the problem \textsc{DiscreteAk}, defined below.

\textsc{DiscreteAk}: Given a sequence of $r$ real numbers $\{c_i\}_{i = 1}^r$ and a number $k$, find a set of $k$ disjoint intervals $I_1, \ldots,
I_{k}$ that maximizes
\[
\sum_{i=1}^{k} \left| w(I_i)\right|.
\]
We will denote the maximum value obtainable $\norm{\{c_i\}}_{\Ak}$, i.e.,
\[\norm{\{c_i\}}_{\Ak} = \max_{\setI} \sum_{I \in \setI} \left| w(I)\right| \; , \]
where the $\setI$ is taken over all collections of $k$ disjoint intervals.

We will show that it is possible to reduce the continuous problem of approximately computing the $\Ak$ distance between $p$ and $\widehat{f}$ to solving 
 \textsc{DiscreteAk} for a suitably chosen sequence of length $O(\pdeg)$. 
Suppose the empirical distribution $\widehat{f}$ is supported at $s$ points 
$a < x_1 \le \ldots \le x_s\le b$ in this interval. 
Let $\mathcal{X}$ be the support of $\widehat{f}$.
Let $p[\alpha, \beta]=\int_{\alpha}^{\beta}p(x)dx$.
Consider the following sequences of length $2s+1$:  
\[E(i) = \left\{ \begin{array}{ll}
         1/n & \mbox{if $i$ is even},\\
       0 & \mbox{if $i$ is odd}.\end{array} \right. 
       ~ , ~
       P_{\disc}(i) = \left\{ \begin{array}{ll}
         p[x_\ell, x_{\ell + 1}] & \mbox{if $i = 2 \ell + 1$},\\
       0 & \mbox{if $i$ is even}.\end{array} \right.,\]
where for simplicity we let $s_0 = a$ and $s_{s + 1} = b$.
The two sequences are displayed in Table~\ref{fig:discreteAk}. 

\begin{table}[h]
\begin{center}
    \begin{tabular}{| c | c | c | c | c | c | c | c |}
    \hline
    $i$ & 1 & 2 & 3&4 &$\ldots$&$2s$&$2s+1$ \\ \hline
    $E(i)$ & 0 & $\frac1n$ & 0&$\frac1n$&$\ldots$&$\frac1n$&0 \\ \hline
    $P_{\disc}(i)$ & $p[a,x_1]$ & 0 &$p[x_1,x_2]$&0&$\ldots$&0&$p[x_s,b]$ \\ \hline
    \end{tabular}
\end{center}
\caption{The sequences $E(i)$ and $P_{\disc} (i)$.}
\label{fig:discreteAk}
\end{table}

Then we have the following lemma:
\begin{lemma}
\label{lem:akdiscrete}
For any polynomial $p$ so that $p(x) \geq - \mu$ on $[-1, 1]$
\[
\left| \norm{p-\widehat{f} }_{\Ak} - \norm{\{ P_{\disc} - E \} }_{\Ak} \right| < 2 \mu \; .
\]
Moreover, given $k$ intervals $I_1, \ldots, I_k$ maximizing $\norm{\{ P_{\disc} - E \} }_{\Ak}$, one can compute $k$ intervals $J_1, \ldots, J_k$ so that
\[\left| \sum_{i = 1}^k \left| p (J_i)-\widehat{f} (J_i) \right| - \norm{\{ P_{\disc} - E \} }_{\Ak} \right| < 2 \mu \]
in time $O(k)$.
\end{lemma}
\begin{proof}
We first show that 
\[
\norm{p-\widehat{f} }_{\Ak} \geq \norm{\{ P_{\disc} - E \} }_{\Ak} \; .
\]
Let $I_1, \ldots, I_k$ be a set of disjoint intervals in $[2\pdeg + 1]$ achieving the maximum on the RHS.
Then it suffices to demonstrate a set of $k$ disjoint intervals $J_1, \ldots, J_k$ in $I$ satisfying 
\begin{equation}
\sum_{i = 1}^k \left| p(J_i) - \widehat{f} (J_i) \right| \geq \sum_{i = 1}^k \left| P_{\disc} (I_i) - E (I_i) \right|  \;.
\label{eq:Akgeqdiscrete}
\end{equation}
We construct the $J_i$ as follows. Fix $i$, and let $I_i = [a_i, b_i]$. 
Define $J_i$ to be the interval from $a_i$ to $b_i$.
If $a_i$ is even (i.e., if $P_{\disc} (a_i) - E(i)$ has only a contribution from $-E(a_i)$), include the left endpoint of this interval from $J_i$, otherwise (i.e., if $P_{\disc} (a_i) - E(i)$ has only a contribution from $P_{\disc} (a_i)$), exclude it, and similarly for the right endpoint.
Then, by observation, we have $P_{\disc} (I_i) - E(I_i) = p(J_i) - \widehat{f}(J_i)$, and thus this choice of $J_i$ satisfies (\ref{eq:Akgeqdiscrete}), as claimed.

Now we show the other direction, i.e., that 
\[
\norm{p-\widehat{f} }_{\Ak} \leq \norm{\{ P_{\disc} - E \} }_{\Ak} + 2 \mu \; .
\label{eq:Akleqdiscrete}
\]

Let $I_1, \ldots, I_k$ denote a set of disjoint intervals in $I$ achieving the maximum value on the LHS. 
It suffices to demonstrate a set of $k$ disjoint intervals $J_1, \ldots, J_k$ in $I$ satisfying 
\begin{equation}
\sum_{i = 1}^k \left| p(I_i) - \widehat{f} (I_i) \right| \leq \sum_{i = 1}^k \left| P_{\disc} (J_i) - E (J_i) \right| + 2 \mu \;.
\end{equation}

We first claim that we may assume that the endpoints of each $I_i$ are at a point in the support of the empirical.
Let $a_i$ and $b_i$ be the left and right endpoints of $I_i$, respectively.
Cluster the intervals $I_i$ into groups, as follows: cluster any set of consecutive intervals $I_{j}, \ldots, I_{j'}$ if it is the case that $p(I_\ell) - \fhat (I_\ell) \geq 0$ for $\ell = j, \ldots, j'$, and $[b_{\ell}, a_{\ell + 1}]$ contains no points of the empirical, for $\ell = j, \ldots, j' - 1$.
Put all other intervals not clustered this way in their own group.
That is, cluster a set of consecutive intervals if and only if on all of them the contribution to the LHS is non-negative, and there are no points of the empirical between them.
Let the clustering be $\setI_1, \ldots, \setI_{k'}$, and let $J_i$ be the smallest interval containing all the intervals in $\setI_j$.
Let $c_i$ and $d_i$ denote the left and right endpoints of $J_i$, respectively.
Associate to each cluster a sign $\sigma_i \in \{-1, 1\}$ which is the (unique) sign of $p(I) - \fhat (I)$ for all $I \in \setI_j$.
Since $p(i) \geq -\mu$, this clustering has the property that for any cluster $\setI_i$, we have
\[\left| \left( \sum_{I \in \setI_i }p(I) - \fhat (I) \right) - (p(J_i) - \fhat (J_i)) \right| \leq \mu \cdot |J_i - \bigcup_{I \in \setI_j} I | \; . \]

Then, for all $i$, if $\sigma_i = 1$, take the interval $I'_i = (x_j, x_\ell)$ where $x_j$ is the largest point in $\mathcal{X}$ so that $x_j \leq c_i$, and where $x_\ell$ is the smallest point in $\mathcal{X}$ so that $x_\ell \geq d_i$.
Then since $p \geq \mu$ on $[-1, 1]$ and the new interval contains no points in the support of $\widehat{f}$ which are not in $\cup_{I \in \setI_j} I$ or $J_i$, we have 
\[p(I'_i) - \fhat (I'_i) \geq p(J_i) - \fhat (I'_i) - \mu \left| I'_i - J_i \right| \geq \left( \sum_{I \in \setI_i }p(I) - \fhat (I) \right) - \mu |J_i - \cup_{I \in \setI_j} I | - \mu \left| I'_i - J_i \right| \; .\]
Alternatively, if $\sigma_i < 0$, take the interval $I'_i = [x_j, x_\ell]$ where $x_j$ is the smallest point in $\mathcal{X}$ so that $x_j \geq c_i$ and $x_\ell$ is the largest point in $\mathcal{X}$ so that $x_\ell \leq d_i$.
By the analogous reasoning as before we have that $p(I'_i) - \fhat (I'_i) \leq p(J_i) - \fhat (J_i) + \mu |J_i|$,\footnote{Since each cluster with negative sign has exactly a single interval in the original partition, notationally we will not distinguish between $J_i$ and the one interval in the original partition in $\setI_i$, when $\sigma_i = -1$.} and therefore $|p(I'_i) - \fhat (I'_i)| + \mu |I_i| \geq |p(J_i) - \fhat (J_i)|$.
Thus,
\begin{align*}
\sum_{i = 1}^k \left| p(I_i) - \widehat{f} (I_i) \right| &\leq \sum_{i = 1}^{k'} \left( |p(I_i') - \fhat (I_i')| + \mu |J_i - \cup_{I \in \setI_j} I | + \mu \left| I'_i - J_i \right| \right) \\
& \leq \sum_{i = 1}^{k'} \left| p(I_i') - \widehat{f} (I_i') \right| + 2 \mu \; .
\end{align*}
since $\sum_{i = 1}^{k'} \left( |J_i - \cup_{I \in \setI_j} I + \mu \left| I'_i - J_i \right| \right) \leq 2$ as the intervals in the sum are disjoint subintervals in $[-1, 1]$.

Now it is straightforward to define the $J_i$. 
Namely, for each $I_i$ with endpoints $x_{i_1} \leq x_{i_2}$ so that $x_{i_1}, x_{i_2} \in \mathcal{X}$, define
\[
J_i = \left\{ \begin{array}{ll}
         \left[ i_1, i_2 \right] & \mbox{if $x_{i_1}, x_{i_2} \in \mathcal{X}$};\\
        \left[i_1 + 1, i_2 \right] & \mbox{if $x_{i_1} \not\in \mathcal{X}$ and $x_{i_2} \in \mathcal{X}$};\\
         \left[ i_1, i_2 - 1 \right] & \mbox{if $x_{i_1} \in \mathcal{X}$ and $x_{i_2} \not\in \mathcal{X}$}; \\
          \left[ i_1 + 1, i_2 - 1 \right] & \mbox{if $x_{i_1}, x_{i_2} \not\in \mathcal{X}$} \;.
         \end{array} \right.
\]
One can check that with this definition of the $J_i$, we have $p(I_i) - \widehat{f}(I_i) = P_{\disc}(J_i) - E (J_i)$; moreover, all the $J_i$ are discrete and thus this choice of $J_i$ satisfies (\ref{eq:Akleqdiscrete}).

Moreover, the transformation claimed in the lemma is the transformation provided in the first part of the argument. 
It is clear that this transformation is computable in a single pass through the intervals $I_1, \ldots, I_k$.
This completes the proof.
\end{proof}

\subsubsection{Description of \textsc{ComputeDiscreteAk}}

For the rest of this section we focus on solving \textsc{DiscreteAk}. A very similar problem was considered in~\cite{Csuros04} who showed an algorithm for the problem of computing the set of $k$ disjoint intervals $I_1, \ldots, I_k$ maximizing
\[
\left| \sum_{i=1}^{k}  w(I_i)\right|
\]
which runs in time $O(\pdeg \cdot\min\{\log \pdeg, k\})$ time. 
We require a modified version of this algorithm which we present and analyze below. 
We call our variant \textsc{ComputeDiscreteAk}.

Here is an informal description of \textsc{ComputeDiscreteAk}.
First, we may assume the original sequence is alternating in sign, as otherwise we may merge two consecutive numbers without consequence.
We start with the set of intervals  $\setI_0 = I_{0, 1} \leq \ldots \leq I_{0, r}$, where $I_{0, i} = [c_i, c_i]$ contains only the point $c_i$.
We first compute $\mathcal{J}_0$ and $m_0$, where $\mathcal{J}_0$ is the  set of $k$ intervals $I$ in $\setI_0$ with largest $|w(I)|$, and $m_0 = \sum_{I \in \mathcal{J}_0} |w(I)|$.
Iteratively, after constructing $\setI_i = \{I_{i, 1}, \ldots, I_{i, r} \}$, we construct $\setI_{i + 1}$ by finding the set $I_{i, j}$ with minimal $|w (I_{i, j})|$ amongst all intervals in $\setI_i$, and merging it with both of its neighbors (if it is the first or last interval and so only has one neighbor, instead discard it), that is, 
\[\setI_{i + 1} = \{I_{i, 1}, \ldots, I_{i, j - 2}, I_{i, j - 1} \cup I_{i, j} \cup I_{i, j + 1}, I_{i, j + 2}, \ldots, I_{i, r_i}\} \; .\]
We then compute $\mathcal{J}_{i + 1}$ and $m_{i + 1}$ where $\mathcal{J}_{i + 1}$ is the  set of $k$ intervals $I$ in $\setI_{i + 1}$ with largest $|w(I)|$, and $m_{i  +1} = \sum_{I \in \mathcal{J}_{i + 1}} |w(I)|$.
To perform these operations efficiently, we store the weights of the intervals we create in priority queues.
We repeat this process until the collection of intervals $\setI_\ell$ has $\leq k$ intervals.
We output $\mathcal{J}_i$ and $w_i$, where $w_i$ is the largest
amongst all $w_{i'}$ computed in any iteration.
An example of an iteration of the algorithm is given in
Figure~\ref{fig:iter-merging}, and 
the formal definition of the algorithm is in
Algorithm~\ref{alg:discreteak}.

\begin{figure}[h]
\begin{center}
\begin{tikzpicture}[mydrawstyle/.style={draw=black, very thick}, x=1mm, y=1mm, z=1mm]
  \draw[mydrawstyle] (-20,20) node{Iteration $i:$};
  \draw[mydrawstyle, -](0,20)--(96,20) node at (-6,20)[left]{};
  \draw[mydrawstyle](0,18)--(0,22) node{};
  \draw (5,24) node{$0.8$};
  \draw[mydrawstyle](10,18)--(10,22) node{};
  \draw (17,24) node{$-0.5$};
  \draw[mydrawstyle](25,18)--(25,22) node{};
  \draw (37,24) node{$0.4$};
  \draw[mydrawstyle](48,18)--(48,22) node{};
  \draw (54,24) node{$-0.1$};
  \draw[mydrawstyle](60,18)--(60,22) node{};
  \draw (63,24) node{$0.3$};
  \draw[mydrawstyle](66,18)--(66,22) node{};
  \draw (70,24) node{$-0.4$};
  \draw[mydrawstyle](76,18)--(76,22) node{};
  \draw (86,24) node{$0.5$};
  \draw[mydrawstyle](96,18)--(96,22) node{};
  \draw[mydrawstyle] (-17,8) node{Iteration $i+1:$};
  \draw[mydrawstyle, -](0,8)--(96,8) node at (-6,10)[left]{};
  \draw[mydrawstyle](0,6)--(0,10) node{};
  \draw (5,12) node{$0.8$};
  \draw[mydrawstyle](10,6)--(10,10) node{};
  \draw (17,12) node{$-0.5$};
  \draw[mydrawstyle](25,6)--(25,10) node{};
  \draw (44,12) node{$0.6$};
  \draw[mydrawstyle](66,6)--(66,10) node{};
  \draw (70,12) node{$-0.4$};
  \draw[mydrawstyle](76,6)--(76,10) node{};
  \draw (86,12) node{$0.5$};
  \draw[mydrawstyle](96,6)--(96,10) node{};
\end{tikzpicture}
\end{center}
\caption{An iteration of \textsc{ComputeDiscreteAk}. The numbers denote the weight of each interval. The interval with smallest weight (in absolute value) is chosen and merged with adjacent intervals. Note that if weights are of alternating signs at the start, then they are of alternating signs at each iteration.}
\label{fig:iter-merging}
\end{figure}
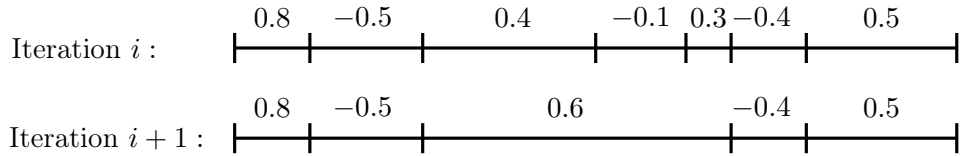

\begin{algorithm}[htb]
\begin{algorithmic}[1]
\Function{ComputeDiscreteAk}{$\{c_i\}_{i = 1}^r, k$}
\State Let $ \mathcal{I}_0 \gets \{ [c_1, c_1], [c_2, c_2], \ldots, [c_r, c_r] \}$ be the initial set of intervals.
\State Let $Q$ be an empty priority queue.

\For{$I \in \mathcal{I}_0$}
\State $Q.push(I, w(I))$
\EndFor

\State $i \gets 0$

\While{ $|\setI_i| > k$}
	\State Let $I \gets Q.deleteMin()$.
	\If{$I$ is not the leftmost or rightmost interval}
		\State Let $I_{left}$ and $I_{right}$ be its left and right neighbors, respectively.
		\State $Q.remove(I_{left})$
		\State $Q.remove(I_{right})$
		\State Let $I' = I_{left} \cup I \cup I_{right}$
		\State $Q.push(I', w(I'))$
	\EndIf
	\State $i \gets i + 1$
	\State Let $\setI_i$ be the elements of $Q$
	\State Let $\mathcal{J}_i$	be the $k$ intervals in $\setI_i$ with maximum weight
	\State Let $w_i = \sum_{I \in \mathcal{J}_i} w(I)$
\EndWhile
\State \textbf{return} $w_j$ and $\mathcal{J}_j$ where $w_j$ satisfies $w_j \geq w_i$ for all $i$.
\EndFunction
\end{algorithmic}
\caption{Computing the discrete $\Ak$ norm of a sequence.}
\label{alg:discreteak}
\end{algorithm}

The following runtime bound can be easily verified:
\begin{theorem}
Given $\{c_i\}_{i = 1}^r$, $\textsc{ComputeDiscreteAk}(\{c_i\}_{i = 1}^r, k)$ runs in time $O(r\cdot\min\{\log r, k\})$.
\end{theorem}
The nontrivial part of the analysis is correctness. 

\begin{theorem}
\label{thm:Akcomp-correctness}
Given $\{c_i\}_{i = 1}^r$ and $k$, the set of intervals returned by the algorithm \textsc{Compute\-DiscreteAk}$(\{c_i\}_{i = 1}^r, k)$ solves the problem $\textsc{DiscreteAk}$.
\end{theorem}
\begin{proof}
Our analysis follows the analysis in~\cite{Csuros04}. 
We call any $\setI^\ast$ which attains the maximum for the \textsc{DiscreteAk} problem a \emph{maximal subset}, or \emph{maximal} for short.
For any two collections of disjoint intervals $\setI', \setI''$ in
$[r]$, we say that $\setI'$ \emph{is contained in} $\setI''$ if all
the boundary points of intervals in $\setI'$ are also boundary points
of intervals in $\setI''$. Figure~\ref{fig:containment} shows an
example of two collections of intervals, one contained in the other. 
If there is a maximal $\setI^\ast$ that is contained in $\setI$ we say
that $\setI$ contains a maximal subset. 
We say that $\setI'$ is \emph{atomic} with respect to $\setI''$ if
every interval in $\setI'$ is also in
$\setI''$. Figure~\ref{fig:atomic} gives an
example of two collections of intervals, one atomic with respect to the other.
If there is a maximal $\setI^\ast$ that is atomic with respect to
$\setI$ then we say that the maximum is atomic with respect to
$\setI$.

\begin{figure}[h]
\begin{center}
\begin{tikzpicture}[mydrawstyle/.style={draw=black, very thick}, x=1mm, y=1mm, z=1mm]
  \draw[mydrawstyle] (-10,14) node{$\mathcal{I}':$};
  \draw[mydrawstyle, -](0,14)--(25,14) node at (-6,14)[left]{};
  \draw[mydrawstyle, -](48,14)--(66,14) node at (-6,14)[left]{};
  \draw[mydrawstyle, -](76,14)--(90,14) node at (-6,14)[left]{};
  \draw[mydrawstyle](0,12)--(0,16) node{};
  \draw[mydrawstyle](10,12)--(10,16) node{};
  \draw[mydrawstyle](25,12)--(25,16) node{};
 \draw[mydrawstyle](48,12)--(48,16) node{};
 \draw[mydrawstyle](66,12)--(66,16) node{};
 \draw[mydrawstyle](76,12)--(76,16) node{};
 \draw[mydrawstyle](90,12)--(90,16) node{};
  \draw[mydrawstyle] (-10,2) node{$\mathcal{I}'':$};
  \draw[mydrawstyle, -](0,2)--(34,2) node at (-6,10)[left]{};
  \draw[mydrawstyle, -](48,2)--(66,2) node at (-6,10)[left]{};
  \draw[mydrawstyle, -](76,2)--(96,2) node at (-6,10)[left]{};
  \draw[mydrawstyle](0,0)--(0,4) node{};
 \draw[mydrawstyle](10,0)--(10,4) node{};
 \draw[mydrawstyle](15,0)--(15,4) node{};
 \draw[mydrawstyle](25,0)--(25,4) node{};
 \draw[mydrawstyle](34,0)--(34,4) node{};
 \draw[mydrawstyle](48,0)--(48,4) node{};
 \draw[mydrawstyle](54,0)--(54,4) node{};
 \draw[mydrawstyle](66,0)--(66,4) node{};
 \draw[mydrawstyle](76,0)--(76,4) node{};
 \draw[mydrawstyle](96,0)--(96,4) node{};
 \draw[mydrawstyle](90,0)--(90,4) node{};
\end{tikzpicture}
\end{center}
\caption{$\mathcal{I}'$ is contained in $\mathcal{I}''$ since each
  boundary point of all intervals in $\mathcal{I}'$ are boundary
  points of some interval in $\mathcal{I}''$.}
\label{fig:containment}
\end{figure}
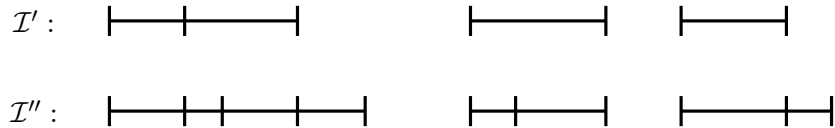
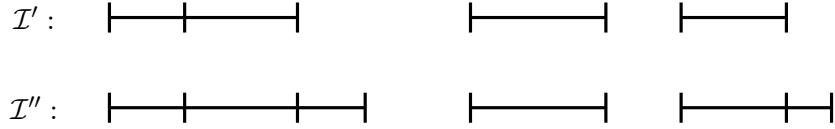
\begin{figure}[h]
\begin{center}
\begin{tikzpicture}[mydrawstyle/.style={draw=black, very thick}, x=1mm, y=1mm, z=1mm]
  \draw[mydrawstyle] (-10,14) node{$\mathcal{I}':$};
  \draw[mydrawstyle, -](0,14)--(25,14) node at (-6,14)[left]{};
  \draw[mydrawstyle, -](48,14)--(66,14) node at (-6,14)[left]{};
  \draw[mydrawstyle, -](76,14)--(90,14) node at (-6,14)[left]{};
  \draw[mydrawstyle](0,12)--(0,16) node{};
  \draw[mydrawstyle](10,12)--(10,16) node{};
  \draw[mydrawstyle](25,12)--(25,16) node{};
 \draw[mydrawstyle](48,12)--(48,16) node{};
 \draw[mydrawstyle](66,12)--(66,16) node{};
 \draw[mydrawstyle](76,12)--(76,16) node{};
 \draw[mydrawstyle](90,12)--(90,16) node{};
  \draw[mydrawstyle] (-10,2) node{$\mathcal{I}'':$};
  \draw[mydrawstyle, -](0,2)--(34,2) node at (-6,10)[left]{};
  \draw[mydrawstyle, -](48,2)--(66,2) node at (-6,10)[left]{};
  \draw[mydrawstyle, -](76,2)--(96,2) node at (-6,10)[left]{};
  \draw[mydrawstyle](0,0)--(0,4) node{};
 \draw[mydrawstyle](10,0)--(10,4) node{};
\draw[mydrawstyle](25,0)--(25,4) node{};
 \draw[mydrawstyle](34,0)--(34,4) node{};
 \draw[mydrawstyle](48,0)--(48,4) node{};
\draw[mydrawstyle](66,0)--(66,4) node{};
 \draw[mydrawstyle](76,0)--(76,4) node{};
 \draw[mydrawstyle](96,0)--(96,4) node{};
 \draw[mydrawstyle](90,0)--(90,4) node{};
\end{tikzpicture}
\end{center}
\caption{$\mathcal{I}'$ is atomic with respect to $\mathcal{I}''$, since each
interval in $\mathcal{I}'$ is also an interval in $\mathcal{I}''$. }
\label{fig:atomic}
\end{figure}
We will prove the following invariant of our algorithm:
\begin{lemma}
\label{lem:ak-induction}
For any $i \geq 0$, if $\setI_i$ contains a maximal subset, then either the maximum is atomic with respect to $\setI_i$ or $\setI_{i + 1}$ contains a maximal subset.
\end{lemma}

Before we prove this lemma, let us see how it suffices to prove
Theorem~\ref{thm:Akcomp-correctness}.  
Now the set $\setI_0$ contains a maximal subset.
By induction and  Lemma~\ref{lem:ak-induction}, for all $i$, as long
as the maximum is not atomic with respect to $\setI_i$, $\setI_{i +
  1}$ contains a maximal subset. 
\textsc{ComputeDiscreteAk} stops iterating at iteration $i_f$
if $\setI_{i_f}$ has at most $k$ intervals. 
At this point either the maximum was atomic with respect to some
$\setI_i$, or $\setI_{i_f}$ contains a maximal subset. 
Let $\setI^\ast$ be any maximal subset it contains.
We observe that
\[\sum_{I \in \setI^\ast} |w(I)| \leq \sum_{I \in \setI_{i_f}} |w(I)| \; ,\]
and moreover, $\setI_{i_f}$ has $k$ pieces, so $\setI_{i_f}$ is itself
maximal, and is atomic with respect to itself. 

Thus, there is some $i$ so that $\setI_i$ contains a maximal subset
that is atomic with respect to $\setI_i$. 
Call this maximal subset $\setI^\ast$.
But then since it is atomic with respect to $\setI_i$, we have that
\[\sum_{I \in \setI^\ast} |w(I)| \leq \sum_{I \in \mathcal{J}_i} |w(I)| = m_i \;,\]
since $\mathcal{J}_i$ is chosen to maximize the sum over all sets of
$k$ intervals which are atomic with respect to $\setI_i$. 
Since $\setI^\ast$ achieves the maximum for \textsc{DiscreteAk}, we conclude that $m_i$ is indeed the maximum.
Thus whatever $m_{i'}$ we output is also the maximum, and its
$\mathcal{J}_{i'}$ attains the maximum.
This completes the proof of Theorem~\ref{thm:Akcomp-correctness}
assuming Lemma~\ref{lem:ak-induction}.
\end{proof}

We now prove Lemma \ref{lem:ak-induction}.
\begin{proof}[Proof of Lemma \ref{lem:ak-induction}]
It suffices to show that if $\setI_i$ contains a maximal subset, but the maximum is not atomic with respect to $\setI_i$, then $\setI_{i + 1}$ also contains a maximal subset.
Thus, let $\setI^\ast$ be such that 
\begin{enumerate}
\item
$\setI^\ast$ is maximal
\item
$\setI^\ast$ is contained in $\setI_i$, and
\item
there is no $\setI^\ast_1 \neq \setI^\ast$ satisfying conditions (1) and (2) so that every interval in $\setI^\ast_1$ is contained in some interval in $\setI^\ast$.
\end{enumerate}
Such an $\setI^\ast$ clearly exists by the assumption on $\setI_i$. 
Note that $\setI^\ast$ cannot be atomic with respect to $\setI_i$.
By observation we may assume that no interval $I'$ in a maximal subset will ever end on a point $a$ so that $w(I')$ and $c_a$ have different signs, since otherwise we can easily modify the partition to have this property while still maintaining properties (1)-(3).
More generally, we may assume there does not exist an interval $I''$ contained in $I'$ with right endpoint equal to $I'$'s right endpoint (resp. left endpoint equal to $I'$'s left endpoint) so that $w(I')$ and $w(I'')$ have different signs.

Let $I_j = [\beta_2, \beta_3]$ be the interval in $\setI_i$ with minimal $|w(I)|$ amongst all $I \in \setI_i$.
WLOG assume that it is not the leftmost or rightmost interval (the analysis for these cases is almost identical and so we omit it).
Let $\beta_1$ be the left endpoint of $I_{j - 1}$ and $\beta_4$ be the right endpoint of $I_{j + 1}$.
WLOG assume that $w(I_j) < 0$.

The initial partition $\setI_0$ had the property that the signs of the values $w(I)$ for $I \in \setI_0$ alternated,
and through a simple inductive argument it follows that for all $\setI_i$, the signs of the values $w(I)$ for $I \in \setI_i$ still alternate.
Thus, we have $w(I_{j -1}), w(I_{j + 1}) \geq 0$.
Since $\setI^\ast$ is not atomic, there is some $I_a \in \setI^\ast$ which contains at least two intervals $I_1, I_2$ of $\setI_i$.
Moreover, since the signs of the $w(I)$ of the intervals in $\setI_i$
alternate, we may assume that $w(I_1)$ and $w(I_2)$ have different
signs. 
Thus, by observation, we may in fact assume that $I_a$ contains
three consecutive intervals $I_1 < I_2 < I_3,$ and that $w(I_a),
w(I_1), w(I_3)$ have the same sign, and $w(I_2)$ has a different
sign. Moreover, $|w(I_j)|\le \min\{w(I_1), w(I_2), w(I_3)\}$. 
Moreover, define $I_a^1$ to be the interval which shares a left 
endpoint with $I_a$, and which has right endpoint the right endpoint
of $I_1$, and $I_a^2$ to be the interval which shares a right endpoint
with $I_a$, and which has left endpoint the left endpoint of
$I_3$ (See Figure~\ref{fig:explain}).

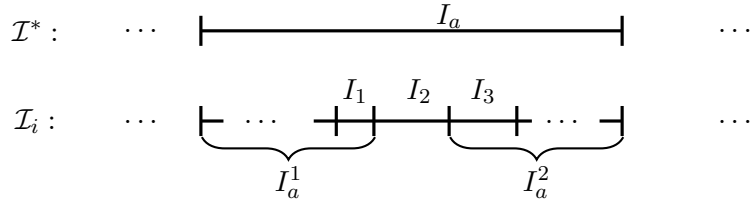
\begin{figure}[h]
\begin{center}
\begin{tikzpicture}[mydrawstyle/.style={draw=black, very thick}, x=1mm, y=1mm, z=1mm]

\draw[mydrawstyle, -](-8,14)--(48,14) node at (-6,10)[left]{};
\draw[mydrawstyle](-8,12)--(-8,16) node{};
\draw[mydrawstyle](48,12)--(48,16) node{};

\draw[mydrawstyle] (-16,14)node{$\ldots$};

\draw[mydrawstyle](63,14)node{$\ldots$};

\draw[mydrawstyle] (-16,2)node{$\ldots$};

\draw[mydrawstyle](63,2)node{$\ldots$};

 \draw[mydrawstyle] (-30,14) node{$\mathcal{I}^\ast:$};

 \draw[mydrawstyle] (-30,2) node{$\mathcal{I}_i:$};
 \draw[mydrawstyle] (25,16) node{$I_a$};

\draw[mydrawstyle, -](-8,2)--(-5,2) node at (-6,10)[left]{};
\draw[mydrawstyle, -](7,2)--(36,2) node at (-6,10)[left]{};
\draw[mydrawstyle, -](45,2)--(48,2) node at (-6,10)[left]{};

  \draw[mydrawstyle](-8,0)--(-8,4) node{};
 \draw[mydrawstyle](10,0)--(10,4) node{};
 \draw[mydrawstyle](15,0)--(15,4) node{};
 \draw[mydrawstyle](25,0)--(25,4) node{};
 \draw[mydrawstyle](34,0)--(34,4) node{};
 \draw[mydrawstyle](48,0)--(48,4) node{};
\draw[mydrawstyle](40,2) node{$\ldots$};
\draw[mydrawstyle](0,2) node{$\ldots$};
\draw[mydrawstyle](29.5,6) node{$I_3$};
 \draw[mydrawstyle](12.5,6) node{$I_1$};
\draw[mydrawstyle](21,6) node{$I_2$};
\draw [thick,
black,decorate,decoration={brace,amplitude=3mm,mirror},xshift=0.4,yshift=-0.4](-8,0)
-- (15,0) node[black,midway,yshift=-0.6cm] {$I_a^1$}; 
\draw [thick,
black,decorate,decoration={brace,amplitude=3mm,mirror},xshift=0.4,yshift=-0.4](25,0)
-- (48,0) node[black,midway,yshift=-0.6cm] {$I_a^2$}; 
\end{tikzpicture}
\end{center}
\caption{The interval $I_a$, and the intervals $I_1$, $I_2$, and
  $I_3$.}
\label{fig:explain}
\end{figure}

We must have that  $w(I_a^1), w(I_a^2)$ are the same sign as $w(I_a)$
as otherwise, say if $w(I_a^1)$'s sign was different from $w(I_a)$'s
sign, we would have $|w(I_a) | \leq |w(I_a^2)|$ and so the existence
of the collection of intervals $\setI' = (\setI^\ast \setminus I_a)
\cup I_a^2$ violates condition (3), since it is contained in
$\setI_i$, and  
\[\sum_{I \in \setI'} |w(I)| = \sum_{I \in \setI^\ast} |w(I)| - |w(I_a)| + |w(I_a^2)| \geq \sum_{I \in \setI^\ast} |w(I)| \;, \]
so it is maximal. 

Since $\setI^\ast$ is contained in $\setI_i$, the only boundary points that intervals in $\setI_i$ can have in the interval $[\beta_1, \beta_4]$ are at the points $\beta_i$ for $i \in \{1, 2, 3, 4\}$.
There are a few cases.

\paragraph{Case 1} If no interval in $\setI^\ast$ has any boundary point at $\beta_2$ or $\beta_3$, then it is still contained in $\setI_{i + 1}$, by the definition of $\setI_{i + 1}$.
\paragraph{Case 2} If $[\beta_2, \beta_3] \in \setI^\ast$, define $\setI' = ( \setI^\ast \setminus \{[\beta_2, \beta_3], I_a \} ) \cup \{ I_a^1, I_a^2 \}$.
Then 
\[\sum_{I \in \setI'} |w(I)| = \sum_{I \in \setI^\ast} |w(I)| - |w(I_j)| + |w(I_2)| \geq \sum_{I \in \setI^\ast} |w(I)| \;\]
by the choice of $I_j$, so $\setI'$ is maximal, contained in
$\setI_i$, and thus its existence violates condition (3), so this case
is impossible. This is illustrated in Figure~\ref{fig:casetwo}, where
for simplicity $I_a$ contains precisely three intervals. 
\begin{figure}[h]
\begin{center}
\begin{tikzpicture}[mydrawstyle/.style={draw=black, very thick}, x=1mm, y=1mm, z=1mm]
  \draw[mydrawstyle] (-10,14) node{$\mathcal{I}^\ast:$};
  \draw[mydrawstyle, -](10,14)--(34,14) node at (-6,14)[left]{};
  \draw[mydrawstyle, -](48,14)--(66,14) node at (-6,14)[left]{};
  \draw[mydrawstyle, -](76,14)--(90,14) node at (-6,14)[left]{};
 \draw[mydrawstyle](10,12)--(10,16) node{};
  \draw[mydrawstyle](34,12)--(34,16) node{};
 \draw[mydrawstyle](48,12)--(48,16) node{};
  \draw[mydrawstyle] (41,14) node{$\ldots$};
  \draw[mydrawstyle] (71,14) node{$\ldots$};
  \draw[mydrawstyle] (41,2) node{$\ldots$};
  \draw[mydrawstyle] (71,2) node{$\ldots$};
\draw[mydrawstyle](23,17) node{$I_a$};
 \draw[mydrawstyle](66,12)--(66,16) node{};
 \draw[mydrawstyle](76,12)--(76,16) node{};
 \draw[mydrawstyle](90,12)--(90,16) node{};
  \draw[mydrawstyle] (-10,2) node{$\mathcal{I}_i:$};
  \draw[mydrawstyle, -](0,2)--(36,2) node at (-6,10)[left]{};
  \draw[mydrawstyle, -](46,2)--(68,2) node at (-6,10)[left]{};
  \draw[mydrawstyle, -](74,2)--(96,2) node at (-6,10)[left]{};
  \draw[mydrawstyle, -](48,2)--(66,2) node at (-6,10)[left]{};
  \draw[mydrawstyle, -](76,2)--(96,2) node at (-6,10)[left]{};
  \draw[mydrawstyle](0,0)--(0,4) node{};
 \draw[mydrawstyle](10,0)--(10,4) node{};
 \draw[mydrawstyle](15,0)--(15,4) node{};
 \draw[mydrawstyle](25,0)--(25,4) node{};
 \draw[mydrawstyle](34,0)--(34,4) node{};
 \draw[mydrawstyle](48,0)--(48,4) node{};
 \draw[mydrawstyle](58,-3) node{$I_j$};
 \draw[mydrawstyle](48.5,-3) node{$\beta_2$};
 \draw[mydrawstyle](66.5,-3) node{$\beta_3$};
 \draw[mydrawstyle](66,0)--(66,4) node{};
 \draw[mydrawstyle](12.5,-3) node{$I_a^1$};
 \draw[mydrawstyle](20.5,-3) node{$I_2$};
 \draw[mydrawstyle](29.5,-3) node{$I_a^2$};
 \draw[mydrawstyle](76,0)--(76,4) node{};
 \draw[mydrawstyle](96,0)--(96,4) node{};
 \draw[mydrawstyle](90,0)--(90,4) node{};
\end{tikzpicture}
\end{center}
\caption{When $I_j$ is an element of $\mathcal{I}^\ast$, we can drop
  it and add the intervals $I_a^1$ and $I_a^2$ achieving a larger weight.}
\label{fig:casetwo}
\end{figure}
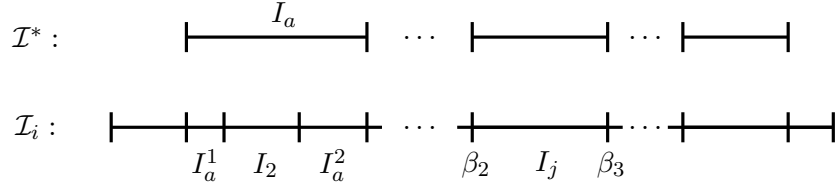

\paragraph{Case 3} If $\beta_3$ is the right endpoint of some interval $I \in \setI^\ast$, then by the same reasoning as before, we may assume that $w(I) < 0$. 
Then, let $I'$ be the interval with the same left endpoint as $I$ but with right endpoint $\beta_1$.
Since then $w(I') = w(I) - w(I_{j - 1}) - w(I_j) \leq w(I)$,
the partition $\setI' = \setI^\ast \setminus I \cup I'$ is maximal,
contained in $\setI_i$, and its existence again violates condition
(3), so this case is impossible. An illustration is given in Figure~\ref{fig:casethree}.
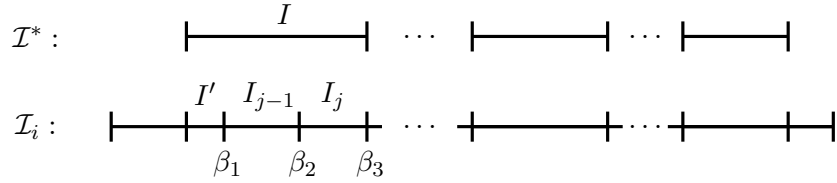
\begin{figure}[h]
\begin{center}
\begin{tikzpicture}[mydrawstyle/.style={draw=black, very thick}, x=1mm, y=1mm, z=1mm]
  \draw[mydrawstyle] (-10,14) node{$\mathcal{I}^\ast:$};
  \draw[mydrawstyle, -](10,14)--(34,14) node at (-6,14)[left]{};
  \draw[mydrawstyle, -](48,14)--(66,14) node at (-6,14)[left]{};
  \draw[mydrawstyle, -](76,14)--(90,14) node at (-6,14)[left]{};
 \draw[mydrawstyle](10,12)--(10,16) node{};
  \draw[mydrawstyle](34,12)--(34,16) node{};
 \draw[mydrawstyle](48,12)--(48,16) node{};
  \draw[mydrawstyle] (41,14) node{$\ldots$};
  \draw[mydrawstyle] (71,14) node{$\ldots$};
  \draw[mydrawstyle] (41,2) node{$\ldots$};
  \draw[mydrawstyle] (71,2) node{$\ldots$};
\draw[mydrawstyle](23,17) node{$I$};
 \draw[mydrawstyle](66,12)--(66,16) node{};
 \draw[mydrawstyle](76,12)--(76,16) node{};
 \draw[mydrawstyle](90,12)--(90,16) node{};
  \draw[mydrawstyle] (-10,2) node{$\mathcal{I}_i:$};
  \draw[mydrawstyle, -](0,2)--(36,2) node at (-6,10)[left]{};
  \draw[mydrawstyle, -](46,2)--(68,2) node at (-6,10)[left]{};
  \draw[mydrawstyle, -](74,2)--(96,2) node at (-6,10)[left]{};
  \draw[mydrawstyle, -](48,2)--(66,2) node at (-6,10)[left]{};
  \draw[mydrawstyle, -](76,2)--(96,2) node at (-6,10)[left]{};
  \draw[mydrawstyle](0,0)--(0,4) node{};
 \draw[mydrawstyle](10,0)--(10,4) node{};
 \draw[mydrawstyle](15,0)--(15,4) node{};
 \draw[mydrawstyle](25,0)--(25,4) node{};
 \draw[mydrawstyle](34,0)--(34,4) node{};
 \draw[mydrawstyle](48,0)--(48,4) node{};
 \draw[mydrawstyle](66,0)--(66,4) node{};
 \draw[mydrawstyle](15.5,-3) node{$\beta_1$};
 \draw[mydrawstyle](29.5,6) node{$I_j$};
 \draw[mydrawstyle](12.5,6.4) node{$I'$};
 \draw[mydrawstyle](34.5,-3) node{$\beta_3$};
 \draw[mydrawstyle](21,6) node{$I_{j-1}$};
 \draw[mydrawstyle](25.5,-3) node{$\beta_2$};
 \draw[mydrawstyle](76,0)--(76,4) node{};
 \draw[mydrawstyle](96,0)--(96,4) node{};
 \draw[mydrawstyle](90,0)--(90,4) node{};
\end{tikzpicture}
\end{center}
\caption{$\mathcal{I}^\ast$ can drop $I$ and instead take $I'$ to get
  a larger weight.}
\label{fig:casethree}
\end{figure}

\paragraph{Case 4} If $\beta_2$ is the left endpoint of some interval $I \in \setI^\ast$, then analogous reasoning to that in Case 3 results in a contradiction, so this case is also impossible.

\paragraph{Case 5a} If $\beta_2$ is the right endpoint of some interval $I \in \setI^\ast$, and no interval in $\setI^\ast$ contains $I_{j + 1}$, then we know that $w(I) \geq 0$.
Let $I'$ be the interval $I \cup I_j \cup I_{j + 1}$.
Then, the partition $\setI' = \setI^\ast \setminus I \cup I'$ is maximal by the same kind of reasoning as before, and $\setI'$ is contained in $\setI_{j + 1}$. Thus, this case is possible and consistent with the Lemma.

\paragraph{Case 5b} If $\beta_2$ is the right endpoint of some interval $I \in \setI^\ast$ and $\beta_3$ is the left endpoint of some interval $I' \in \setI^\ast$, then we know that $w(I), w(I') \geq 0$. Let $I'' = I \cup I_j \cup I'$.
Define $\setI' = (\setI^\ast \setminus \{I, I', I_a\}) \cup \{ I'', I_a^1, I_a^2\}$.
Then,
\[
\sum_{I \in \setI'} |w(I)| =  \sum_{I \in \setI^\ast} |w(I)| - |w(I_j)| + |w(I_2)| \geq  \sum_{I \in \setI^\ast} |w(I)| \;,
\]
so again this is a maximal subset which is now contained in $\setI_{j + 1}$.

\paragraph{Case 6} If $\beta_3$ is the left endpoint of some interval in $\setI^\ast$, by analogous reasoning to that in Cases 5a and 5b, we may conclude that in this case, the Lemma holds.

These cases encompass all possible cases, and thus we conclude that
the Lemma holds, as claimed. 
\end{proof}

\subsubsection{Description of \textsc{ComputeAk}}
Our algorithm \textsc{ComputeAk} uses Fact
\ref{fact:fastpolyeval} below, produces the sequence $P_{\disc}(i) - E(i)$,
and computes $\| \{ P_{\disc} - E \} \|_{\Ak}$ using
\textsc{ComputeDiscreteAk}. 

It thus suffices to show that we can construct this sequence $P_{\disc}(i) - E(i)$ efficiently when we are given the empirical distribution and the polynomial $p$.
The only difficulty lies in efficiently computing the $p[x_i, x_{i + 1}]$.
This problem is equivalent to efficiently evaluating the integral of $p$ at all the points in $\mathcal{X}$, which is in turn equivalent to efficiently evaluating a degree $\pdeg + 1$ polynomial at $s$ points.
To do so, we use the following well-known fact:
\begin{fact}[\cite{von2013modern}, p.~299 and p.~245]
\label{fact:fastpolyeval}
Let $x_1, \ldots, x_s$ be a set of $s$ real numbers and let $p$ be a polynomial of degree at most $s$.
Then there is an algorithm that computes $p(x_1), \ldots, p(x_s)$ in time $O(s \log^2 s)$.
\end{fact}

After solving the discretized version of the problem, the algorithm outputs the estimate that \textsc{ComputeDiscreteAk} computed and the processed version of the intervals, where the processing is the one described in Lemma \ref{lem:akdiscrete}.
Thus, we have:

\begin{proof}[Proof of Theorem \ref{thm:computeak}]
The correctness of the algorithm follows from Lemma \ref{lem:akdiscrete} and the arguments given above. 
Thus it suffices to bound the running time.
The time required to produce the sequence $P_{\disc}(i) - E(i)$ is bounded by computing the $p[x_i, x_{i + 1}]$, which can be done  in time $O((s + \pdeg) \log^2 (s + \pdeg))$ by Fact \ref{fact:fastpolyeval}.
Moreover, the running time of $\textsc{ComputeDiscreteAk}$ on the sequence $P_{\disc}(i) - E(i)$ is $O(s \log s)$.
Hence, the running time of the overall algorithm is $O((s + \pdeg) \log^2 (s + \pdeg))$, as claimed.
\end{proof}

\section{Applications} \label{sec:applications}
In this section, we apply our main result to obtain near optimal estimators for various
classes of structured distributions. As described in Table~\ref{tab:long}, 
we consider arbitrary mixtures of well-studied distribution families, including log-concave distributions,
normal distributions, densities with bounded number of modes, and density functions
in Besov spaces. We also consider mixtures of discrete structured distributions over an ordered domain,
such as  multi-modal distributions, monotone hazard rate
distributions, Poisson, and Binomial distributions. For all these
classes, our sample complexity and running 
time match the information-theoretic optimum, up to at most logarithmic
factors.  

We note that even though our algorithm is stated for distributions over a known finite interval,
they are also applicable to distributions over the entire real line, such
as (mixtures of) Gaussians or Poisson distributions. This follows from the
following fact: let $x_{\min}$ and $x_{\max}$ be the smallest and
largest elements among $\frac{\log(1/\delta)}{\eps^2}$ draws from any
distribution.
Then with probability at least $1-\delta$, the
distribution assigns probability mass at least $1-\eps$ to the interval
$[x_{\min},x_{\max}].$ 
Thus, at a cost of $\frac{\log(1/\delta)}{\eps^2}$  samples, we may 
truncate the distribution and thereafter only consider
this finite interval.

\subsection{Mixture of log-concave distributions}

For an interval $I \subseteq \R$, a function $g: I \to \R$ is called {\em concave} if
for any $x, y \in I$ and $ \lambda \in [0,1]$ it holds
$g\left( \lambda x + (1-\lambda)y \right) \ge \lambda g(x)+(1-\lambda) g(y).$
A function $h: I \to \R_+$ is called {\em log-concave}
if $h(x) = \exp\left( g(x) \right)$, where $g:I \to \R$ is concave.
A density $f$ is a $k$-mixture of log-concave density functions if
there exist $w_1,\ldots,w_k\ge0$, $\sum_i w_i =1$ and log-concave
density functions $f_1, \ldots, f_k$ such that $f=\sum w_if_i$. 
The class of log concave distributions is very broad and contains the
class of Gaussians, uniform, exponential, Gamma, Beta, 
and Weibull distributions. Log-concave distributions have received significant interest in economics and
statistics~\cite{BagnoliB05, Cule10a, DumbgenRufibach:09, DossW13, ChenSam13, KimSam14, BalDoss14, HW15}.   

It was shown in~\cite{CDSS14} that a $k$-mixture of log-concave
density functions can be $\eps$-approximated in $L_1$-norm by a $t$-piecewise linear density, 
for $t = \Otilde(k/\sqrt{\eps})$. 
Using this structural result, \cite{CDSS14} gave a polynomial time algorithm with sample complexity $\Otilde(t/\eps^2) = \Otilde(k/\eps^{5/2})$ to agnostically learn a $k$-mixture of log-concave distributions.
This sample bound is nearly optimal, as $\Omega(k/\eps^{5/2})$ samples are necessary for this learning problem.

Our main result yields a sample optimal and nearly-linear time algorithm
for this problem. In particular, this follows from a combination of Theorem~\ref{thm:main-intro} and
a recently obtained tight structural result that removes the
logarithmic factors from the previous construction of \cite{CDSS14}. 
In particular, it is shown in \cite{DiakonikolasK15} that a $k$-mixture of log-concave
density functions can be $\eps$-approximated in $L_1$-norm by a $t$-piecewise linear density, 
for $t = O(k/\sqrt{\eps})$. 
As a corollary, we obtain the following:
\begin{theorem}
\label{thm:log-concave}
There is an agnostic learning algorithm for the class of $k$-mixtures of log-concave distributions over the real line
that uses  $O(k/\eps^{5/2})$ samples and runs in time $\Otilde((k/\eps^{5/2}))$.
\end{theorem}

\subsection{Mixture of Gaussians}

Let $N(\mu,\sigma^2)$ denote the normal distribution with mean $\mu$
and variance $\sigma^2$.  
A density $f: \R \to \R_+$ is a $k$-mixture of Gaussians if there exist
$w_1,\ldots,w_k\ge0$, $\sum_i w_i =1$, $\mu_1,\ldots,\mu_k \in \R$, and
$\sigma_1,\ldots, \sigma_k \in \R_+$ such that $f=\sum_{i=1}^k w_iN(\mu_i,\sigma_i^2)$. 

In the theoretical computer science community, the problem of parameter estimation for Gaussian mixtures was initiated by~\cite{Dasgupta:99}.
Recent work has obtained polynomial sample and time algorithms for this problem
under the conditions of identifiability~\cite{MoitraValiant:10, BelkinSinha:10}. 
We remark that learning the parameters of a mixture of $k$ 
univariate Gaussians to accuracy $\eps$ requires
$\Omega ((1 / \epsilon)^{6 k - 2})$ samples~\cite{HP:15} .

The problem of proper learning for Gaussian mixtures has also been recently studied in
\cite{DK14, SOAJ14} who obtain algorithms that draw
$\Otilde(k/\eps^2)$ samples and run in time $O((1/\eps)^{3k - 1})$.
Another approach, due to \cite{BSZ15}, outputs a mixture of $O(k/\eps^3)$ Gaussians in time and sample complexity of $O(k/\eps^6)$.

It is well-known (see, e.g., ~\cite[Section 7.21]{timan1963theory}  or \cite{CDSS14}) 
that a normal distribution is $\eps$-close to a $3$-piecewise polynomial of
degree $O(\log(1/\eps))$. Using this structural result, \cite{CDSS14}
obtain a nearly sample optimal and polynomial time agnostic learning algorithm for this problem. 

As a corollary of Theorem~\ref{thm:main-intro}, we obtain a 
nearly sample optimal and nearly-linear time algorithm.
(The sample complexity of our algorithm is better than that of \cite{CDSS14} by logarithmic factors.)
In particular:
\begin{theorem} \label{thm:mog}
There is an agnostic learning algorithm for $k$-mixtures of univariate Gaussians 
that draws $O((k/\eps^2)\log(1/\eps))$ samples and runs in time  $\Otilde(k/\eps^2)$.  
\end{theorem}

\subsection{Densities in Besov spaces}

Densities in Besov spaces constitute a broad family of distributions, including
piecewise polynomials and the exponential family.
Density estimation for functions in Besov spaces has received considerable attention in the
statistics and information theory literature. A lot of the early work on the topic relied on 
wavelet techniques, based on the fact that functions in Besov spaces are amenable to multiscale
decompositions~\cite{DeVore98, Donoho96, Donoho98}.  

A piecewise smooth density function $f$ has the following decomposition, 
\[
f(x)= \sum_{k}c_{j_0,k}\phi_{j_0,k}(x)+\sum_{j=j_0}^{\infty}\sum_kd_{j_0,k}\psi_{j_0,k}(x)
\]
where the $\phi$'s are scaling functions and the $\psi$'s are wavelet functions. The
Besov space $\Besovapq$ is the following subset of such density functions
\[
\Besovapq\eqdef\left\{f:\|c_{j_0,k}\|_{\ell_p}+\left(\sum_{j=j_0}^{\infty}\left(2^{\alpha
        jp}\sum_k |d_{j,k}|^p\right)^{q/p}\right)^{1/q}<\infty\right\},
\]
for parameters $\alpha>\frac1p>0$ and $q>0$, where $\{c_{j_0,k}\}$ and
$\{d_{j,k}\}$ are the scaling and wavelet coefficients in the wavelet
expansion of $f$. 

Nowak and Willett~\cite{WillettN07} showed that any density $f$ in
$\Besovapq$ 
for $0<q\le p$, with $\frac1p=\alpha+\frac12$, can be approximated up to
$L_1$ error $\epsilon$ with $n=O_{\alpha}\left(\frac
  {\log^2(1/\eps)}{\eps^{\alpha+1/2}}\right)$ samples. They also propose an algorithm for this problem with running time 
  $\Omega(n^3)$.

As a corollary of our main result, we obtain a sample optimal and nearly-linear time agnostic algorithm 
for this problem.
A result in~\cite{DeVore98} implies that under the above
assumptions on $\alpha, p, q$, any function in
$\Besovapq$ can be $\eps$-approximated in $L_1$-norm by an
$O_{\alpha}(\eps^{-1/\alpha})$-piece degree-$O(\lceil \alpha \rceil)$ polynomial.
Combined with our main result, we obtain an algorithm with sample complexity 
$O_{\alpha}\left(\frac1{\eps^{2+1/\alpha}}\right)$, which is optimal
up to constant factors~\cite{WillettN07}.  
Moreover, the running time of our algorithm is nearly-linear in the number of samples. In
particular:
\begin{theorem}
\label{thm:besov}
There is an agnostic learning algorithm for 
$\Besovapq$, with $0<q<p$, $1/p=\alpha+1/2$
with sample complexity $O_{\alpha}\left(\frac1{\eps^{2+1/\alpha}}\right)$ and running time
$\Otilde_{\alpha}\left(\frac1{\eps^{2+1/\alpha}}\right).$
\end{theorem}

\subsection{Mixtures of \texorpdfstring{$t$}{t}-monotone distributions}
A density $f: \R \to \R_+$ is 1-monotone if it is non-increasing. 
It is 2-monotone if it is non-increasing and convex, and $t$-monotone for $t\ge3$ if 
$(-1)^jf^{(j)}$ is non-negative, non-increasing, and convex for
$j=0,\ldots, t-2$. 
A number of recent works in statistics studied the problem of estimating $t$-monotone
density functions in the context of the MLE~\cite{BW07aos,GW09sc, BW10sn}. 

Implicit in~\cite{KL04,KL07} is the fact that any $t$-monotone bounded density function
over $[0,1]$ can be approximated with an $O(1/\eps^{1/t})$ piecewise
degree $t-1$ polynomial. Using this along with our main result yields
the following guarantee on learning $t$-monotone distributions. 
\begin{theorem}
\label{thm:kmonotone}
There exists an agnostic learning algorithm for $k$-mixtures of $t$-monotone distributions
that uses $O(tk/\eps^{2+1/t})$
samples and runs in time $\Otilde(kt^{2+\omega}/\eps^{2+1/t})$.
\end{theorem}

The above is a significant improvement in the running time compared to~\cite{CDSS14}.
Note that for $t=1,2$, the sample complexity of our algorithm is optimal.
This follows from known lower bounds of $\Omega(1/\eps^3)$ for
$t=1$~\cite{Birge:87} and of $\Omega(1/\eps^{5/2})$ for $t=2$~\cite{DL:01}. 

\subsection{Mixtures of discrete distributions}
Our main result applies to the discrete setting as well,
leading to fast algorithms for learning 
mixtures of discrete distributions that can be well-approximated by piecewise polynomials.

\medskip

\noindent\textbf{Mixtures of $t$-modal discrete distributions and MHR distributions.} A
distribution over $[N]$ is unimodal if there is a
$j\in[N]$ such that the pmf is non-decreasing up to $j$, and
non-increasing after $j$. A distribution is $t$-modal if there is a
partition of $[N]$ into at most $t$ intervals over which the
conditional pmf is unimodal. It follows from \cite{Birge:87b, CDSS13} that any
mixture of $k$ $t$-modal distributions is $\eps$-close to a $(kt/\eps)\log(N/kt)$-histogram.
\cite{CDSS14b} implies an algorithm for this problem that uses 
$n = \Otilde({kt\log (N)/\eps^3})$ samples and runs in time $\Otilde(n)$. 
As a corollary of our main result, we obtain the first sample optimal (up to constant factors) 
and nearly-linear time algorithm:
\begin{theorem} \label{thm:tmodal}
There is an agnostic learning algorithm for $k$-mixtures of $t$-modal distributions
over $[N]$ that draws $O(\frac{kt\log(N/kt)}{\eps^3})$ samples and runs in time 
$O(\frac{kt\log(N/kt)}{\eps^3}\log(1/\eps))$.  
\end{theorem}

We similarly obtain a sample optimal and near-linear time algorithm for learning
mixtures of MHR distributions.

\medskip

For a distribution $p$ on $[N]$, the function $H(i)\eqdef \frac{p(i)}{\sum_{j\ge i}p(j)}$ 
is called the hazard rate function of $p$. The distribution $p$ is a
monotone hazard distribution (MHR) if $H(i)$ is non-decreasing.~\cite{CDSS13}
shows that a mixture of $k$ MHR distributions over $[N]$ can be
approximated up to distance $\eps$ using an $O(k\log(N/\eps)/\eps)$-histogram. Using this,~\cite{CDSS14b} yields a $\Otilde(k\log(N/\eps)/\eps^3)$
sample, $\Otilde(k\log(N/\eps)/\eps^3)$ time algorithm to estimate mixtures of MHR
distributions. We obtain
\begin{theorem}
\label{thm:mhr}
There is an agnostic learning algorithm for $k$-mixtures of MHR distributions 
over $[N]$ that draws $O(k\log(N/\eps)/\eps^3)$ samples and runs in time 
$O(\frac{k\log(N/\eps)}{\eps^3}\log(1/\eps))$.  
\end{theorem}

\noindent\textbf{Mixtures of Binomial and Poisson distributions.} We consider
mixtures of $k$ Binomial and Poisson distributions.
For these distribution families, the best sample complexity attainable using the techniques of \cite{CDSS14, CDSS14b}
 is $\Otilde (k / \epsilon^3)$. This follows from the fact that approximating a $k$-mixture of 
 Binomial or Poisson distributions by piecewise constant distributions requires $\Theta(k / \eps)$ pieces.

A recent result of~\cite{DaskalakisDS15} shows that any Binomial or Poisson 
distribution can be approximated to $L_1$ distance $\eps$
using $t$-piecewise degree-$\pdeg$ polynomials for $t = O(1)$ and $\pdeg = O(\log(1/\eps))$. 
Therefore, a  Binomial or Poisson  $k$-mixture can be approximated with $O(k)$-piecewise,
degree-$O(\log(1/\eps))$ polynomials. Since our main result applies to discrete piecewise polynomials as well, 
we obtain the following: 
\begin{theorem}
\label{thm:poi}
There is an agnostic learning algorithm for $k$-mixtures of Binomial or Poisson distributions
that uses  $O(\frac{k}{\eps^2}\log(1/\eps))$ samples and runs in time $\Otilde(k/\eps^2)$.  
\end{theorem}

\section{Experimental Evaluation}
\label{sec:experiments}
In addition to the strong theoretical guarantees proved in the previous sections, our algorithm also demonstrates very good performance in practice.
In order to evaluate the empirical performance of our algorithm, we conduct several experiments on synthetic data.
We remark that the evaluation here is preliminary, and we postpone a more detailed experimental study, including a comparison with related algorithms, to future work. Nevertheless, our results here show that both the empirical sample and time complexity are nearly optimal in a strong sense.
For example, no histogram learning algorithm that requires sorted samples can outperform the running time of our method by more than 30\%.
Similarly, our learning algorithm for piecewise linear hypotheses only adds a factor of $2 - 3\times$ overhead to the time needed to sort the samples.
Moreover, the sample complexity of our algorithm matches the quantity $t \cdot (\pdeg + 1)/ \eps^2$ up to a small constant between $1$ and $2$.

All experiments in this section were conducted on a laptop computer from 2010, using an Intel Core i7 CPU with 2.66 GHz clock frequency, 4 MB of cache, and 8 GB of RAM. We used Mac OS X 10.9 as operating system and g++ 4.8 as compiler with the -O3 flag (we implemented our algorithms in C++).
All reported running times and learning errors are averaged over 100 independent trials.
As an illustrative baseline, sorting $10^6$ double-precision floating point numbers with the std::sort algorithm from the C++ STL takes about 100 ms on the above machine.

Figure \ref{fig:distributions} shows the three distributions we used in our experiments: a mixture of two Gaussians, a mixture of two Beta distributions, and a mixture of two Gamma distributions.
The three distributions have different shapes (e.g., different numbers of modes), and the support size considered for these distributions differs.

\pgfplotsset{filter discard warning=false}

\pgfplotsset{distributionplot/.style={%
  scale only axis,
  enlarge x limits=false,
  ymin=0,
  width=4cm,
  height=3cm,
  scaled ticks=false,
  no markers,
  y tick label style={/pgf/number format/fixed},
  every axis plot/.append style={line width=.7pt},
  every axis title/.append style={at={(.5,-.3)},anchor=north,inner sep=0pt},
}}

\pgfplotsset{resultplot/.style={%
  scale only axis,
  enlarge x limits=false,
  enlarge y limits=false,
  error bars/x dir=none,
  error bars/y dir=both,
  error bars/y explicit,
  width=6cm,
  height=5cm,
  scaled ticks=false,
  grid style={dotted,gray,thin},
  grid=major,
  xlabel={Number of samples},
  ylabel shift=-4pt,
  ylabel={Running time (seconds)},
  every tick label/.append style={font=\small},
  every axis plot/.append style={line width=.7pt},
  cycle list={{blue,mark=*},{green,mark=square*},{red,mark=triangle*}},
  legend style={anchor=north west,at={(.05,.95)}},
}}
\pgfplotsset{timeresult/.style={%
  resultplot,
  ylabel={Running time (seconds)},
  y tick label style={/pgf/number format/fixed},
  legend style={anchor=north west,at={(.05,.95)}},
}}
\pgfplotsset{timeresultlog/.style={%
  timeresult,
  ytick={0.01,0.1,1.0},
  ymax=1.0,
  ymin=0.01,
}}
\pgfplotsset{timeresultlinear/.style={%
  timeresult,
  xmin=0,
  xtick={0,200000,400000,600000,800000,1000000},
}}
\pgfplotsset{errorresult/.style={%
  resultplot,
  ylabel={Learning error ($L_1$-distance)},
  ymin=0,
  y tick label style={/pgf/number format/fixed},
  legend style={anchor=north east,at={(.95,.95)}},
}}
\pgfplotsset{errorresultlog/.style={%
  errorresult,
  ytick={0.01,0.1,1.0},
  ymax=1.0,
  ymin=0.005,
}}
\pgfplotsset{errorresultlinear/.style={%
  errorresult,
  xmin=0,
  xtick={0,200000,400000,600000,800000,1000000},
}}

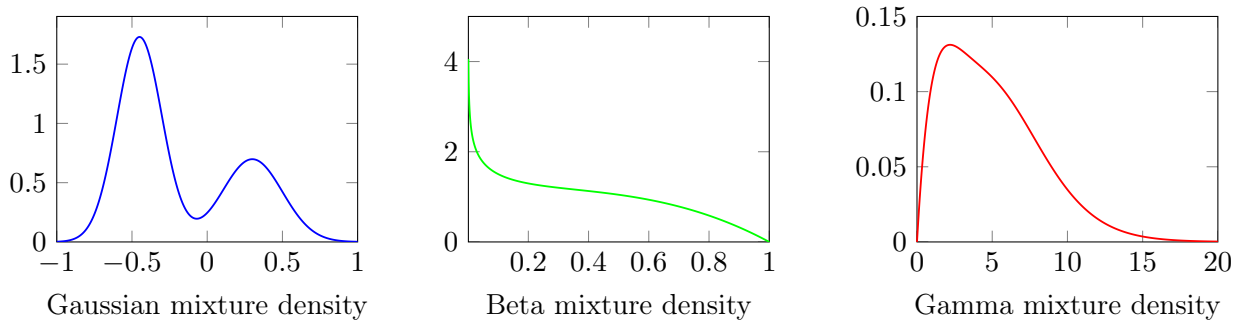
\begin{figure}[htb!]
\begin{tikzpicture}
\begin{axis}[distributionplot,title={Gaussian mixture density}]
\iftoggle{plotexperiments}{
  \addplot table[x=x,y=y] {data/distribution_plot_data_gmm.txt};
}{
}
\end{axis}
\end{tikzpicture}
\hfill
\begin{tikzpicture}
\begin{axis}[distributionplot,ymax=5,cycle list={{green}},title={Beta mixture density}]
\iftoggle{plotexperiments}{
  \addplot table[x=x,y=y] {data/distribution_plot_data_beta.txt};
}{
}
\end{axis}
\end{tikzpicture}
\hfill
\begin{tikzpicture}
\begin{axis}[distributionplot,ymax=0.15,cycle list={{red}},title={Gamma mixture density}]
\iftoggle{plotexperiments}{
  \addplot table[x=x,y=y] {data/distribution_plot_data_gamma.txt};
}{
}
\end{axis}
\end{tikzpicture}
\caption{The three test distributions.}
\label{fig:distributions}
\end{figure}

\subsection{Histogram hypotheses}
In order to evaluate our histogram learning algorithm (see Section~\ref{ssec:hist}), we use the following test setup.
For a given unknown distribution with pdf $f$, we draw $n$ i.i.d.\ samples from the unknown distribution.
We then give the sorted samples as input to our algorithm, which produces a histogram hypothesis $h$.
We set the parameters of our algorithm so that the resulting histogram contains 80 constant pieces.
As performance measures, we record the running time of our algorithm (excluding sorting) and the $L_1$-learning error achieved, i.e., $\norm{f - h}_1$.

Figure \ref{fig:histtime} contains the running time results, both on a linear scale and on a logarithmic scale.
The results indicate three important points:
(i) The running time of our algorithm scales nearly-linearly with the input size, i.e., the number of samples $n$.
(ii) The constant hidden in the big-$O$ notation of our analysis is very small.
In particular, the algorithm runs in less than 35 milliseconds for $10^6$ samples.
Note that this is three times faster than sorting the samples.
(iii) The running time of our algorithm essentially does not depend on the unknown distribution.
Such robustness guarantees are very desirable for reliable performance in practice.

The $L_1$-learning error results are displayed in Figure \ref{fig:histerror}.
The results show that the best learning error achievable with 80-piece histograms depends on the shape of the underlying distribution: 2-GMMs are harder to approximate than the Beta and Gamma mixtures.
This shows that for large number of samples, it is beneficial to use richer hypotheses classes such as piecewise linear functions (see the next subsection).
Nevertheless, our algorithm exhibits a good decay of the learning error before the regime where $\OPT_{80}$ dominates.

\begin{figure}[htb!]
\centering
\begin{tikzpicture}[/pgf/number format/sci generic={mantissa sep={\!\cdot\!},exponent={10^{#1}}}]
\begin{axis}[timeresultlinear]
\iftoggle{plotexperiments}{
  \addplot table[x=n,y=time_avg,y error minus=time_error_low, y error plus=time_error_high] {data/histogram_results_gmm.txt};
  \addplot table[x=n,y=time_avg,y error minus=time_error_low, y error plus=time_error_high] {data/histogram_results_beta.txt};
  \addplot table[x=n,y=time_avg,y error minus=time_error_low, y error plus=time_error_high] {data/histogram_results_gamma.txt};
  \legend{GMM,Beta,Gamma};
}{
}
\end{axis}
\end{tikzpicture}
\hfill
\begin{tikzpicture}
\begin{loglogaxis}[timeresultlog,ymin=0.00001,ymax=0.1,ytick={0.00001,0.0001,0.001,0.01,0.1}]
\iftoggle{plotexperiments}{
  \addplot table[x=n,y=time_avg,y error minus=time_error_low, y error plus=time_error_high] {data/histogram_results_gmm.txt};
  \addplot table[x=n,y=time_avg,y error minus=time_error_low, y error plus=time_error_high] {data/histogram_results_beta.txt};
  \addplot table[x=n,y=time_avg,y error minus=time_error_low, y error plus=time_error_high] {data/histogram_results_gamma.txt};
  \legend{GMM,Beta,Gamma};
}{
}
\end{loglogaxis}
\end{tikzpicture}
\caption{Running times for density estimation with histogram hypotheses.
The left plot shows the results on a linear scale, the right plot on a logarithmic scale.
As predicted by our analysis, the running time of our algorithm scales nearly-linearly with the input size $n$.
Moreover, the constant in the big-$O$ is very small: for $n=10^6$, our algorithm takes less than 35 milliseconds, which is about three times faster than sorting the samples.
The running time performance of our algorithm is also essentially independent of the unknown distribution.}
\label{fig:histtime}
\end{figure}
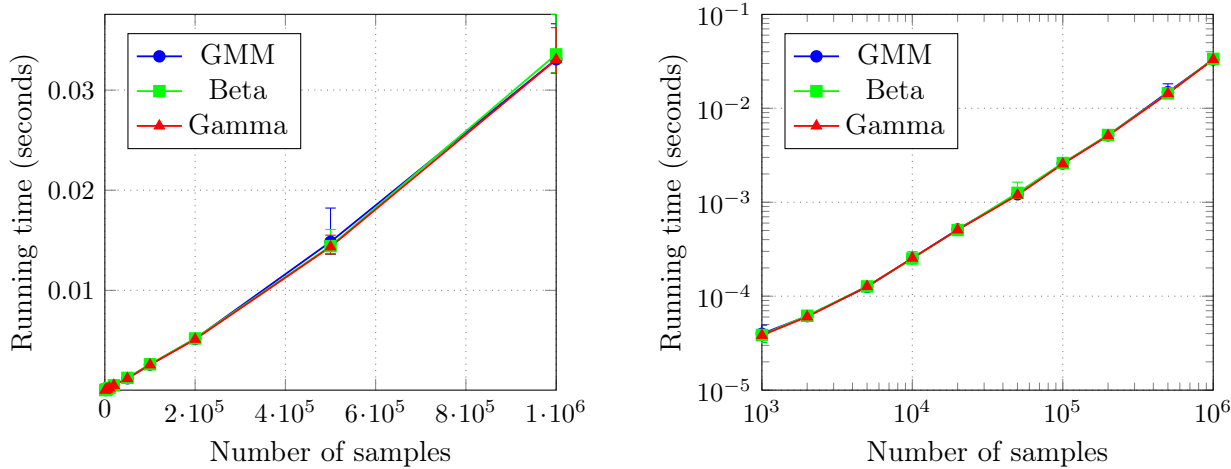

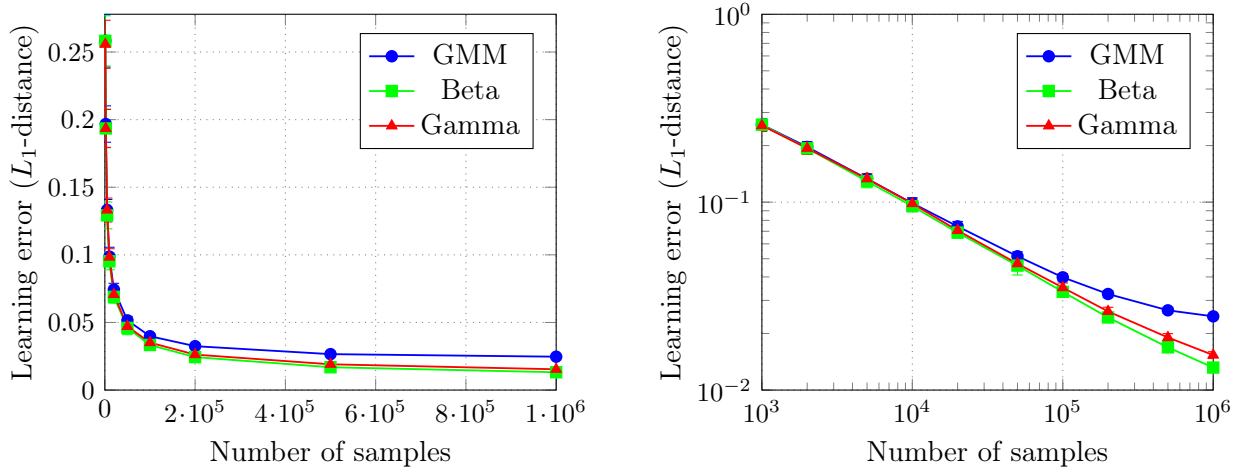
\begin{figure}[htb!]
\centering
\begin{tikzpicture}[/pgf/number format/sci generic={mantissa sep={\!\cdot\!},exponent={10^{#1}}}]
\begin{axis}[errorresultlinear]
\iftoggle{plotexperiments}{
  \addplot table[x=n,y=error_avg,y error=error_std] {data/histogram_results_gmm.txt};
  \addplot table[x=n,y=error_avg,y error=error_std] {data/histogram_results_beta.txt};
  \addplot table[x=n,y=error_avg,y error=error_std] {data/histogram_results_gamma.txt};
  \legend{GMM,Beta,Gamma};
}{
}
\end{axis}
\end{tikzpicture}
\hfill
\begin{tikzpicture}
\begin{loglogaxis}[errorresultlog,ymax=1.0,ymin=0.01]
\iftoggle{plotexperiments}{
  \addplot table[x=n,y=error_avg,y error=error_std] {data/histogram_results_gmm.txt};
  \addplot table[x=n,y=error_avg,y error=error_std] {data/histogram_results_beta.txt};
  \addplot table[x=n,y=error_avg,y error=error_std] {data/histogram_results_gamma.txt};
  \legend{GMM,Beta,Gamma};
}{
}
\end{loglogaxis}
\end{tikzpicture}
\caption{Learning error for density estimation with histogram hypotheses.
The left plot shows the results on a linear scale, the right plot on a logarithmic scale.
The results clearly show that some distributions such as 2-GMMs are harder to approximate with 80-piecewise constant hypotheses than others.
Before the optimal learning error $\OPT_{80}$ dominates, our algorithm nevertheless demonstrates a quickly diminishing learning error.}
\label{fig:histerror}
\end{figure}

\subsection{Piecewise linear hypotheses}
Next, we turn our attention to the more challenging case of agnostically learning piecewise linear densities.
This is an interesting case because, in contrast to the histogram algorithm, 
the piecewise linear algorithm requires our full set of tools developed in Sections \ref{sec:outline} -- \ref{sec:seporacle}.
For the case of piecewise linear functions, the structure of the feasible set is still somewhat simpler than for general degree-\pdeg{} polynomials because the non-negativity constraint on a given interval can be encoded with two linear inequalities, i.e., the feasible set is a polytope instead of a spectrahedron.
We use this additional structure in our piecewise linear algorithm.
However, we did not implement further potential optimizations and resorted to an off-the-shelf linear program (LP) solver (GLPK, the GNU Linear Programming Kit) instead of a customized LP solver.
We believe that the running time of our algorithm can be improved further by implementing a custom LP solver that better utilizes the structure and small size of our LPs (and also takes into account that we solve many such small LPs).

We repeat the same experimental procedure as for piecewise histogram hypotheses, but use 40 linear pieces this time.
Figure \ref{fig:lineartime} contains the running time results of our algorithm.
Again, the results show three important points:
(i) As predicted, the running time scales nearly-linearly with $n$.
(ii) In spite of using an off-the-shelf LP solver, the constant factor in our running time is still good.
In particular, our algorithm requires less than 0.3 seconds for $10^6$ samples.
This is only three times slower than the time required for sorting the samples.
We believe that with a customized LP solver, we can bring this overhead down to a factor closer to two.
(iii) Again, the running time of our algorithm is very robust and does not depend on the shape of the unknown distribution.

Next, we consider the learning error achieved by our piecewise-linear algorithm, which is displayed in Figure \ref{fig:linearerror}.
Compared with the plots for piecewise constant hypotheses above, the results show that piecewise linear hypotheses can approximate the unknown densities significantly better, especially for the case of the $2$-GMM.
Three points are worth noting:
(i) The slope of the curve in the log-scale plot is about $-0.477$.
Note that this matches the $\frac{1}{\eps^2}$ term in our learning
error guarantee $O(\frac{t\cdot (\pdeg + 1)}{\eps^2})$ almost perfectly. 
(ii) Moreover, the constant factor achieved by our algorithm is close to $1$.
In particular, the learning error for the $2$-GMM and $n=10^6$ samples is roughly $0.00983$.
Using this as $\eps = 0.00983$ together with $t=40$ and $\pdeg = 1$ in $\frac{t \cdot (\pdeg + 1)}{\eps^2}$ gives about 830,000, 
which almost matches the $n=10^6$ samples for which this error was obtained.
(iii) The learning error of our algorithm is robust and essentially independent of the underlying distribution.

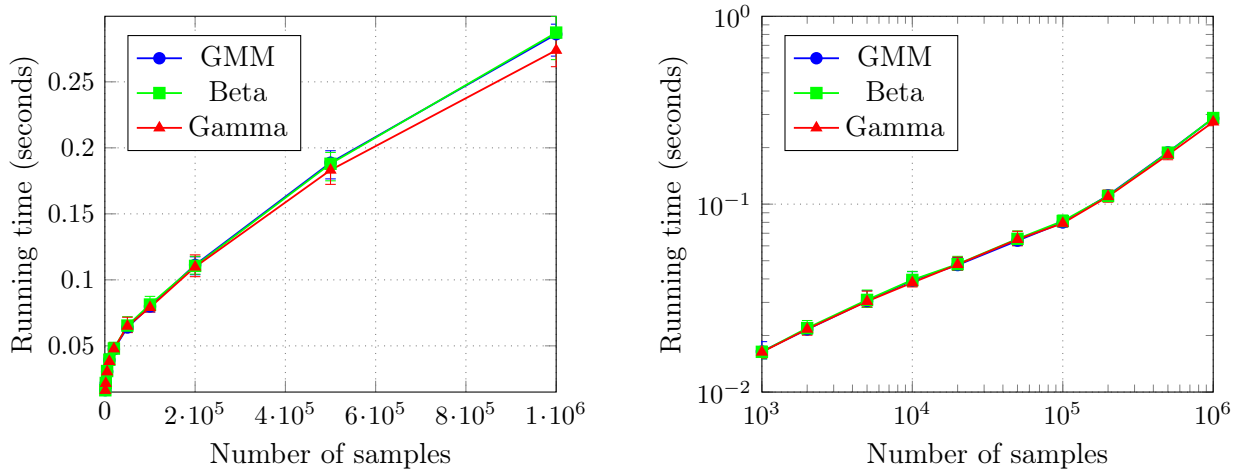
\begin{figure}[htb!]
\centering
\begin{tikzpicture}[/pgf/number format/sci generic={mantissa sep={\!\cdot\!},exponent={10^{#1}}}]
\begin{axis}[timeresultlinear]
\iftoggle{plotexperiments}{
  \addplot table[x=n,y=time_avg,y error minus=time_error_low, y error plus=time_error_high] {data/piecewise_linear_results_gmm.txt};
  \addplot table[x=n,y=time_avg,y error minus=time_error_low, y error plus=time_error_high] {data/piecewise_linear_results_beta.txt};
  \addplot table[x=n,y=time_avg,y error minus=time_error_low, y error plus=time_error_high] {data/piecewise_linear_results_gamma.txt};
  \legend{GMM,Beta,Gamma};
}{
}
\end{axis}
\end{tikzpicture}
\hfill
\begin{tikzpicture}
\begin{loglogaxis}[timeresultlog]
\iftoggle{plotexperiments}{
  \addplot table[x=n,y=time_avg,y error minus=time_error_low, y error plus=time_error_high] {data/piecewise_linear_results_gmm.txt};
  \addplot table[x=n,y=time_avg,y error minus=time_error_low, y error plus=time_error_high] {data/piecewise_linear_results_beta.txt};
  \addplot table[x=n,y=time_avg,y error minus=time_error_low, y error plus=time_error_high] {data/piecewise_linear_results_gamma.txt};
  \legend{GMM,Beta,Gamma};
}{
}
\end{loglogaxis}
\end{tikzpicture}
\caption{Running times for density estimation with piecewise-linear hypotheses.
The left plot shows the results on a linear scale, the right plot on a logarithmic scale.
As predicted by our analysis, the running time of our algorithm scales nearly-linearly with the input size $n$.
Moreover, the constant in the big-$O$ is quite small: for $n=10^6$, our algorithm takes less than 0.3 seconds, which is only three times slower than sorting the samples.
Note that this means that no algorithm that relies on sorting the samples can be more than 4 times faster than our algorithm when the total running time with sorting is taken into account.
As before, the running time of our algorithm is also essentially independent of the unknown distribution.}
\label{fig:lineartime}
\end{figure}

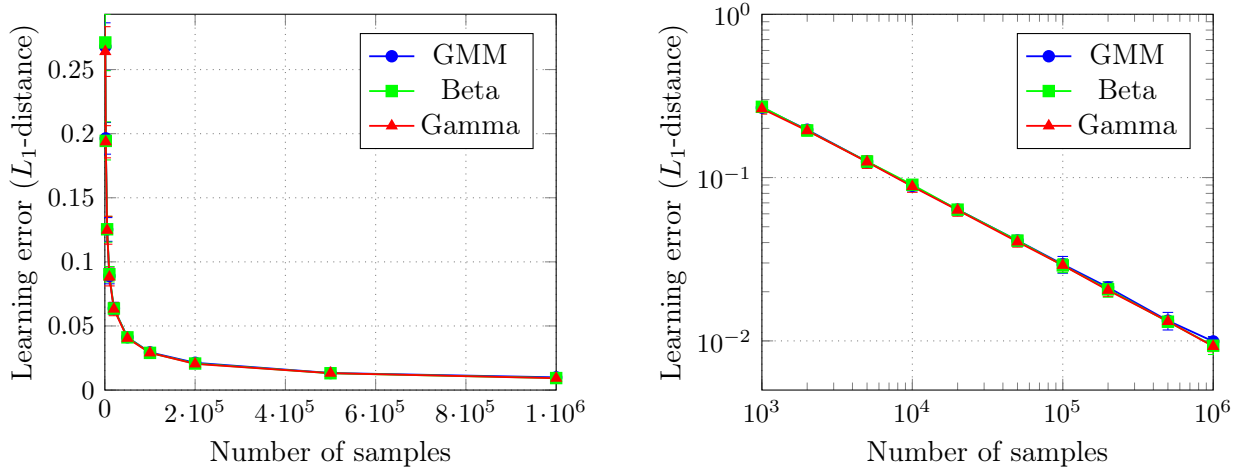
\begin{figure}[htb!]
\centering
\begin{tikzpicture}[/pgf/number format/sci generic={mantissa sep={\!\cdot\!},exponent={10^{#1}}}]
\begin{axis}[errorresultlinear]
\iftoggle{plotexperiments}{
  \addplot table[x=n,y=error_avg,y error=error_std] {data/piecewise_linear_results_gmm.txt};
  \addplot table[x=n,y=error_avg,y error=error_std] {data/piecewise_linear_results_beta.txt};
  \addplot table[x=n,y=error_avg,y error=error_std] {data/piecewise_linear_results_gamma.txt};
  \legend{GMM,Beta,Gamma};
}{
}
\end{axis}
\end{tikzpicture}
\hfill
\begin{tikzpicture}
\begin{loglogaxis}[errorresultlog]
\iftoggle{plotexperiments}{
  \addplot table[x=n,y=error_avg,y error=error_std] {data/piecewise_linear_results_gmm.txt};
  \addplot table[x=n,y=error_avg,y error=error_std] {data/piecewise_linear_results_beta.txt};
  \addplot table[x=n,y=error_avg,y error=error_std] {data/piecewise_linear_results_gamma.txt};
  \legend{GMM,Beta,Gamma};
}{
}
\end{loglogaxis}
\end{tikzpicture}
\caption{Learning error for density estimation with piecewise-linear hypotheses.
The left plot shows the results on a linear scale, the right plot on a logarithmic scale.
The slope of the curve in the log-scale plot is roughly $-.477$, which almost exactly matches the asymptotic guarantee for our algorithm.
Moreover, the average learning error for $2$-GMMs with $n=10^6$ samples is about $0.00983$.
Substituting this into the theoretical guarantee $\frac{t \cdot (\pdeg + 1)}{\eps^2}$ gives a sample requirement of roughly 830,000, i.e., very close to the $10^6$ samples our algorithm required to achieve this error.
Similar to the running time, the learning error is also robust and essentially independent of the underlying distribution.
}
\label{fig:linearerror}
\end{figure}

\section*{Acknowledgements}
We thank Chinmay Hegde for his contributions to the early stages of this work.
We would like to thank Yin Tat Lee and Aaron Sidford for useful discussions, 
and Richard Samworth for his help with the statistics literature.

\bibliographystyle{alpha}

\bibliography{allrefs}

\appendix

\section*{Appendix}

\section{Analysis of the General Merging Algorithm: Proof of
  Theorem~\ref{thm:generalmain}}
\label{app:general-merging-analysis}

This section is dedicated to the proof of Theorem~\ref{thm:generalmain}.
The proof is a generalization of that of Theorem~\ref{thm:mergingmain}.  
Recall the statement of Theorem~\ref{thm:generalmain}:

\generalmain*

\begin{proof}
We first bound the running time. 
The number of iterations of the algorithm is
$O(\log(n/\alpha t))$ by the same argument as for histograms, since the
number of intervals reduces by a factor of $3/4$ in each iteration. 
In each iteration, we compute the closest function in $\mathcal{D}$ and the corresponding
$\mathcal{A}_{\pdeg+1}$ distance, hence the runtime per iteration is bounded by $R_p(n)+R_c(n)$, by
definition. 

We now prove the error guarantee. 
Let $\setI = \{I_1, \ldots, I_{t'}\}$ be the partition of $I$ returned
by \textsc{General-Merging}, and let $h$ be the function returned. 
The desired bound on $t'$ is immediate since the
algorithm terminates only when $t' \leq 2 \alpha t$. 
We now prove~\eqref{eq:generalmergingmain}. 

Let $\hopt \in \mathcal{D}_t$ be such that $\|\hopt - f\|_1 = \OPT_{\mathcal{D}, t}.$
Let $\setI^{\ast} = \{I^{\ast}_1, \ldots, I^{\ast}_t\}$ be a partition
with at most $t$ pieces such that $\hopt\in\mathcal{D}_{I^\ast_i}$ for
all $i$.  
Call the end-points of $I^\ast_j$'s as \emph{jumps} of $\hopt$.
For any interval $J \subseteq I $ let $\Gamma (J)$ be the
number of jumps of $\hopt$ in the interior of $J$.
Since we draw $n = \Omega((\alpha \pdeg t + \log 1 / \delta) / \eps^2)$
samples, Corollary~\ref{cor:vc_whp} implies that with probability at
least $1-\delta$,
\[ \| \fhat - f\|_{\mathcal{A}_{(2 \alpha + 1) (\pdeg +1) t}} \le \eps \; . \] 
We condition on this event throughout the analysis.

\noindent We split the total error into three terms based on the final partition $\setI$:
\begin{description}
\item[Case 1:] Let $\mathcal{F}$ be the set of intervals  in $\setI$
  with no jumps in $\hopt$, i.e., $\mathcal{F} =
  \{J \in \setI \, | \, \Gamma(J) = 0\}$. 
\item[Case 2a:] Let $\mathcal{J}_0$ be the set of intervals in $\setI$
  that were created in the initial partitioning step of the algorithm
  and contain a jump of $\hopt$,  
i.e., $\mathcal{J}_0 = \{J \in \setI \mid \Gamma(J) > 0 \mbox{ and } J \in \setI_0 \}$.
\item[Case 2b:] Let $\mathcal{J}_1$ be the set of intervals in
    $\setI$ that contain at least one jump, and were created by merging two other intervals,
i.e., $\mathcal{J}_1 = \{J \in \setI \mid  \Gamma(J) >0 \text{ and } J \notin \setI_0 \}$.
\end{description}
Notice that $\mathcal{F}, \mathcal{J}_0, \mathcal{J}_1$ form a partition of $I$, and thus
\[ \| h - f \|_1 = \| h - f \|_{1, \mathcal{F}} + \| h - f \|_{1,
  \mathcal{J}_0} + \| h - f \|_{1, \mathcal{J}_1} \; .\] 
 
We bound the error from above in the three cases separately. In
particular, we will show:
\begin{align}
 \| h - f \|_{1, \mathcal{F}} &\leq 3 \cdot \| f - \hopt \|_{1, \mathcal{F}}  + 2\cdot\|
 \fhat - f\|_{\mathcal{A}_{|\mathcal{F}| \cdot(\pdeg +1)}, \mathcal{F}} +
 \frac{\eta}{2 \alpha t} |\mathcal{F}|\;, \label{eq:case1-general} \\ 
\| h - f \|_{1, \mathcal{J}_0} &\leq\| \fhat -
f\|_{\mathcal{A}_{|\mathcal{J}_0|\cdot(m+1)}, \mathcal{J}_0}
\;, \label{eq:case2a-general}\\ 
 \| h - f \|_{1, \mathcal{J}_1} &\leq \frac{ \OPT_{\mathcal{D}, t} + \eps}{(\alpha
          - 1)} +\| \fhat - f
        \|_{\mathcal{A}_{\pdeg \cdot t +
            |\mathcal{J}_1|}, \mathcal{J}_1} + \| f -\hopt
        \|_{1,\mathcal{J}_1} +
 \frac{\eta}{2(\alpha-1) }  \;. \label{eq:case2b-general} 
\end{align}
Using these results along with the fact that $\| f - \hopt \|_{1,
  \mathcal{F}} + \| f - \hopt \|_{1, \mathcal{J}_1} \leq \OPT_{\mathcal{D}, t}$ and
$\alpha>2$, we 
have 
\begin{align*}
\| h - f\|_1 &\leq 3 \cdot \OPT_{\mathcal{D}, t} + \frac{ \OPT_{\mathcal{D}, t} +
  \eps}{\alpha - 1} +  2\| \fhat - f\|_{\mathcal{A}_{|\mathcal{F}| (\pdeg +1)}}
+ \| \fhat - f\|_{\mathcal{A}_{|\mathcal{J}_0| \pdeg }} \\
  &\qquad\qquad + \| \fhat - f\|_{\mathcal{A}_{(|\mathcal{J}_1| + t) \pdeg }} + \frac{\eta}{2 \alpha t}
(|\mathcal{F}| + \mathcal{J}_1 )\\ 
&\stackrel{(a)}{\leq} 3 \cdot \OPT_{\mathcal{D}, t} + 
\frac{\OPT_{\mathcal{D}, t}+\eps}{\alpha - 1} +  2\| \fhat - f\|_{\mathcal{A}_{2\alpha t (\pdeg +1)}} + \eta \\ 
& \stackrel{(b)}{\leq} 3 \cdot \OPT_{\mathcal{D}, t} 
+ \frac{ \OPT_{\mathcal{D}, t}+ \eps}{\alpha - 1} + 2\eps + \eta \;, 
\end{align*}
where $(a)$ follows from Fact~\ref{lem:Ak-basics}(d) and
since $(|\mathcal{F}| + |\mathcal{J}_1|+|\mathcal{J}_0|) \leq 2 \alpha t$, and
$(b)$ follows from the VC inequality. 
Thus, it suffices to prove Equations
(\ref{eq:case1-general})--(\ref{eq:case2b-general}). 

\paragraph{Case 1.} We first consider the set of intervals in
$\mathcal{F}$. By the triangle inequality we have  
\[\| h - f \|_{1, \mathcal{F}} \leq \| f - \hopt \|_{1, \mathcal{F}} +
\| h - \hopt \|_{1, \mathcal{F}} \; .\]
For any interval $J\in\mathcal{F}$, since $h$ and $\hopt$ are both in
$\mathcal{D}$, they have at most $\pdeg$ sign changes, and 
\[
\| h - \hopt \|_{1, J} = \| h - \hopt \|_{\mathcal{A}_{\pdeg +1}, J}
\le  \| h - \fhat \|_{\mathcal{A}_{\pdeg +1}, J}+\| \fhat-\hopt \|_{\mathcal{A}_{\pdeg +1}, J}.
\]
By the definition of $h$ and the projection oracle, 
\[
\| h - \fhat \|_{\mathcal{A}_{\pdeg +1}, J} 
\le\min_{h'\in\mathcal{D}_J} \| h' - \fhat \|_{\mathcal{A}_{\pdeg +1}, J}+
\frac{\eta}{2\alpha t}
\le \| \hopt - \fhat \|_{\mathcal{A}_{\pdeg +1}, J}+
\frac{\eta}{2\alpha t}.
\]
Therefore, 
\[
\| h - \hopt \|_{1, J} \le2\cdot\| \hopt - \fhat \|_{\mathcal{A}_{\pdeg +1},
  J}+\frac{\eta}{2\alpha t}. 
\]
Again by the triangle inequality,
\[
\| \hopt - \fhat \|_{\mathcal{A}_{\pdeg +1},J}\le \| \hopt - f
\|_{\mathcal{A}_{\pdeg +1},J}+\| f - \fhat \|_{\mathcal{A}_{\pdeg +1},J}. 
\]
Summing over the intervals in $\mathcal{F}$, 
\begin{align*}
\sum_{J\in\mathcal{F}}\| \hopt - \fhat \|_{\mathcal{A}_{\pdeg +1},J}
&\le \sum_{J\in\mathcal{F}}\| \hopt - f
\|_{\mathcal{A}_{\pdeg +1},J}+\sum_{J\in\mathcal{F}}\| f - \fhat
\|_{\mathcal{A}_{\pdeg +1},J}\\
&\le\| \hopt - f
\|_{1,\mathcal{F}}+\| f - \fhat
\|_{\mathcal{A}_{|\mathcal{F}|(\pdeg +1)},\mathcal{F}}
\end{align*}
Combining these, we obtain, 
\[
\| h - f \|_{1, \mathcal{F}} \leq 3\cdot\| f - \hopt \|_{1, \mathcal{F}} +
2\cdot\| f - \fhat \|_{\mathcal{A}_{|\mathcal{F}|(\pdeg +1)},\mathcal{F}}
+\frac{\eta}{2\alpha t}|\mathcal{F}|\;,
\]
which is precisely~\eqref{eq:case1-general}.

\paragraph{Case 2a.}
We now analyze the error for the intervals $\mathcal{J}_0$.
The set $\setI_0$ contains only singletons and intervals with no
  sample points. By definition, with probability 1, only the intervals in $\setI_0$
  that contain no samples may contain a jump of $\hopt$.
The singleton intervals containing the sample points do not include
jumps, and are hence covered by Case 1. 
Since $\mathcal{J}_0$ does not contain any samples, our algorithm
assigns
\[
h(J) = \fhat (J) = 0
\]
for any $J\in\mathcal{J}_0$. 
Hence,
\begin{equation*}
  \norm{h - f}_{1, \mathcal{J}_0} = \norm{f}_{1, \mathcal{J}_0} \; ,
\end{equation*}
and
\begin{align*}
\norm{h-f}_{1,\mathcal{J}_0} &=\norm{f}_{1, \mathcal{J}_0} \\
                        &= \; \sum_{J \in \mathcal{J}_0} |f(J)| \\
                         &= \; \sum_{J \in \mathcal{J}_0}  |f(J) - \fhat (J)|   \\
                         &\leq \; \norm{f -
                           \fhat}_{\mathcal{A}_{|\mathcal{J}_0|(\pdeg +1)},
                           \mathcal{J}_0} \;, 
\end{align*}
where the last step simply follows from non-negativity of $f-\fhat$
over $\mathcal{J}_0$. 

\paragraph{Case 2b.} 
We finally consider $\mathcal{J}_1$, the set of  
intervals created by merging in some iteration of our
algorithm that also contain jumps. 
As before, our first step is the following triangle inequality:
\[\| h - f \|_{1, \mathcal{J}_1} \leq \| h - \hopt \|_{1, \mathcal{J}_1} + \| \hopt - f \|_{1, \mathcal{J}_1} \; .\]

Consider an interval  $J \in \mathcal{J}_1$ with $\Gamma(J)\ge1$ jumps of
$\hopt$. Since $h\in\mathcal{D}_J$, $h - \hopt$ has at most $\pdeg \cdot\Gamma(J)$
sign changes in $J$. Therefore, 
\begin{align*}
\| h - \hopt \|_{1, J} &\stackrel{(a)}{=} \| h - h^* \|_{\mathcal{A}_{\pdeg \cdot\Gamma (J) + 1}, J}  \nonumber \\
	&\stackrel{(b)}{\leq} \| h - \fhat \|_{\mathcal{A}_{ \pdeg \cdot\Gamma(J) + 1}, J} + \| \fhat - f\|_{\mathcal{A}_{\pdeg \cdot\Gamma(J) + 1}, J} + \| f - \hopt \|_{\mathcal{A}_{\pdeg \cdot\Gamma(J) + 1}, J} \nonumber \\
	&\stackrel{(c)}{\leq} \Gamma(J) \| h - \fhat
        \|_{\mathcal{A}_{\pdeg +1}, J} +  \| \fhat -
        f\|_{\mathcal{A}_{\pdeg \cdot\Gamma(J) + 1}, J}  + \| f - \hopt
        \|_{1, J} \numberthis \label{eqn:five-general} \;, 
\end{align*}
where $(a)$ follows from Fact~\ref{lem:Ak-basics}(a), $(b)$ is the
triangle inequality, and  
inequality $(c)$ uses Fact~\ref{lem:Ak-basics}(c) along with the fact
that $\Gamma(J)\ge1$ and $\pdeg \ge1$. 
We start by bounding the $\mathcal{A}_{\pdeg+1}$ distance in the first term above. 
\begin{lemma} \label{lem:2b-general}
For any $J \in \mathcal{J}_1$, we have
\begin{equation}
\| h - \fhat \|_{\mathcal{A}_{\pdeg +1}, J} \leq \frac{\OPT_{\mathcal{D}, t} + \eps}{(\alpha - 1) t} + \frac{\eta}{2 (\alpha-1)
  t}\;. \label{eq:case2bkeylemma-general}  
\end{equation}
\end{lemma}
Before proving this lemma, we use it to complete Case 2b. 
Summing \eqref{eqn:five} over $J\in\mathcal{J}_1$ and plugging in the
lemma,
\begin{align*}
\| h - \hopt \|_{1, \mathcal{J}_1} 
	&\leq \left( \sum_{J \in \mathcal{J}_1} ( \Gamma(J)  \right)\cdot\left(\frac{ \OPT_{\mathcal{D}, t} + \eps}{(\alpha -
          1) t}+\frac{\eta}{2 (\alpha-1) t}\right) +\sum_{J \in
        \mathcal{J}_1} \| \fhat - f
      \|_{\mathcal{A}_{\pdeg \cdot\Gamma(J)+1}, J} + \| f -\hopt 
        \|_{1,\mathcal{J}_1} \\ 
	& \stackrel{(a)}{\leq} \frac{ \OPT_{\mathcal{D}, t} +  \eps}{(\alpha
          - 1)} +\frac{\eta}{2(\alpha-1)} +\| \fhat - f
        \|_{\mathcal{A}_{\pdeg \cdot t +
            |\mathcal{J}_1|}, \mathcal{J}_1} + \| f -\hopt
        \|_1,{\mathcal{J}_1} 
\end{align*} 
where the first term in $(a)$ uses 
the fact that
$\sum_{J \in \mathcal{J}_1} \Gamma(J)  \le t$ 
and the second term uses this in conjunction with
Fact~\ref{lem:Ak-basics}(d). 

We now prove Lemma~\ref{lem:2b-general}.
\begin{proof}[Proof of Lemma~\ref{lem:2b-general}]
Each iteration of our algorithm merges pairs of intervals
except those with the $\alpha t$ largest errors. Therefore, if two
intervals were merged, there were at least $\alpha t$ other interval
pairs with larger error. We will use this fact to bound the error on
the intervals in $\mathcal{J}_1$.  

Suppose an interval $J \in \mathcal{J}_1$ was
created in the $j$th iteration of the while loop of our algorithm, i.e.,  
$J = I'_{i, j+1} = I_{2i-1, j} \cup I_{2i, j}$ for some $i \in \{1, \ldots, s_j/2\}.$
Recall that the intervals $I'_{i, j+1}$, for $i \in  \{1, \ldots,
s_j/2\}$, are the candidates for merging at iteration $j$. 
Let $h'$ be the distribution given by applying the projection oracle
to the empirical distribution over each candidate interval
$\setI'_{j+1} = \{ I'_{1,j+1}, \ldots, I'_{s_j / 2, j+1}\}$. Note that
$h'(x) = h(x)$ for $x \in J$ since $J$ remains intact through the
remainder of the algorithm. 

As with the histogram estimation, for a class $\mathcal{D}$ with at
most $\pdeg$ sign changes, let $e_\pdeg (g,
J)=\min_{g'\in\mathcal{D}_J}\|g-g'\|_{A_{\pdeg +1}}$.  
Let ${\cal L}$ be the set of candidate intervals $I'_{i, j +1}$ in the set $\setI'_{j+1}$ 
with the largest $\alpha \cdot t$ errors
$\|h'-\fhat\|_{\mathcal{A}_{\pdeg +1}}$. By the guarantee of projection
oracle,  
\[
\|h'-\fhat\|_{\mathcal{A}_{\pdeg +1}, I'_{i, j+1}}\le e_\pdeg (\fhat, I'_{i, j+1})+\frac{\eta}{2\alpha t}. 
\]
Let ${\cal L}_0$ be the intervals in $\cal L$ that do 
not contain any  jumps of $\hopt$. Since $\hopt$ has at most $t$
jumps, $\abs{{\cal L}_0} \geq (\alpha - 1) t$.  

Therefore, 
\begin{align*}
\sum_{I'\in{\cal L}_0}\|h'-\fhat\|_{\mathcal{A}_{\pdeg +1}, I'}&\le 
\sum_{I'\in{\cal L}_0}\left(e_\pdeg (\fhat, I')+\frac{\eta}{2\alpha
  t}\right)\\
&\le\sum_{I'\in{\cal L}_0}\left(\|\hopt-\fhat\|_{\mathcal{A}_{\pdeg +1}, I'}+\frac{\eta}{2\alpha
  t}\right)\\
&\le \|f-\hopt\|_{1,{\cal
    L}_0}+\|f-\fhat\|_{\mathcal{A}_{(\pdeg +1)\alpha t}, {\cal
    L}_0}+\eta/2\\
&\le \OPT_{\mathcal{D}, t} + \eps+\eta/2.
\end{align*}
Since $h'$ is $h$ on the interval $J$, combining with
$\abs{\mathcal{L}_0} \geq (\alpha - 1) t$, we obtain 
\begin{align*}
\norm{h' - \fhat}_{\mathcal{A}_{\pdeg +1}, J}  =  \norm{h - \fhat}_{\mathcal{A}_{\pdeg +1}, J} \;
                              \leq \; \frac{ \OPT_{\mathcal{D}, t}+ 2\eps}{(\alpha - 1) t}+\frac{\eta}{2 (\alpha-1)
  t} \; ,
\end{align*}
completing the proof of the lemma. 
\end{proof}
\renewcommand{\qedsymbol}{}
\end{proof}

\section{Additional Omitted Proofs}
\label{app:proof-polynomials}
\subsection{Proof of Fact~\ref{lem:polybound2}}
\label{app:markov-proof}

We first require the following classical lemma, first proved by Markov~\cite{markov1892functions}.
For completeness, we include an elegant proof by the mathoverflow user
fedja\footnote{See
  \url{http://mathoverflow.net/questions/97769/approximation-theory-reference-for-a-bounded-polynomial-having-bounded-coefficie}}.
We remark that the bounds in the following fact are essentially tight.

\begin{fact}[\cite{markov1892functions}]
\label{lem:polybound}
Let $p(x) = \sum_{j = 0}^\pdeg c_j x^j$ be a degree-\pdeg{} polynomial so that $| p(x) | \leq 1$ for all $x \in [-1, 1]$. Then $\max_j |c_j| \leq (\sqrt{2} + 1)^\pdeg$ for all $j = 0, \ldots, \pdeg$.
\end{fact}
\begin{proof}
We first claim that $|c_j| \leq \max_{z \in \mathbb{D}} |p(z)|$ where $\mathbb{D}$ is the unit complex disc.
To see this, we notice that by Cauchy's integral formula,
\[
  c_j = \frac{1}{j!} p^{(j)}(0) = \frac{1}{2 \pi i} \int_{|\zeta| = 1} \frac{p(\zeta)}{\zeta^{j+1}} d \zeta \; ,
  \]
where we also changed the order of differentiation and integration and used
\[
  \frac{\diff}{\diff x^j} \frac{p(\zeta)}{\zeta - x} \; = \; \frac{j! \cdot p(\zeta)}{(\zeta - x)^{j+1}} \; .
\]
Therefore, we get
\begin{align*}
|c_j| &= \frac{1}{2 \pi} \left|  \int_{|\zeta| = 1} \frac{p(\zeta)}{\zeta^{j+1}} d \zeta \right| \\
&\leq \frac{1}{2 \pi} \int_{|\zeta| = 1} \left| \frac{p(\zeta)}{\zeta^{j+1}} \right| d \zeta \\
&\leq \max_{|\zeta| = 1} |p(z)| \;.
\end{align*}
Consider the function 
\[F(z) = z^{-m} p \left( \frac{z + z^{-1}}{2} \right) \;.\]
On the domain $\{z: |z| \geq 1 \}$, this function is analytic.
So by the maximum modulus principle, it is bounded by its value on the unit circle.
Since for all $z \in \mathbb{D}$, $(z + z^{-1}) / 2 = \Re(z)$, we conclude that $|F(z)| \leq \max_{x \in [-1, 1]} p(x) \leq 1$ by assumption.
Thus we have that
\[p \left( \frac{z + z^{-1}}{2} \right) \leq z^\pdeg\]
for all $|z| > 1$.
Fix any $w \in \mathbb{D}$.
It is straightforward to see that  $w = (z + z^{-1}) / 2$ for some $z \in \mathbb{C} \setminus \{0 \}$; by symmetry of $z$ and $z^{-1}$ we conclude that this also holds for some $z$ with $|z| \geq 1$.
For each $w$, arbitrarily choose such a $z$ and denote it $z_w$.
Moreover, for all $|z| > (\sqrt{2} + 1)$, we have 
\begin{align*}
\left| \frac{z + z^{-1}}{2} \right| &\geq \frac{|z| - |z^{-1}| }{2} \\
&> \frac{\sqrt{2} + 1 - \frac{1}{\sqrt{2} + 1} } {2} \geq 1 \\
\end{align*}
and thus we conclude that for all $w \in \mathbb{D}$ we have that its corresponding $z_w$ satisfies $|z_w| \leq \sqrt{2} + 1$ and therefore $|p(w)| = |p((z_w + z_w^{-1}) / 2)| \leq z_w^\pdeg \leq (\sqrt{2} + 1)^\pdeg$, as claimed.
\end{proof}

The above statement is for polynomials that are uniformly bounded on $[-1, 1]$.
We will be interested in bounds for polynomials that integrate to a fixed constant.
In order to relate these bounds, we use the following classical result.
\begin{fact}[Bernstein's Inequality \cite{Cheney1982}]
\label{fact:bernstein}
Let $p$ be a degree-$\pdeg$ polynomial and let $p'$ be its derivative.
Then
\[
  \max_{x \in [-1, 1]} \abs{p'(x)} \; \leq \; \pdeg^2 \cdot \max_{x \in [-1, 1]} \abs{p(x)} \; .
\]
\end{fact}

With these results, we are now ready to prove
Lemma~\ref{lem:polybound2}.
\begin{proof}[Proof of Lemma~\ref{lem:polybound2}]
Consider the degree-$(\pdeg + 1)$ polynomial $P$ such that $P(-1) = 0$ and $P' = p$.
This implies that $P(x) = \int_{-1}^x p(y) \diff y$.
Since $p$ is non-negative on $[-1, 1]$, the bound on $\int_{-1}^1 p(y)\diff y$ then gives
\[
  \max_{x \in [-1, 1]} \abs{P(x)} \; \leq \; \alpha \cdot (\sqrt{2} + 1)^\pdeg\; .
\]
Using Bernstein's Inequality (Fact~\ref{fact:bernstein}), we can convert this bound into a bound on $P' = p$, i.e., we get that $\abs{p(x)} \leq t\cdot (\pdeg +1)$ for all $x \in [-1, 1]$.
Combining this uniform bound on $p$ with Fact~\ref{lem:polybound} gives the desired bound on the coefficients of $p$.
\end{proof}

\subsection{Proof of Lemma~\ref{prob:nonneg}}
\label{app-non-neg}
Our approach to proving Lemma~\ref{prob:nonneg} is relatively straightforward.
Assume we had an algorithm $A$ that finds the roots of $p$ exactly.
Then one could perform a non-negativity test by running $A$ to find the roots of $p'$, which correspond to the extrema of $p$.
Given the extrema of $p$, it suffices to check whether $p$ is non-negative at those points and the endpoints of the interval.

However, such an exact root-finding algorithm $A$ does not exist in general.
Nevertheless, there are efficient algorithms for finding the approximate roots of $p$ in certain regimes.
We leverage these results to construct an efficient non-negativity test.
Before we proceed, we remark briefly that we could also utilize the
univariate SOS
algorithm~\cite{shor1987class,lasserre2001global,parrilo2003semidefinite},
which is arguably more elementary than our approach here, but slower.

Formally, we build on the following result.

\begin{fact}[\cite{Pan2001}, Part II, Theorem 1.1]
\label{thm:fast-roots}
Let $\mathbb{D}$ denote the complex unit disc.
For all $\nu > 0$, there exists an algorithm $\textsc{FindRoots}(q, \beta)$ satisfying the following guarantee:
given any degree-\pdeg{} polynomial $q(z) : \mathbb{C} \to \mathbb{C}$ with roots $z_1 \ldots, z_\pdeg$ such that $z_i \in \mathbb{D}$ for all $i$ and $\beta \geq \pdeg \log \pdeg$, returns $z^*_1, \ldots, z^*_\pdeg$ so that $|z^*_j - z_j| \leq 2^{2 - \beta / \pdeg}$ for all $j$.
Moreover, \textsc{FindRoots} runs in time $O(\pdeg \log^2 \pdeg \cdot (\log^2 \pdeg + \log \beta))$.
\end{fact}

Our polynomials do not necessarily have all roots within the complex unit disc.
Moreover, we are only interested in real roots.
However, it is not too hard to solve our problems with the algorithm from Fact \ref{thm:fast-roots}.
We require the following structural result:

\begin{fact}[\cite{Henrici1974}, Sect. 6.4]
\label{thm:henrici}
Let $q(x) = x^\pdeg + c_{\pdeg - 1} x^{\pdeg - 1} + \ldots + c_1 x + c_0$ be a monic polynomial of degree $\pdeg$ (i.e., the leading coefficient is $1$). Let $\rho (q)$ denote the norm of the largest zero of $q$. Then 
\[\rho(q) \leq 2 \max_{1 \leq i \leq \pdeg} |c_{\pdeg - i}|^{1 / i} \; .\]
\end{fact}

In order to use the result above, we process our polynomial $p$ so that it becomes monic and still has bounded coefficients.
We achieve this by removing the leading terms of $p$ with small coefficients.
This then allows us to divide by the leading coefficient while increasing the other coefficients by a controlled amount only.
Formally, we require the following definitions.

\begin{definition}[Truncated polynomials]
\label{def:truncatedpoly}
For any degree-$\pdeg$ polynomial $p = \sum_{i = 0}^\pdeg c_i x^i$ and $\nu > 0$ let 
\[\Delta = \Delta(p, \nu) = \max \left\{i : |c_i| \geq  \frac{\nu}{2 \pdeg}\right\} \;,\] 
and let $\Pi = \Pi_{\nu}$ be the operator defined by 
\[(\Pi p) (x) = \sum_{i = 0}^{\Delta(p, \nu)} c_i x^i.\] 
\end{definition}
Formally, $\Pi$ acts on the formal coefficient representation of $p$ as $q = \sum c_i x^i$. 
It then returns a formal representation $\sum_{i = 0}^{\Delta(p, \nu)} c_i x^i$.
In a slight abuse of notation, we do not distinguish between the formal coefficient representation of $p$ and the polynomial itself. 
Then Facts \ref{thm:fast-roots} and \ref{thm:henrici} give us the following:

\begin{lemma}
\label{cor:FastApproxRoots}
There exists an algorithm $\textsc{FastApproxRoots}(p, \nu, \mu)$ with the following guarantee.
Let $p$ be a polynomial as in Definition \ref{prob:nonneg}, and let $\nu, \mu >0$ such that $\nu \leq \frac{1}{2 \alpha \pdeg}$ (where $\alpha$ and $\pdeg$ are as in Def.\ \ref{prob:nonneg}).
Then \textsc{FastApproxRoots} returns approximate roots $x^*_1, \ldots, x^*_{\Delta(p, \nu)} \in \mathbb{R}$ so that for all real roots $y$ of $\Pi_{\nu} p$, there is some $j$ so that $|y - x^*_j| \leq \mu$.
Moreover, \textsc{FastApproxRoots} runs in time $O(\pdeg \log^2 \pdeg \cdot (\log^2 \pdeg + \log \log \alpha + \log \log (1/\nu) + \log \log (1 / \mu)))$.
\end{lemma}
\begin{proof}
$\textsc{FastApproxRoots}(p, \nu, \mu)$ proceeds as follows.
We find $\Delta = \Delta(p, \nu)$ and $\Pi p = \Pi_{\nu} p$ in time $O(\pdeg)$ by a single scan through the coefficients $c_i$ of $p$. Let $q_1(x) = \frac{1}{c_\Delta}  (\Pi p) (x)$. Note that the roots of $q_1$ are exactly the roots of $\Pi p$. Then, by Theorem \ref{thm:henrici}, we have that 
\[ A  \ed 2 \max_{1 \leq i \leq \Delta} \left| \frac{c_{\Delta - i}}{c_{\Delta}} \right|^{1 / i} \geq \rho(q_1) \; .\]
The quantity $A$ is also simple to compute in a single scan of the $c_i$. Notice that we have 
\[A \leq \max \left( 2 \max_{1 \leq i \leq \Delta} \left| \frac{c_{\Delta - i}}{c_{\Delta}} \right|, 1 \right) \leq \underbrace{\frac{2 \alpha \pdeg}{\nu}}_{B}\]
by the definition of $\Delta$ and the assumption that the $c_i$ are bounded by $\alpha$ (Definition \ref{prob:nonneg}).
Let $B$ denote the right hand side of the expression above. If we let $q(x) = q_1(A x)$, we have that the roots of $q$ all lie within the complex unit disc. Let $z_1, \ldots, z_{\Delta}$ be the roots of $\Pi p$. Then the roots of $q$ are exactly $z_1 / A, \ldots, z_{\Delta} / A$.
Run $\textsc{FindRoots}(q, 2\pdeg + \pdeg \log B + \pdeg \log (1 / \mu))$, which gives us $z_1^*, \ldots, z_{\Delta}^*$ so that for all $i$, we have $|z_i^* - z_i / A| < \mu / B$. Thus, for all $i$, we have
\[\left| A z_i^* - z \right| \leq A \frac{\mu}{B} \leq \mu \; . \]
 $\textsc{FastApproxRoots}(p, \nu, \mu)$ returns the numbers $x_i^* = \Re(A z_i^*)$. 
For any real root $x$ of $\Pi p$, there is some $z_i^*$ so that $|A z_i^* - x| < \mu$, and thus $| x_i^* - x| < \mu$ as well.
Thus, we output numbers which satisfy the conditions of the Lemma. 
Moreover, the runtime of the algorithm is dominated by the runtime of $\textsc{FindRoots}(q, 2\pdeg + \pdeg \log B + \pdeg \log (1 / \mu))$, which runs in time
\begin{align*}
O(\pdeg \log^2 \pdeg \cdot (\log^2 \pdeg + \log (& \pdeg \log B + \pdeg \log (1 / \mu)))) = \\
  &O(\pdeg \log^2 \pdeg \cdot (\log^2 \pdeg + \log \log \alpha + \log \log (1/\nu) + \log \log (1 / \mu))) \;
\end{align*}
This completes the proof.
\end{proof}

\begin{proof}[Proof of Lemma~\ref{prob:nonneg}]
Let $\nu =\frac{\mu}{2}$, and let  $\nu' =  \frac{\mu}{4 \alpha \pdeg (\pdeg + 1)}$.
Set 
\[
  r = (\Pi_{\nu} p)(x) = \sum_{i = 1}^{\Delta(p, \nu)} c_i x^{ i} \; .
\] 
We can compute the coefficients of $r$ in time $O(\pdeg)$.
Moreover, $\Pi (r'(x)) = r'(x)$.  Let $x_1, \ldots, x_{\pdeg'}$, where $\pdeg' \leq \Delta$, be the roots of $r'(x)$ in $[-1, 1]$. 
These points are exactly the local extrema of $r$ on $[-1, 1]$. 
Our algorithm $\textsc{TestNonneg} (p, \mu)$ then is simple:
\begin{enumerate}
\item Run $\textsc{FastApproxRoots}(r, \nu', \mu)$ and let $x_1^*, \ldots, x_{\Delta}^*$ be its output. 
\item Let $J = \{i: x_i^* \in [-1, 1] \}$ and construct the set $S = \{-1, 1\} \cup \{x_i: i \in J\}$. 
\item Denote the points in $S$ by $x_0 = -1 \leq x_1 \leq \ldots \leq x_{\pdeg' - 1} \leq x_{\pdeg'} = 1$, where $\pdeg' \leq \Delta + 1$.
\item Evaluate the polynomial $p$ at the points in $S$ using the fast algorithm from Fact \ref{fact:fastpolyeval}.
\item If at any of these points the polynomial evaluates to a negative number, return that point. Otherwise, return ``OK''.
\end{enumerate}
The running time is dominated by the call to \textsc{FastApproxRoots}.
By Lemma \ref{cor:FastApproxRoots}, this algorithm runs in time $O(\pdeg \log^2 \pdeg \cdot (\log^2 \pdeg + \log \log \alpha + \log \log (1 / \mu)))$ as claimed.

It suffices to prove the correctness of our algorithm. 
Clearly, if $p$ is nonnegative on $[-1, 1]$, it will always return ``OK''. 
Suppose there exists a point $y \in I$ so that $p(y) < -\mu$. 

For a function $f$, and an interval $I=[a,b]$, let $|f|_{\infty,
  I}=\sup_{x\in I}\|f(x)\|$. Then,
\begin{equation}
\label{eq:poly-unif-bound}
\| p - r \|_{\infty, [-1, 1]} \; \leq \; \sup_{x \in [-1, 1]} \left| \sum_{i = \Delta + 1}^\pdeg c_i x^i  \right| \; \stackrel{(a)}{\leq} \; (\pdeg - \Delta) \cdot \frac{\mu}{4 \pdeg}  \; \leq \; \mu / 4 \; ,
\end{equation}
where the inequality (a) follows from the choice of $\Delta$. 
Thus $r(y) < -3 \mu / 4$. Since the points $x_0, x_1, \ldots, x_{\pdeg'}$ are extremal for $r$ on $I$, there exists a $0 \leq j \leq \pdeg'$ so that $r(x_j) < - 3 \nu / 4$. 
If $j = 0$ (resp. $j = m'$), so if $r(-1) < -3 \mu / 4$ (resp. $r(1) < -3 \mu / 4$), then by Equation (\ref{eq:poly-unif-bound}), we have $p(-1) < \mu / 2$ (resp. $p(1) < -\mu / 2$).
Thus our algorithm correctly detects this, and the polynomial fails the non-negativity test as intended.

Thus assume $j \in \{ 1, \ldots, \Delta \}$. 
By Lemma \ref{cor:FastApproxRoots}, we know that there is a $x_\ell^*$ so that $|x_\ell^* - x_j| < \nu'$. Since $x_j \in I$, either $\ell \in J$ or $|x_j + 1| < \nu'$ or $|x_j - 1| < \nu'$, so in particular, there is a point $s \in S$ so that $|x_j - s| < \nu'$. 
Since for all $x \in [-1, 1]$, we have
\[|p'(x)| \leq \sum_{i = 1}^{\pdeg} \left| i c_i x^i \right| \leq \alpha \pdeg (\pdeg + 1) \]
by the bound on the coefficients of $p$ (see Definition \ref{prob:nonneg}).
By a first order approximation, we have that 
\[|p(x_j) - p(s)| \leq \alpha \pdeg (\pdeg + 1) |x_j - s| \leq \mu / 4\]
where the last inequality follows by the definition of $\nu'$.
Thus, we have that $p(s) < -\mu / 2$, and we will either return $s$ or some other point in $s' \in S$ with $p(s') \leq p(s)$.
Thus our algorithm satisfies the conditions on the theorem.
\end{proof}

\section{Learning discrete piecewise polynomials} \label{ap:discr}
Throughout this paper we focused on the case that the unknown distribution 
has a density $f$ supported on $[-1, 1]$, and that the error metric is the $L_1$-distance with respect to the Lebesgue measure on the real line.
We now show that our algorithm and analysis naturally generalize to the case of discrete distributions.

In the discrete  setting, the unknown distribution is supported on the set $[N] \eqdef \{1, \ldots, N\}$, 
and the goal is to minimize the $\ell_1$-distance between the corresponding probability mass functions.
The $\ell_1$-norm of a function $f: [N] \to \R$ is defined to be $\| f \|_1 = \sum_{i = 1}^N |f (i)|$ and the $\ell_1$-distance between $f, g: [N] \to \R$ is
$\| f -g\|_1.$

In the following subsections, we argue that our algorithm also applies to the discrete setting with only minor adaptations.
That is, we can agnostically learn discrete piecewise polynomial distributions with the same sample complexity and running time as in the continuous setting.

\subsection{Problem statement in the discrete setting}
Fix an interval $I \subseteq [N]$.
We say that a function $p: I \to \R$ is a degree-$\pdeg$ polynomial if there is a degree-$\pdeg$ real polynomial $q: \R \to \R$ such that $p(i) = q(i)$ for all $i \in I$.
We say that $h: [N] \to \R$ is a $t$-piecewise degree-$\pdeg$ polynomial 
if there exists a partition of $[N]$ into $t$ intervals so that on each interval, $h$ is a degree-$\pdeg$ polynomial.
Let $\setP_{t, \pdeg}^\disc$ be the set of $t$-piecewise degree-$\pdeg$ polynomials on $[N]$ which are nonnegative at every point in $[N]$.
Fix a distribution (with probability mass function) $f: [N] \to \R$. 
As in the continuous setting, define
$\OPT_{t, \pdeg}^\disc \ed \min_{g \in \setP_{t, \pdeg}^\disc} \norm{g - f}_1 \; .$
As before, our goal is the following: given access to $n$ i.i.d.\ samples from $f$, to compute a hypothesis $h$ so that probability at least $9/10$ 
over the samples, we have $\norm{h - f}_1 \leq C \cdot \OPT_{t, \pdeg}^\disc + \eps \; ,$
for some universal constant $C$.
As before, we let $\fhat$ denote the empirical after taking $n$ samples.

Our algorithms for the continuous setting also work for discrete distributions, albeit with slight modifications. For the case of histogram approximation, 
the algorithm and its analysis hold verbatim for the discrete setting. The only difference is in the definition of flattening; 
Definition~\ref{def:flattening} applies to continuous functions. For a function $f: [N] \to \R$ and an interval $J \subseteq [n]$ 
the flattening of $f$ on $J$ is now defined to be the constant function on $J$ which divides the total $\ell_1$ mass of the function within $J$ uniformly among all the points in $J$. Formally, if $J = \{a, \ldots,  b\}$, we define the flattening of $f$ on $J$ to be the constant function $\bar{f}_J (x) = \frac{\sum_{i \in I} f(i)}{b - a + 1}$.

\subsection{The algorithm in the discrete setting}
Our algorithm in the discrete setting is nearly identical to the algorithm in the continuous setting, 
and the analysis is very similar as well.
Here, we only present the high-level ideas of the discrete algorithm 
and highlight the modifications necessary to move from a continuous to a discrete distribution.

\subsubsection{The \texorpdfstring{$\Ak$}{Ak}-norm and general merging in the discrete setting}

We start by noting that the notion of the $\Ak$-norm and the VC inequality also hold in the discrete setting.
In particular, the $\Ak$-norm of a function $f : [N] \to \R$ is defined as
\[\| f \|_\Ak = \max_{I_1, \ldots, I_k} \sum_{i = 1}^k |f (I_i)| \; , \]
where the maximum ranges over all $I_1, \ldots, I_k$ which are disjoint sub-intervals of $[N]$.

The basic properties of the $\Ak$-norm (i.e., those in Lemma \ref{lem:Ak-basics}) still hold true.
Moreover, it is well-known that the VC inequality (Theorem \ref{thm:vc}) still holds in this setting.
These properties of the $\Ak$-norm are the only ones that we use in the analysis of \textsc{GeneralMerging}.
Therefore, it is readily verified that the same algorithm is still correct, 
and has the same guarantees in the discrete setting, 
assuming appropriate approximate $\Ak$-projection and $\Ak$-computation oracles for polynomials on a fixed subinterval of $[N]$.

\subsubsection{Efficient \texorpdfstring{$\Ak$}{Ak}-projection and computation oracles for polynomials}
We argue that, as in the continuous setting, we can give efficient $\Ak$-projection and computation oracles for non-negative polynomials of degree $\pdeg$ 
on a discrete interval $I$, using an $O(\pdeg)$-dimensional convex program.
By appropriately shifting the interval, we may assume without loss of generality that the interval is of the form $[m] = \{1, \ldots, m\}$ for some $m \leq N$.
\paragraph{The Convex Program}
As in the continuous case, it can be shown that the set of non-negative polynomials $p$ 
on $[m]$ satisfying $\| p - \fhat \|_\Ak \leq \tau$ is convex (as in Lemma \ref{lem:feasibleconvex}), for any fixed $\tau > 0$ (since $\| \cdot \|_\Ak$ is a norm).
Moreover, using explicit interpolation formulas for polynomials on $[m]$, it is easy to show that every polynomial in this feasible region has a representation with bounded coefficients (the analogue of Theorem \ref{thm:volupper}), and that the feasible region is robust to small perturbations in the coefficients (the analogue of Theorem \ref{thm:vollower}).
Thus, it suffices to give an efficient separation oracle for the feasible set.

\paragraph{The Separation Oracle}
Recall that the separation oracle in the continuous case consisted of two components: (i) a non-negativity checker (Subsection \ref{sec:nonneg}), 
and (ii) a fast $\Ak$-computation oracle (Subsection \ref{sec:akcomp}).
We still use the same approach for the discrete setting.

To check that a polynomial $p: I \to \R$ with bounded coefficients is non-negative on the points in $I$, we proceed as follows: 
we use \textsc{Fast-Approx-Roots} to find all the real roots of $p$ up to precision $1/4$, 
then evaluate $p$ on all the points in $I$ which have constant distance to any approximate root of $p$.
Since $p$ cannot change sign in an interval without roots, this is guaranteed to find a point in $I$ at which $p$ is negative, if one exists.
Moreover, since $p$ has at most $\pdeg$ roots, we evaluate $p$ at $O(\pdeg)$ points; 
using Fact \ref{fact:fastpolyeval}, this can be done in time $O(\pdeg \log \pdeg \log \log \pdeg)$.

Finally, to compute the $\Ak$-distance between $p = \sum_{j = 0}^\pdeg c_j x^j$ and $\fhat$ on an interval $I$, 
we use the same reduction as in Section \ref{sec:akcomp} with minor modifications.
The main difference is that between two points $x_i, x_{i + 1}$ in the support of the empirical distribution, 
the quantity $p[x_i, x_{i + 1}]$ (see section \ref{sec:akcomp}) is now defined to be
\begin{align*}
p[x_i, x_{i + 1}] &= \sum_{\ell = x_i + 1}^{x_{i + 1} - 1} p (\ell) \\
&=  \sum_{\ell = x_i + 1}^{x_{i + 1} - 1} \sum_{j = 0}^\pdeg c_j \ell^j \\
&= \sum_{j = 0}^\pdeg c_j \left( \sum_{\ell = x_{i} + 1}^{x_{i  +1} - 1} \ell^j \right) \; .
\end{align*}
Notice that the above is still a linear expression in the $c_j$, and 
there are simple closed-form expressions for $\left( \sum_{\ell = \alpha}^{\beta} \ell^j \right)$ for all integers $\alpha, \beta$ and for all $0 \leq j \leq \pdeg$.
Following the arguments in Section \ref{sec:akcomp} with this substituted quantity, one can show that the quantity returned by \textsc{ApproxSepOracle} in the discrete setting is still a separating hyperplane for $p$ and the current feasible set.
Moreover, \textsc{ApproxSepOracle} still runs in time $\Otilde(\pdeg)$.

\end{document}